\documentclass[acmsmall,dvipsnames]{acmart}

\setcopyright{rightsretained}
\acmJournal{PACMPL}
\usepackage{graphicx}
\usepackage{booktabs}
\usepackage{subcaption}
\usepackage{mathtools}
\usepackage{dsfont}
\usepackage{paralist}
\usepackage{caption}
\usepackage{amscd}
\usepackage{tikz}
\usetikzlibrary{calc,tikzmark,arrows.meta,automata,positioning,shadows,decorations.text,shapes.geometric,shapes.multipart,chains}
\usepackage{multirow}
\usepackage{dashbox}
\usepackage{amsfonts}
\usepackage{fontawesome}
\usepackage{xspace}
\usepackage{algorithm}
\usepackage{algorithmicx}
\usepackage{multicol}
\usepackage[noend]{algpseudocode}
\MakeRobust{\Call}
\algrenewcommand\algorithmicrequire{\textbf{Input:}}
\algrenewcommand\algorithmicensure{\textbf{Output:}}
\usepackage{amsmath} %
\usepackage[mathb]{mathabx}
\usepackage{hyperref}
\usepackage{listings}
\usepackage[utf8x]{inputenc}
\usepackage[T1]{fontenc}
\DeclareUnicodeCharacter{1EC5}{\~{\^e}}
\usepackage{stmaryrd} %
\usepackage{stackrel}
\usepackage{todonotes}
\usepackage{apxproof}

\usepackage{version}
\usepackage{cancel} %
\usepackage[capitalise]{cleveref}
\newcommand{\crefpart}[2]{%
 \hyperref[#2]{\namecref{#1}~\labelcref*{#1}~\ref*{#2}}%
}
\usepackage[inline,shortlabels]{enumitem}
\usepackage{framed}
\usepackage[commandnameprefix=ifneeded]{changes}
\usepackage{tabularx}

\usepackage{mathpartir}

\def\OPTIONConf{1}
\usepackage{jdunfield}
\usepackage{xifthen}

\usepackage{amsthm}
\theoremstyle{definition}

\newtheoremrep{theorem}{Theorem}[section]
\newtheoremrep{lemma}[theorem]{Lemma}
\newtheoremrep{definition}[theorem]{Definition}
\newtheoremrep{corollary}[theorem]{Corollary}
\newtheoremrep{remark}[theorem]{Remark}
\newtheoremrep{conjecture}[theorem]{Conjecture}
\newtheoremrep{example}[theorem]{\mbox{\emph{Example}}}

\definecolor{kwd}{HTML}{3F6FA8}
\definecolor{ndkwd}{HTML}{E69F00}
\colorlet{err}{Lavender}
\colorlet{warn}{Dandelion}

\usepackage{dafny}
\newcommand{\ocamllang}{\textsc{OCaml}\xspace}

\newcommand{\HATch}{\textrm{HATch}\xspace}
\newcommand{\LTLf}{$\textrm{LTL}_f$\xspace}

\newcommand{\kwd}[1]{{\color{kwd}\texttt{\bfseries #1}}}
\newcommand{\err}[1]{{\color{err}\texttt{\bfseries #1}}}
\newcommand{\ltl}[1]{{\color{ndkwd}\textsf{\bfseries #1}}\xspace}
\newcommand{\EM}[1]{\ensuremath{#1}\xspace}

\newcommand{\figref}[1]{Fig. \ref{#1}\xspace}

\newcommand{\myidx}[1]{%
  \IfInteger{#1}{#1}{\mathsf{#1}}
}
\newcommand{\myidxwrapper}[2]{%
  \ifthenelse{\isempty{#2}}
  {\EM{#1}}
  {\EM{#1_{\myidx{#2}}}}%
}

\newcommand{\Global}{\ltl{G}}
\newcommand{\Final}{\ltl{F}}
\newcommand{\Next}{\ltl{X}}
\renewcommand{\Until}{\ltl{U}}
\newcommand{\WeakUntil}{\ltl{W}}

\newcommand{\defeq}{\doteq}%

\newcommand{\denot}[2][]{\EM{\llbracket #2 \rrbracket_{#1}}}

\newcommand{\eff}[1]{{\textsf{#1}}\xspace}

\newcommand{\var}[1]{\EM{#1}}
\newcommand{\op}{\EM{\mathit{op}}\xspace}

\newcommand{\letin}[2]{\EM{\kwd{let}\ #1 = #2\ \kwd{in}}}

\newcommand{\?}[1][]{\EM{{?}_\mathtt{#1}}}
\newcommand{\abort}{\err{abort}\xspace}
\newcommand{\assume}[1]{\kwd{assume}~ #1}
\newcommand{\assert}[1]{\kwd{assert}~ #1}
\newcommand{\affirm}[2][]{\kwd{affirm}~ #2}
\newcommand{\admit}[2][]{\kwd{admit}~ #2}
\newcommand{\append}[1]{\kwd{append}~ #1}

\newcommand{\funbind}[1]{\kwd{fun}~#1.\,}
\newcommand{\fixbind}[1]{\kwd{fix}~#1.\,}
\newcommand{\inangle}[1]{\EM{\langle #1 \rangle}}
\newcommand{\inparen}[1]{\left(#1\right)}

\newcommand{\lneg}{{\EM{\backsim}}}
\newcommand{\lconj}{\EM{\sqcap}}
\newcommand{\ldisj}{\EM{\sqcup}}
\newcommand{\Lits}{\mathfrak{L}}

\newcommand{\lit}[1][]{%
  \ifthenelse{\isempty{#1}}
  {\EM{\ell}}
  {\EM{\ell_{\textsf{#1}}}}
}

\newcommand{\deriv}[3][]{%
  \EM{\mathsf{d}_{#2}\ifx r#1\inparen{#3}\else #3\fi}%
}
\newcommand{\Deriv}[3][]{%
  \EM{\mathsf{D}_{#2}\ifx r#1\inparen{#3}\else#3\fi}%
}

\newcommand{\hookrightarrowd}{\hookrightarrow_{\mathcal{D}}}
\newcommand{\trace}[1][]{\myidxwrapper{\pi}{#1}}
\newcommand{\Trace}[1][]{\myidxwrapper{\Pi}{#1}}
\let\oldPhi\Phi
\renewcommand{\Phi}[1][]{\myidxwrapper{\oldPhi}{#1}}
\newcommand{\true}{\EM{\top}}
\newcommand{\FA}[1][]{\myidxwrapper{\mathcal{A}}{#1}}

\newcommand{\TrConstr}{\textsf{constr}}
\newcommand{\FAnext}{\textsf{next}\xspace}

\newcommand{\sat}[1]{\textsf{isSat}\left(#1\right)}

\newcommand{\symb}[2][]{\EM{\hat{#2}_\mathtt{#1}}}

\newcommand{\domain}{\EM{\Sigma}}
\newcommand{\preds}{\EM{\Psi}}

\newcommand{\R}[1][]{\myidxwrapper{\mathcal{R}}{#1}}

\AtBeginDocument{
\addtolength{\abovedisplayskip}{-.5ex}
\addtolength{\abovedisplayshortskip}{-1ex}
\addtolength{\belowdisplayskip}{-.8ex}
\addtolength{\belowdisplayshortskip}{-1ex}
}

\begin{document}

\title{Derivative-Guided Symbolic Execution}

\author{Yongwei Yuan}
\orcid{0000-0002-2619-2288}
\affiliation{%
  \institution{Purdue University}
  \city{West Lafayette}
  \country{USA}
}
\email{yuan311@purdue.edu}

\author{Zhe Zhou}
\orcid{0000-0003-3900-7501}
\affiliation{%
  \institution{Purdue University}
  \city{West Lafayette}
  \country{USA}
}
\email{zhou956@purdue.edu}

\author{Julia Belyakova}
\orcid{0000-0002-7490-8500}
\affiliation{%
  \institution{Purdue University}
  \city{West Lafayette}
  \country{USA}
}
\email{julbinb@gmail.com}

\author{Suresh Jagannathan}
\orcid{0000-0001-6871-2424}
\affiliation{%
  \institution{Purdue University}
  \city{West Lafayette}
  \country{USA}
}
\email{suresh@cs.purdue.edu}

\begin{abstract}
We consider the formulation of a symbolic execution (SE) procedure for
functional programs that interact with effectful, opaque libraries.
Our procedure allows specifications of libraries and abstract data
type (ADT) methods that are expressed in \emph{Linear Temporal Logic
  over Finite Traces} (\LTLf), interpreting them as \emph{symbolic
  finite automata} (SFAs) to enable intelligent specification-guided
path exploration in this setting.  We apply our technique to
facilitate the falsification of complex data structure safety
properties in terms of effectful operations made by ADT methods on
underlying opaque representation type(s).  Specifications naturally
characterize admissible traces of temporally-ordered events that ADT
methods (and the library methods they depend upon) are allowed to
perform.  We show how to use these specifications to construct
feasible symbolic input states for the corresponding methods, as well
as how to encode safety properties in terms of this formalism.  More
importantly, we incorporate the notion of \emph{symbolic derivatives},
a mechanism that allows the SE procedure to intelligently
underapproximate the set of precondition states it needs to explore,
based on the automata structures latent in the provided
specifications and the safety property that is to be falsified.
Intuitively, derivatives enable symbolic execution to exploit temporal
constraints defined by trace-based specifications to quickly prune
unproductive paths and discover feasible error states.  Experimental
results on a wide-range of challenging ADT implementations demonstrate
the effectiveness of our approach.

\end{abstract}

\begin{CCSXML}
  <ccs2012>
   <concept>
       <concept_id>10003752.10003766.10003776</concept_id>
       <concept_desc>Theory of computation~Regular languages</concept_desc>
       <concept_significance>500</concept_significance>
       </concept>
   <concept>
       <concept_id>10003752.10003790.10003794</concept_id>
       <concept_desc>Theory of computation~Automated reasoning</concept_desc>
       <concept_significance>500</concept_significance>
       </concept>
   <concept>
       <concept_id>10003752.10003790.10003793</concept_id>
       <concept_desc>Theory of computation~Modal and temporal logics</concept_desc>
       <concept_significance>300</concept_significance>
       </concept>
   <concept>
       <concept_id>10011007.10010940.10010992.10010998.10011000</concept_id>
       <concept_desc>Software and its engineering~Automated static analysis</concept_desc>
       <concept_significance>500</concept_significance>
       </concept>
 </ccs2012>
\end{CCSXML}

\ccsdesc[500]{Theory of computation~Regular languages}
\ccsdesc[500]{Theory of computation~Automated reasoning}
\ccsdesc[300]{Theory of computation~Modal and temporal logics}
\ccsdesc[500]{Software and its engineering~Automated static analysis}

\keywords{symbolic execution, regular expression derivatives}  %

\maketitle
\section{Introduction}\label{sec:intro}

Symbolic
execution~\cite{baldoniSurveySymbolicExecution2018,cadarSymbolicExecutionSoftware2013}
(SE) is a well-studied program analysis technique whose goal is to
statically explore a bounded set of (symbolic) program executions in
search of one that yields a symbolic state inconsistent with a given
safety property.  The states generated during the course of these
executions consist of a set of path constraints; a violation is
identified if the conjunction of these constraints with the negation
of the safety property is logically satisfiable.  By knowing the
prestate under which a method may be invoked, SE can be performed on
individual methods in a \emph{compositional} fashion.  Oftentimes,
however, the program being analyzed interacts with libraries whose
implementations are unavailable for analysis.  In this case, we can
augment the SE procedure to interpret
models~\cite{chipounovS2EPlatformVivo2011} or
specifications~\cite{tobin-hochstadtHigherorderSymbolicExecution2012,xuStaticContractChecking2009a}
attached to library methods that describe the intended behavior of
their implementation in a form suitable for symbolic reasoning.

In this paper, we consider the design of an SE procedure for
functional programs that interface with effectful, opaque libraries.
Since we cannot express the behavior of library methods directly in
terms of how they manipulate their hidden state (since their
implementations are opaque), we instead reason about the interaction
of clients with these methods in terms of \emph{traces}, sequences of
method invocations and return values that constrain the shape of
allowed symbolic states that the symbolic interpreter needs to
consider.  Our primary contribution is a formalization of symbolic
execution in this setting that directly leverages the temporal
ordering constraints latent in these traces to intelligently guide
path exploration.

A particular useful setting in which this style of symbolic reasoning
is likely to be effective are abstract data type (ADT) implementations
whose specifications and safety properties are often couched in terms
of temporal modalities that constrain how datatype instances can be
constructed and used.  For example, to establish that an
implementation of a functional \texttt{Set} datatype, implemented
using an effectful list representation, correctly respects the
semantics of a mathematical set (e.g., $|S\ \cup \{x\}| = |S|$ if $x
\in S$) necessitates showing that any element added to its list
representation is different from any \emph{previously} added element.
Because the list implementation is potentially effectful, but does not
expose the state it manages to its clients, we can only reason about
the \texttt{Set} ADT methods that use it behaviorally, in terms of how
inputs to the list type's setters affect the values returned by its
getters that are subsequently consumed.

In our running example, the representation type \texttt{List}, defined
as a library, may provide a number of operations on a list instance,
some of which are pure like \texttt{mem} that checks for list
membership, and others of which are effectful, such as
\texttt{append!} that destructively appends its argument to its
instance.  The \texttt{Set} ADT might provide methods like \texttt{in}
that simply uses the \texttt{mem} method from \texttt{List} to check
if an element is included in a set instance, or \texttt{insert} that
adds a new element to the set using \texttt{append!}.  Suppose
\texttt{insert}'s implementation incorrectly adds a new element by
simply invoking \texttt{append!}, without first checking if the
element is already present.  Constructing a set using this
implementation would violate our desired safety property, namely that
every element in the set is unique.  Our goal is to use symbolic
execution to identify such errors.

Given the availability of specifications on ADT and representation-type
methods, symbolic execution of an ADT then involves:
\begin{enumerate*}
\item the generation of feasible (aka constructible)
    precondition states for an ADT method being analyzed in the form
    of \emph{symbolic traces} of method calls (and return values) on
  the representation type that is nonetheless consistent with the
  ADT's specification, and
\item devising an effective search procedure from this precondition
  state that identifies a feasible execution path, again expressed as
  a symbolic trace over symbolic invocations of methods on the
  representation type, whose final state violates the safety property.
\end{enumerate*}

In this work, we develop an SE procedure for a class of behavioral
specifications that can be concisely expressed in linear temporal
logic
(\LTLf~\cite{degiacomoLinearTemporalLogic2013,degiacomoSynthesisLTLLDL2015a,bansalModelCheckingStrategies2023});
these specifications correspond to symbolic finite automata
(SFA~\cite{veanesApplicationsSymbolicFinite2013b,dantoniMinimizationSymbolicAutomata2014,dantoniPowerSymbolicAutomata2017}),
in which automata transitions represent effectful and opaque
operations made by the ADT on its representation type(s).  Our SE
procedure exploits the latent SFA structure through \emph{symbolic
  derivatives}.  By computing the residual language after consuming a
prefix, Brzozowski derivatives simplify membership checking of regular
and context-free
languages~\cite{mightParsingDerivativesFunctional2011}.  In our
setting, symbolic derivatives compute the residual specification after
observing a sequence of ADT operations, enabling both the extraction
of admissible temporally ordered symbolic events (\ie, method
invocations and returns expressed in terms of symbolic variants of
program variables) of the ADT's representation type and the prediction
of future admissible events by progressively refining the space of
safe behaviors.

When equipped with such derivatives, our SE procedure is capable of
\begin{enumerate*}
\item generating precondition traces whose interpretation yields a
  prestate consistent with method specifications,
\item correlating pre- and post-invocation events with the safety
  property, and
\item guiding exploration along paths likely to lead to a
  falsification of the safety property.
\end{enumerate*}
By viewing the set of traces prior to and after method invocations
from the lens of the safety property we wish to falsify, our SE
procedure intelligently performs path exploration.  In the case of
\eff{Set} ADT, our SE procedure may, in the presence of a past
\eff{append!} event, actively look for another \eff{append!} of the
identical element, in an attempt to accelerate the falsification of
the unique-element property.
As a result, we oftentimes observe many orders-of-magnitude
improvement in path enumeration times, enabling it to scale favorably
with specification complexity.

In summary, this paper makes the following contributions:

\begin{enumerate}

\item We formalize an SE framework suitable for falsifying safety
  properties of effectful ADT implementations that manage hidden
  states.  Specifications are expressed as \LTLf formulae and capture
  temporal dependencies over a history of interactions between an ADT
  implementation and its underlying representation type(s).

\item We identify the latent SFA structures within these
  specifications and treat them as executable representations that
  enable the formalization of an SE procedure in terms of the traces
  characterized by these automata.
  
\item We propose to integrate a notion of symbolic derivatives as part
  of our SE procedure that intelligently underapproximates trace-based
  symbolic states and accelerates the search for a falsification
  witness.
  
\item We describe an implementation of these ideas in OCaml and
  show its effectiveness on a challenging set of data structure
  programs.
\end{enumerate}  

\noindent The remainder of the paper is organized as follows.
Motivation and informal explanation of our ideas is provided in the
next section. Section~\ref{sec:prelim} provides preliminaries and
details about \LTLf, SFAs, and derivatives.  The syntax of a core
language and a na\"{i}ve (derivative-free) semantics is given in
Section~\ref{sec:naive}.  The semantics of deriviative-based execution
is provided in Section~\ref{sec:derivative}.  We show how to translate
the declarative semantics of derivatives into an efficient algorithm
in Section~\ref{sec:algo}.  Implementation details and evaluation
results are provided in Section~\ref{sec:eval}.  Related work and
conclusions are given in Section~\ref{sec:related} and~\ref{sec:conc},
resp.

\section{Motivation}\label{sec:motiv}

\begin{figure}[t]
\begin{subfigure}{0.45\textwidth}
\begin{camlNoLines}[basicstyle=\fontsize{8}{9}\ttfamily]
module type KVStore
  (K: Key) (V: Value) : T = 
sig
  (** (k: K.t) -> (v: V.t)
      ghost (v': V.t)
      context $\textsf{stored}(\eff{T},k,v')$
      effect $\inangle{\eff{get}\ \var{k}\ \var{v}}$
      ensure $v=v'$ *)
  val get : K.t -> V.t

  (** (k:K.t) -> (v:V.t) -> unit
      effect $\inangle{\eff{put}\ \var{k}\ \var{v}}$ *)
  val put : K.t -> V.t -> unit
end
module type Node = sig
  type t
  val null : t  ... 
end
module type Elem = sig
  type t ...
end
\end{camlNoLines}
\end{subfigure}
\begin{subfigure}{0.57\textwidth}
\begin{camlNoLines}[numbers=left,numbersep=1pt,basicstyle=\fontsize{8}{9}\ttfamily]
module Nxt = KVStore (Node) (Node)
module Val = KVStore (Node) (Elem)

(** (hd:Node.t) -> (v:Elem.t) -> Node.t
    ghost (a: Node.t), (b: Node.t)
    context $\textsf{stored}(\eff{Nxt}, a, b)$
    effect $\lneg\inangle{\eff{Nxt.put}\ \cancel{\var{a}}\ \var{b}}\WeakUntil\inangle{\eff{Nxt.put}\ \var{a}\ \cancel{\var{b}}}$ *)
let remove (hd: Node.t) (elem: Elem.t) =
  if hd = null then hd(@\label{line:return-hd}@)
  else if Val.get hd = elem then(@\label{line:get-value-of-hd}@)
    Nxt.get hd(@\label{line:ret-next-of-hd}@)
  else
    let rec loop prev =(@\label{line:loop}@)
       let curr = Nxt.get prev in(@\label{line:prevget}@)
      if curr = null then ()(@\label{line:curr-null}@)
      else if Val.get curr = elem then(@\label{line:get-value-of-curr}@)
        let next = Nxt.get curr in(@\label{line:before}@)
        (@@)`Nxt.put curr null;`(@\label{line:missing}@)
        Nxt.put prev next(@\label{line:falsify}@)
      else loop curr
    in loop hd; hd(@\label{line:exit}@)
\end{camlNoLines}
\end{subfigure}
\caption{An implementation of a node \texttt{remove} operation in a linked-list ADT using two key-value stores.}\label{fig:motive-ex}
\end{figure}

To motivate our ideas, consider the program shown in
\cref{fig:motive-ex}.  The function \inlinecaml{remove} is a method in
a linked-list ADT that uses two effectful key-value stores as its
representation type, one to maintain an ordering relation among nodes
in the list (named \eff{Nxt}), and the other to record the elements associated with these
nodes (named \eff{Val}).  The implementation of the store is opaque to the ADT.  Given a
node \inlinecaml{curr} in a linked-list instance containing argument
value \inlinecaml{v}, \inlinecaml{remove} removes \inlinecaml{curr}
from the list by first initializing its successor field to
\inlinecaml{null} (given as the shaded statement at Line \ref{line:missing}), and then
adjusting the link from its predecessor \inlinecaml{prev} to point to
its successor \inlinecaml{next}.  

\subsection{Specifications}

The specification associated with \inlinecaml{remove} is expressed as
\LTLf~\cite{degiacomoLinearTemporalLogic2013}\footnote{Linear temporal
  logic over finite sequences.} formulae given in the comment above
its definition.  Informally, we can think of such formulae as
characterizing a set of admissible \emph{traces}, event sequences
defined in terms of method invocations and results.  The specification
has several elements.  In the case of \inlinecaml{remove}, it
introduces ghost variables \var{a} and \var{b}; these variables
represent an arbitrary pair of nodes in the list, constrained by the
method's precondition (identified by the keyword \inlinecaml{context})
and postcondition (identified by the keyword \inlinecaml{effect}).
The precondition characterizes all traces that construct a linked-list
in terms of the underlying key-value store representation type,
identifying an arbitrary consecutive pair of nodes using the
introduced ghost variables $a$ and $b$; it uses the following
definition:
\[ \textsf{stored}(\eff{Store}, \var{k}, \var{v}) \defeq 
\Final(\inangle{\eff{Store.put}\ k\ v}\wedge
\Next\Global\lneg\inangle{\eff{Store.put}\ k\ \_}) \]
The
postcondition reflects the actions performed by the method: it
specifies that node \var{b} can be linked to a predecessor other than \var{a} 
(as denoted by $\cancel{a}$)
only after \var{a} is linked to a successor other than \var{b}
(as denoted by $\cancel{b}$).  Both the pre- and
post-condition use LTL modalities.  The precondition uses the
\emph{finally} modality ($\Final$) to represent the eventual
establishment of a link between \var{a} and \var{b} in a trace, and
next ($\Next$) and global ($\Global$) modalities to prevent subsequent
actions in the trace from modifying that link; similarly, the
postcondition uses the weak-until ($\WeakUntil$) modality to specify a
conditional action, namely that \var{b} can be linked to a predecessor
other than \var{a} \emph{only after} \var{a} is no longer \var{b}'s predecessor.

The specification for the key-value stores used by the linked-list ADT
is given in the left of \cref{fig:motive-ex}.  As before, we capture
the effectful behavior of these methods using \LTLf specifications.
The precondition for \inlinecaml{get} requires that it be invoked in a
state constructed from a sequence of actions that include a
\inlinecaml{put} operation which associates key \var{k} to value
\var{v'}; it leverages the definition of \textsf{stored} defined
above, except using the key-value store instance in which the
\inlinecaml{get} is performed.  The method's postcondition ensures
that this property holds upon return.  Additionally, the specification
establishes an equality constraint, using the \inlinecaml{ensure}
annotation, between the value returned (\var{v}) and the value
previously \inlinecaml{put} on key \var{k} (\var{v'}). Note that
specifications used in this way constrain the set of precondition
states that a symbolic execution engine should consider; in
particular, the specification ignores any state that does not contain
a binding for \var{k}.  The specification for \inlinecaml{put} imposes
no structure on the store that must hold before it can execute, and
only guarantees that the \inlinecaml{put} action is performed upon
return.

\paragraph{Trace Specification as Safety Property}
To reiterate, specifications written in this way characterize
admissible execution \emph{traces} whose effects determine the context
in which a function can execute as well as the behavior manifested by the
function upon return from a call, allowing us to reason about the
behavior of the ADT without having to expose implementation details
about its underlying representation type.  Together, a pair of such
pre- and post-condition traces captures a safety property against
which the function must be checked.  In the case of \texttt{remove},
an execution under the specified context (precondition trace) that
does not satisfy the post-condition trace serves as a witness of a
violation of the \emph{predecessor uniqueness} safety property.

\begin{figure}[t]
  \begin{subfigure}[b]{0.38\textwidth}
    \centering
    \footnotesize
    \begin{tikzpicture}[shorten >=1pt,node distance=3mm,auto]
      \node[state,initial] (q0) {$q_0$};
      \node[state,accepting] (q1) [below=of q0] {$q_1$};
      \path[->]
      (q0) 
      edge [out=330,in=30] node {$\inangle{\eff{Nxt.put}\ \var{a}\ \var{b}}$} (q1)
      edge [loop above] node {$\lneg\inangle{\eff{Nxt.put}\ \var{a}\ \var{b}}$} ()
      (q1)
      edge [out=150,in=210] node {$\inangle{\eff{Nxt.put}\ \var{a}\ \cancel{\var{b}}}$} (q0)
      edge [loop below] node {$\lneg\inangle{\eff{Nxt.put}\ \var{a}\ \cancel{\var{b}}}$} ()
      ;
    \end{tikzpicture}
    \caption{Admissible traces prior to \texttt{remove}.}\label{fig:remove-context}
  \end{subfigure}
  \begin{subfigure}[b]{.45\textwidth}
    \centering
    \footnotesize
    \begin{tikzpicture}[shorten >=1pt,node distance=2cm,auto]
      \node[state,initial,accepting] (q0) {$q_2$};
      \node[state,accepting] (q1) [right=of q0] {$q_3$};
      \node[state,draw=red] (q2) [below=1.5cm of q0] {$\emptyset$};
      \path[->]
      (q0) edge node {$\inangle{\eff{Nxt.put}\ \var{a}\ \cancel{\var{b}}}$} (q1)
      edge node [swap] {$\inangle{\eff{Nxt.put}\ \cancel{\var{a}}\ \var{b}}$} (q2)
      edge [in=300,out=330,loop,align=left] node {
        $\phantom{\ldisj} \inangle{\eff{Nxt.put}\ \var{a}\ \var{b}}$\\
        $\vee \inangle{\eff{Nxt.put}\ \cancel{\var{a}}\ \cancel{\var{b}}}$\\
        $\vee \lneg\inangle{\eff{Nxt.put}}$} ()
      (q1) edge [loop right] node {$\bullet$} ()
      (q2) edge [loop right] node {$\bullet$} ()
      ;
    \end{tikzpicture}
    \caption{New actions allowed from \texttt{remove}.}\label{fig:remove-effect}
  \end{subfigure}
  \caption{The SFA representation of \texttt{remove}'s trace-based specification.}
  \label{fig:remove-automata}
\end{figure}

\paragraph{SFA Representation of Trace Specifications}
The set of traces characterized by \LTLf specifications can be
naturally represented by (symbolic) finite
automata~\cite{veanesApplicationsSymbolicFinite2013b} (SFA) structures
whose labels are events representing ADT method invocations and their
return values, and whose transitions reflect control dependencies over
these actions, defined by modalities used in the specification.
\cref{fig:remove-automata} shows how the \LTLf specifications given in
\cref{fig:motive-ex} can be represented as SFAs.  The automaton in
\cref{fig:remove-context} captures the precondition for
\inlinecaml{remove}.  The start state $q_0$ admits traces which
contain an arbitrary number of \inlinecaml{get} or \inlinecaml{put}
operations, not involving \inlinecaml{put} operations with key \var{a}
or value \var{b}; it allows such traces to be augmented with
\inlinecaml{put} operations that store a binding of \var{a} to
\var{b}, thus establishing the required shape of lists to which
\inlinecaml{remove} can be applied.  The store can be subsequently
updated with the effects of other \inlinecaml{put} operations on key
\var{a} that bind the key to nodes other than \var{b}, leading to a
transition that exits the accepting state $q_1$.  Traces accepted by
the precondition automaton encapsulate program states that can be used
as the basis for a successful symbolic execution run of
\texttt{remove}.  The postcondition for \texttt{remove} can be
represented as the automaton shown in \cref{fig:remove-effect}.  Here,
the initial state of the postcondition $q_2$ presumes the
precondition, namely ghost nodes \var{a} and \var{b} such that \var{a}
is the predecessor of \var{b} in the list.  A safe implementation of
\eff{remove} is allowed to repeatedly (re)link \var{a} to \var{b}
($\inangle{\eff{Nxt.put}\ \var{a}\ \var{b}}$), link other nodes
besides \var{a} to \var{b}
({$\inangle{\eff{Nxt.put}\ \cancel{\var{a}}\ \var{b}}$}), or perform
\inlinecaml{get} operations ($\lneg\inangle{\eff{Nxt.put}}$).  An
event that links another node to \var{b} \emph{without} first removing
the link from \var{a} results in a violation of the safety property,
however, depicted by the error state with a red circle containing
$\emptyset$.  State $q_3$ represents another accepting state
corresponding to a linked-list in which node \var{a} no longer points
to \var{b}.  The traces admitted by these automata correspond to the
\emph{hidden states} constructible by method invocations to the
underlying \inlinecaml{Nxt} and \inlinecaml{Val} store instances.

\paragraph{Symbolic Derivatives and SFAs}
To reveal the latent SFA representations of trace specifications that
qualify over symbolic variables, we propose to compute a variant of
Brzozowski derivatives
\cite{brzozowskiDerivativesRegularExpressions1964}, dubbed a
\emph{symbolic derivative}.  Let's revisit the postcondition for
\texttt{remove} in connection with its SFA representations in
\cref{fig:remove-effect}: its \LTLf formula
$\lneg\inangle{\eff{Nxt.put}\ \cancel{\var{a}}\
  \var{b}}\WeakUntil\inangle{\eff{Nxt.put}\ \var{a}\
  \cancel{\var{b}}}$ admits
\begin{enumerate*}
\item action $\inangle{\eff{Nxt.put}\ \var{a}\ \var{b}}$,
  $\inangle{\eff{Nxt.put}\ \cancel{\var{a}}\ \cancel{\var{b}}}$, or
  $\lneg\inangle{\eff{Nxt.put}}$, followed by traces admissible by the
  formula itself ($q_2\to q_2$), or
\item action $\inangle{\eff{Nxt.put}\ a\ \cancel{b}}$ followed by any
  trace of actions ($q_2\to q_3$).
\end{enumerate*}
Additionally, it does not admit
$\inangle{\eff{Nxt.put}\ \cancel{a}\ b}$ regardless of the following
actions ($q_2\to\emptyset$).  The symbolic derivatives of the
postcondition over
$\inangle{\eff{Nxt.put}\ \var{a}\ \var{b}} \vee
\inangle{\eff{Nxt.put}\ \cancel{\var{a}}\ \cancel{\var{b}}} \vee
\lneg\inangle{\eff{Nxt.put}}$,
$\inangle{\eff{Nxt.put}\ \var{a}\ \cancel{\var{b}}}$, and
$\inangle{\eff{Nxt.put}\ \cancel{\var{a}}\ \var{b}}$ corresponds to
states $q_2$, $q_3$, and $\emptyset$ respectively.  Such derivatives
allow symbolic execution to, as we will discuss in
\cref{sec:motiv-deriv-exec}, ``execute'' trace specifications
following their latent SFA representations and make the \eff{put}
operation at \cref{line:falsify} a witness of the action
$\inangle{\eff{Nxt.put}\ \cancel{\var{a}}\ \var{b}}$ that leads to the
dead state $\emptyset$.

\subsection{Trace-Based Symbolic Execution}\label{sec:motiv-symb-exec}

\paragraph{Expressing Hidden States as Traces.}
While our specification language can express a rich set of behaviors
that can be exhibited by the ADT, it is not immediately obvious how to
incorporate such specifications as part of an efficient symbolic
execution procedure.  Yet, it is clear that \inlinecaml{remove}'s
specification naturally entails the uniqueness property that we wish
to check, albeit in terms of traces over the representation type's
operations, rather than directly in terms of the method's
implementation.

Conventionally, symbolic execution explores symbolic states along a
program's CFG to find a reachable path that ends at an erroneous
state; however, in our setting, the linked list maintained by the
key-value stores \eff{Nxt} and \eff{Val} does not have an explicit
state representation that can be trivially constructed from the
program; it is instead manifested by traces extracted from SFAs
associated with the ADT's specification.  We will show shortly in
\cref{sec:motiv-deriv-exec} how to precisely relate the trace
structure described by the specification with the execution paths
explored by symbolic execution to manifest these hidden states.
Establishing this relation via the use of symbolic derivatives, which
will also be described shortly, enables a novel form of
property-directed exploration that can be exploited by a symbolic
execution procedure to avoid searching over unproductive execution
paths.

A relatively straightforward approach is to encapsulate symbolic
states into a set of traces that records the temporally ordered events
produced along the current execution path being explored.  Such a set
of traces, like our specifications, has a natural representation in
SFA.  Take the precondition of \texttt{remove} as an example;
\textsf{stored}(\eff{Nxt}, \var{a}, \var{b}) encapsulates a symbolic
state $S_0$ as follows:

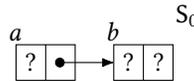
\begin{center}\begin{tikzpicture}[
  baseline=(current bounding box.center),
  list/.style={
    rectangle split,
    rectangle split parts=2,
    rectangle split horizontal,
    label={[anchor=south,inner sep=1pt]north west:#1},
    draw},
  dotarrow/.style={Circle-Latex[round]}]
  \node[list=\var{a}] (a) {?};
  \node[list=\var{b},right=5mm of a] (b) {?\nodepart{second}?};
  \node[above right=1mm and -1mm of b] {$\textsf{S}_0$};
  \draw[dotarrow] (a.two |- a.center) -- (b);
\end{tikzpicture}\end{center}

\noindent Since \var{a} and \var{b} are two symbolic variables
denoting two arbitrary nodes as long as the former is the predecessor
of the latter, the \emph{path condition} is $\top$ initially.  Before
symbolically executing \inlinecaml{remove}, we introduce two
additional symbolic variables $\var{n}_0$ and $\var{u}$ as the
arguments passed to \inlinecaml{remove}, respectively denoting the
first node in the input linked-list, and the element that the node to
be removed stores.  Since the precondition of \inlinecaml{remove}
places no constraint on these variables, the path condition is $\top$
initially.  Substituting $\var{n}_0$ for \texttt{hd} (and $\var{u}$
for \texttt{elem}) in the body of \texttt{remove}, symbolic execution
may choose to enter the first branch at \cref{line:return-hd}, augment
the path condition with $\var{n}_0=\texttt{null}$, and return
$\var{n}_0$.  An \emph{empty} symbolic trace (of length zero) is
produced from this execution and ghost nodes \var{a} and node \var{b}
are left untouched in the symbolic state.  Therefore, the execution
complies with the predecessor uniqueness safety property as specified
by the method's postcondition.

Conventionally, symbolic states are constrained by path
conditions whose satisfiability \emph{directly} manifests the
reachability of the respective state.
In this setting, however, events executed along the path may also
quantify over symbolic variables like path conditions and thus may
impose additional constraints on the associated symbolic state.  Take
$S_0$ as an example: the SFA representation of
\textsf{stored}(\eff{Nxt}, \var{a}, \var{b}), as shown in
\cref{fig:remove-context}, admits the traces of executed events
allowed upon entering \eff{remove}.  Transitions labeled by
\inangle{\eff{Nxt.put}~a~\cancel{b}} admit event
$\eff{Nxt.put}~key~val$ only if $key=a \wedge val\neq b$ holds.  As a
result, path condition $\top$ and the SFA \emph{synergistically}
encapsulate the symbolic state $S_0$, reflecting the conditions
necessary for the execution paths that lead to \eff{remove} to be
feasible.

\paragraph{Refinement of Trace-Based Symbolic States.}
We further illustrate this synergy by considering other feasible
execution paths in \texttt{remove} along which symbolic states are
refined.  A symbolic execution procedure may also assume that the
linked list is not empty, \ie, $n_0\neq\texttt{null}$, and can thus
take the branch at \cref{line:get-value-of-hd} where a symbolic method
invocation \eff{Val.get} takes place with $n_0$ as its argument.  Just
as path conditions are refined by branching conditions, SFAs are
refined by the effectful operations performed along execution paths,
consistent with the constraints given by the preconditions of executed
operations.  Here, the precondition of \eff{Val.get}, when applied to
$n_0$, is $\textsf{stored}(\eff{Val},n_0,u_0)$, where $u_0$ is a fresh
symbolic variable representing the result of the invocation.  That is,
$u_0$ is the element stored at $n_0$.  Noticing that $n_0$ is not even
mentioned in $S_0$, $n_0$ may be $a$, $b$, or some other node not
explicitly specified in $S_0$.  The refined SFA representation, as
denoted by
$\textsf{stored}(\eff{Nxt},a,b)\wedge\textsf{stored}(\eff{Val},n_0,u_0)$,
encapsulates the prestate of the method invocation
$\inangle{\var{u}_0{\gets}\eff{Val.get}\ \var{n}_0}$, which can be
cleanly dissected into the following three states:
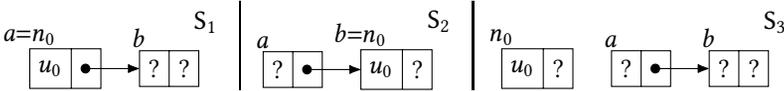
\begin{center}
\begin{tabular}{c|c|c}
\begin{tikzpicture}[
  baseline=(current bounding box.center),
  list/.style={
    rectangle split,
    rectangle split parts=2,
    rectangle split horizontal,
    label={[anchor=south,inner sep=1pt]north west:#1},
    draw},
  dotarrow/.style={Circle-Latex[round]}]
  \node[list={$\var{a}{=}\var{n}_0$}] (a) {$\var{u}_0$\nodepart{second}\phantom{?}};
  \node[list=\var{b},right=5mm of a] (b) {?\nodepart{second}?};
  \node[above right=1mm and -2mm of b] {$\textsf{S}_1$};
  \draw[dotarrow] (a.two |- a.center) -- (b);
\end{tikzpicture}
&
\begin{tikzpicture}[
  baseline=(current bounding box.center),
  list/.style={
    rectangle split,
    rectangle split parts=2,
    rectangle split horizontal,
    label={[anchor=south,inner sep=1pt]north west:#1},
    draw},
  dotarrow/.style={Circle-Latex[round]}]
  \node[list=\var{a}] (a) {?};
  \node[list={$\var{b}{=}\var{n}_{0}$},right=5mm of a] (b) {$\var{u}_0$\nodepart{second}?};
  \node[above right=1mm and -2mm of b] {$\textsf{S}_2$};
  \draw[dotarrow] (a.two |- a.center) -- (b);
\end{tikzpicture}
&
\begin{tikzpicture}[
  baseline=(current bounding box.center),
  list/.style={
    rectangle split,
    rectangle split parts=2,
    rectangle split horizontal,
    label={[anchor=south,inner sep=1pt]north west:#1},
    draw},
  dotarrow/.style={Circle-Latex[round]}]
  \node[list=$\var{n}_0$,right=5mm of b] (n0) {$\var{u}_0$\nodepart{second}?};
  \node[list=\var{a},right=5mm of n0] (a) {?};
  \node[list=\var{b},right=5mm of a] (b) {?\nodepart{second}?};
  \node[above right=1mm and -2mm of b] {$\textsf{S}_3$};
  \draw[dotarrow] (a.two |- a.center) -- (b);
\end{tikzpicture}
\end{tabular}
\end{center}
\noindent Augmenting the path condition with $\var{u}_0=\var{u}$ as we
enter \cref{line:ret-next-of-hd}, symbolically executing the
invocation of \eff{Nxt.get} with $\var{n}_0$ first introduces a new
symbolic variable $\var{n}_1$ denoting the successor of $\var{n}_0$,
and then further refines the traces of executed events to be
$((\textsf{stored}(\eff{Nxt},a,b)\wedge\textsf{stored}(\eff{Val},n_0,u_0))\cdot\inangle{\var{u}_0{\gets}\eff{Val.get}\
  \var{n}_0})\wedge\textsf{stored}(\eff{Nxt},n_0,n_1)$, where
$\cdot\inangle{\var{u}_0{\gets}\eff{Val.get}\ \var{n}_0}$ records the
previous method invocation.  Within the refined state, similarly, node
$\var{n}_1$ could be \var{a}, \var{b}, some other node, or even
$\var{n}_0$.  Three ordinary cases are depicted below:
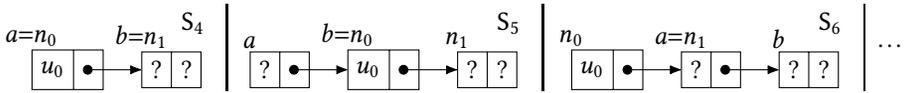
\begin{center}
\begin{tabular}{c|c|c|c}
\begin{tikzpicture}[
  baseline=(current bounding box.center),
  list/.style={
    rectangle split,
    rectangle split parts=2,
    rectangle split horizontal,
    label={[anchor=south,inner sep=1pt]north west:#1},
    draw},
  dotarrow/.style={Circle-Latex[round]}]
  \node[list={$\var{a}{=}\var{n}_0$}] (a) {$\var{u}_0$\nodepart{second}\phantom{?}};
  \node[list={$\var{b}{=}\var{n}_1$},right=5mm of a] (b) {?\nodepart{second}?};
  \node[above right=1mm and -4mm of b] {$\textsf{S}_4$};
  \draw[dotarrow] (a.two |- a.center) -- (b);
\end{tikzpicture}
&
\begin{tikzpicture}[
  baseline=(current bounding box.center),
  list/.style={
    rectangle split,
    rectangle split parts=2,
    rectangle split horizontal,
    label={[anchor=south,inner sep=1pt]north west:#1},
    draw},
  dotarrow/.style={Circle-Latex[round]}]
  \node[list=\var{a}] (a) {?};
  \node[list={$\var{b}{=}\var{n}_0$},right=5mm of a] (b) {$\var{u}_0$\nodepart{second}\phantom{?}};
  \node[list=$\var{n}_1$,right=5mm of b] (c) {?\nodepart{second}?};
  \node[above right=1mm and -4mm of c] {$\textsf{S}_5$};
  \draw[dotarrow] (a.two |- a.center) -- (b);
  \draw[dotarrow] (b.two |- b.center) -- (c);
\end{tikzpicture}
&
\begin{tikzpicture}[
  baseline=(current bounding box.center),
  list/.style={
    rectangle split,
    rectangle split parts=2,
    rectangle split horizontal,
    label={[anchor=south,inner sep=1pt]north west:#1},
    draw},
  dotarrow/.style={Circle-Latex[round]}]
  \node[list=$\var{n}_0$] (n0) {$\var{u}_0$\nodepart{second}\phantom{?}};
  \node[list={$\var{a}{=}\var{n}_1$},right=5mm of n0] (a) {?};
  \node[list=\var{b},right=5mm of a] (b) {?\nodepart{second}?};
  \node[above right=1mm and -4mm of b] {$\textsf{S}_6$};
  \draw[dotarrow] (a.two |- a.center) -- (b);
  \draw[dotarrow] (n0.two |- n0.center) -- (a);
\end{tikzpicture}
&\dots
\end{tabular}
\end{center}
where $S_4$, $S_5$ and $S_6$ refines the above dissected states $S_1$,
$S_2$, and $S_3$ respectively.  Then, node $\var{n}_1$ is returned at
(Line~\ref{line:ret-next-of-hd}).  Although the symbolic state is
refined, node \var{a} and node \var{b} are still left intact to comply
with the constraints defined by the method's specification.  Since the
safety property holds over all these possible symbolic states, this
execution path is also deemed to satisfy the postcondition.

The violation of the postcondition may happen within the loop
(\cref{line:loop}) if we do not execute the shaded operation at
(\cref{line:missing}).  Substituting $\var{n}_0$ for \texttt{prev} in
the loop body, the invocation of \eff{Nxt.get} with $\var{n}_0$ then
takes us to the same symbolic states, $S_4$, $S_5$, and $S_6$,
generated in the earlier explored branch.  If we continue the
execution from $S_6$ up until \cref{line:before}, one possible
symbolic state that holds after the invocation of \eff{Val.get} and
\eff{Nxt.get} with $\var{n}_1$ would be:
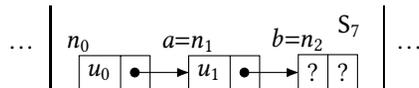
\begin{center}
\begin{tabular}{c|c|c}
\dots&
\begin{tikzpicture}[
  baseline=(current bounding box.center),
  list/.style={
    rectangle split,
    rectangle split parts=2,
    rectangle split horizontal,
    label={[anchor=south,inner sep=1pt]north west:#1},
    draw},
  dotarrow/.style={Circle-Latex[round]}]
  \node[list=$\var{n}_0$] (a) {$\var{u}_0$};
  \node[list=$\var{a}{=}\var{n}_1$,right=5mm of a] (b) {$\var{u}_1$};
  \node[list=$\var{b}{=}\var{n}_2$,right=5mm of b] (c) {?\nodepart{second}?};
  \node[above right=1mm and -4mm of c] {$\textsf{S}_7$};
  \draw[dotarrow] (a.two |- a.center) -- (b);
  \draw[dotarrow] (b.two |- b.center) -- (c);
\end{tikzpicture}
&\dots
\end{tabular}
\end{center}
where $\var{u}_1$ and $\var{n}_2$ are fresh symbolic variables that
denotes the element stored in $n_1$ and the successor of $n_1$
respectively.  Without executing \cref{line:missing}, the invocation
of \eff{Nxt.put} at \cref{line:falsify} makes node \var{b} (\ie,
$\var{n}_2$) the successor of both $\var{n_0}$ and \var{a} (\ie,
$\var{n}_1$):
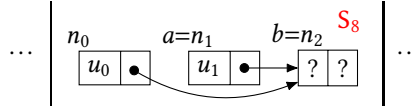
\begin{center}
\begin{tabular}{c|c|c}
\dots&
\begin{tikzpicture}[
  baseline=(current bounding box.center),
  list/.style={
    rectangle split,
    rectangle split parts=2,
    rectangle split horizontal,
    label={[anchor=south,inner sep=1pt]north west:#1},
    draw},
  dotarrow/.style={Circle-Latex[round]}]
  \node[list=$\var{n}_0$] (a) {$\var{u}_0$};
  \node[list=$\var{a}{=}\var{n}_1$,right=5mm of a] (b) {$\var{u}_1$};
  \node[list=$\var{b}{=}\var{n}_2$,right=5mm of b] (c) {?\nodepart{second}?};
  \node[above right=1mm and -4mm of c] {$\color{red}\textsf{S}_8$};
  \draw[dotarrow,bend right=25] (a.two |- a.center) to (c);
  \draw[dotarrow] (b.two |- b.center) -- (c);
\end{tikzpicture}
&\dots
\end{tabular}
\end{center}
This symbolic state ($S_8$), among other possible states that we omit
here, happens to manifest a violation of the method's safety property
because $S_8$ with path condition
$n_0\neq\texttt{null} \wedge u_0\neq u \wedge n_1\neq\texttt{null}
\wedge u_1=u$ is obviously reachable and it is obvious from the above
illustration that any trace of events encapsulated in $S_8$ is
rejected by the method's postcondition.

\subsection{Symbolic Execution with Symbolic Derivatives}\label{sec:motiv-deriv-exec}

Efficiency is a serious problem that must be considered by any
symbolic execution procedure.  Conventional techniques can prune
infeasible paths~\cite{baldoniSurveySymbolicExecution2018} by
leveraging the control structure in programs along with provided
preconditions to consider an underapproximation of program behavior
that follows a single execution path at a time.  For example, if a
precondition requires the input list to be not empty, path exploration
can ignore paths that contradict this constraint (e.g.,
\cref{line:return-hd} in \cref{fig:motive-ex}).  A trace-aware
symbolic execution procedure can additionally discover the shape of
linked-lists automatically from the SFA structures latent in
specifications that induce these traces, and thus can introduce new
opportunities for pruning unproductive paths.  For example, as
currently described, although the precondition of
$\inangle{\var{u}_0{\gets}\eff{Val.get}\ \var{n}_0}$ refines $S_0$,
the refined symbolic state does not specify whether $n_0$ is equal to
$a$, $b$, or some other node.
More notably, when the symbolic method invocation
$\inangle{\var{n}_1{\gets}\eff{Nxt.get}\ \var{n}_0}$ is made, the SFA
encapsulating the symbolic state after the invocation includes
contradicting paths not depicted among $S_4$, $S_5$, and $S_6$, in
which $n_0$ can be both equal to and not equal to $a$.  Explicitly
leveraging the control structure latent in SFAs would enable us to
correlate traces (and the hidden states they induce) to specific
program paths, exposing new pruning opportunities that would otherwise
not be possible.  We propose to compute \emph{symbolic derivatives}
over the specifications to explore and exploit the latent SFA
structures within the trace-based specifications.

\paragraph{Derivative-Guided Path Exploration}
Our symbolic execution procedure employs symbolic derivatives to
explore the SFA structures latent in the specifications and
intelligently enumerate admissible traces that encapsulate prestates
along execution paths, lowering the cost of path feasibility check and
enabling effective path pruning.  In our running example, the initial
symbolic state $S_0$ of \eff{remove} requires that some node \var{a}
is the predecessor of some node \var{b} but does not specify when
\var{b} is made the successor of \var{a} via \eff{Nxt.put}.  Recall
the precondition automaton (\cref{fig:remove-automata}), which denotes
the traces encapsulating $S_0$: the relevant \eff{Nxt.put} operation
may be performed ($q_0\to q_1$) after some indefinite number of
irrelevant actions ($q_0\to q_0$), or after a previously executed
\inangle{\eff{Nxt.put}\ \var{a}\ \var{b}} is invalidated
($q_1\to q_0$).  A symbolic derivative computation helps explore this
automaton structure by sampling paths from the start state $q_0$ to
the accepting state $q_1$.  In the case of \texttt{remove}, its
behavior happens to be invariant to when the call
\inangle{\eff{Nxt.put}\ \var{a}\ \var{b}} is performed.  Therefore, in
order to reach the erroneous symbolic state $S_8$, it is sufficient to
begin the symbolic execution of \texttt{remove} with a precondition
trace that consists of \inangle{\eff{Nxt.put}\ \var{a}\ \var{b}}
followed by actions that do not invalidate this operation.

Our symbolic execution procedure further exploits the latent SFA
structures to make informed decisions in choosing the precondition
trace that favors the efficient exploration of feasible execution
paths.  To minimize the complexity of reasoning about the behavior of
method invocations performed, one straightforward strategy is to
choose the ``simplest'' symbolic state based on the length of the
corresponding trace induced from the precondition automaton.  For
example, \texttt{remove} may be invoked under a state encapsulated by
the singleton trace \inangle{\eff{Nxt.put}\ \var{a}\ \var{b}}.  This
trace, however, cannot be refined to admit a \eff{Val.put} event,
which is required by the invocation of \eff{Val.get} at
\cref{line:get-value-of-hd} in \cref{fig:motive-ex}, thus failing to
reach the error state $S_8$.  As we proceed to consider longer
precondition traces, the simplicity criteria soon becomes insufficient
to distinguish between precondition traces.  Here are two traces of
four symbolic events that can be induced from the precondition
automaton and thus equally encapsulate a valid prestate of
\texttt{remove}:
\[
  \inangle{\eff{Nxt.put}\ \var{a}\ \var{b}}
  (\lneg\inangle{\eff{Nxt.put}\ \var{a}\ \cancel{\var{b}}})^3
  \quad \text{(repeated 3 times)}
\]
\[
  \inangle{\eff{Nxt.put}\ \var{a}\ \var{b}}
  \inangle{\eff{Nxt.put}\ \var{a}\ \cancel{\var{b}}}
  \inangle{\eff{Nxt.put}\ \var{a}\ \var{b}}
  (\lneg\inangle{\eff{Nxt.put}\ \var{a}\ \cancel{\var{b}}})
\]
We know from before that the erroneous execution path leading to $S_8$
requires at least two \eff{Val.put} events but the later trace can
only admit one \eff{Val.put} event.  In this case, symbolic
derivatives guide SE to first consider the former trace when executing
\texttt{remove}.
This behavior arises from symbolic derivatives' tendency to maximize
``progress'' when inducing traces from the precondition automaton.  As
symbolic derivatives facilitate the exploration within the latent SFA
structures of precondition automata, this tendency manifests in
several ways:
\begin{enumerate*}
\item consistently select the states closer to the accepting state;
\item avoid unnecessarily transiting back to non-accepting states, and;
\item steer clear of generating the stagnation pattern of "setting",
  "unsetting", and "resetting".
\end{enumerate*}
In the case of the precondition automaton shown in
\cref{fig:remove-context}, the execution tends to move from $q_0$ to
$q_1$ via \inangle{\eff{Nxt.put}\ \var{a}\ \var{b}} and stays at
$q_1$.  Intuitively, the former trace induced in the described fashion
encapsulates a relatively more permissive state that potentially leads
to more interesting feasible paths being explored, including the one
that leads to the error state $S_8$.
\newcommand{\lref}[1]{\texttt{L}\ref{#1}:\xspace}
\begin{figure}
\tikzstyle{decision} = [diamond, draw, text centered, aspect=3, text width=2.5cm, inner sep=.5pt]
\tikzstyle{process} = [rectangle split, rectangle split horizontal, rectangle split part fill={gray!50,gray!50,white}, rectangle split parts=3, draw, text centered, rounded corners, minimum height=2em]
\begin{tikzpicture}[node distance=.3cm and 1.5cm,>=Latex]
  \node [decision] (n0null) {\lref{line:return-hd}$n_0 = \texttt{null}$};
  \node [process, below left=of n0null] (ret9) {$S_0$\nodepart{two}$q_2$\nodepart{three}\lref{line:return-hd}$\texttt{return}~n_0$};
  \draw [->] (n0null.south) -- (ret9) node [above,midway] {\textrm{Y}};
  \node [process, below right=of n0null] (getValOfn0) {$S_0$\nodepart{two}$q_2$\nodepart{three}\lref{line:get-value-of-hd}$u_0\gets\eff{Val.get}~n_0$};
  \draw [->] (n0null.south) -- (getValOfn0) node [above,midway] {\textrm{N}};
  \node [decision, below left=of getValOfn0] (u0test) {\lref{line:get-value-of-hd}$u_0 = u$};
  \draw [->] (getValOfn0) -- (u0test.north);
  \node [process, align=left, below left=of u0test] (ret11) {$S_3$\\$S_6$\nodepart{two}$q_2$\\$q_2$\nodepart{three}\lref{line:ret-next-of-hd}$n_1\gets\eff{Nxt.get}~n_0$\\\lref{line:ret-next-of-hd}$\texttt{return}~n_1$};
  \draw [->] (u0test.south) -- (ret11) node [above,midway] {\textrm{Y}};
  \node [process, below right=of u0test] (getNxtOfn0) {$S_3$\nodepart{two}$q_2$\nodepart{three}\lref{line:prevget}$n_1\gets\eff{Nxt.get}~n_0$};
  \draw [->] (u0test.south) -- (getNxtOfn0) node [above,midway] {\textrm{N}};
  \node [decision, below left=of getNxtOfn0] (n1null) {\lref{line:curr-null}$n_1 = \texttt{null}$};
  \draw [->] (getNxtOfn0) -- (n1null.north);
  \node [process, below left=of n1null] (ret15) {$S_6$\nodepart{two}$q_2$\nodepart{three}\lref{line:exit}$\texttt{return}~n_0$};
  \draw [->] (n1null.south) -- (ret15) node [above,midway] {\textrm{Y}};
  \node [process, below right=of n1null] (getValOfn1) {$S_6$\nodepart{two}$q_2$\nodepart{three}\lref{line:get-value-of-curr}$u_1\gets\eff{Val.get}~n_1$};
  \draw [->] (n1null.south) -- (getValOfn1) node [above,midway] {\textrm{N}};
  \node [decision, below left=of getValOfn1] (u1test) {\lref{line:get-value-of-curr}$u_1=u$};
  \draw [->] (getValOfn1) -- (u1test.north);
  \node [process, below left=of u1test, align=left] (getNxtOfn1) {$\cdot\cdot$\\$S_7$\\$S_8$\nodepart{two}$q_2$\\$q_2$\\$\color{red}\emptyset$\nodepart{three}\lref{line:before}$n_2\gets\eff{Nxt.get}~n_1$\\\lref{line:falsify}$\eff{Nxt.put}~n_0~n_2$\\\lref{line:exit}$\texttt{return}~n_0$};
  \draw [->] (u1test.south) -- (getNxtOfn1) node [above,midway] {\textrm{Y}};
  \node [rectangle, draw, rounded corners, below right=of u1test,minimum width=1cm,minimum height=.7cm] (loop) {\dots};
  \draw [->] (u1test.south) -- (loop) node [above,midway] {\textrm{N}};
\end{tikzpicture}
\caption{Derivative-guided symbolic execution of \texttt{remove}.}\label{fig:remove-exec}
\end{figure}
\paragraph{Derivative-Guided Falsification}
In addition to intelligently enumerating precondition traces, symbolic
derivatives can guide the symbolic execution of \inlinecaml{remove}
itself, by again exploiting the latent SFA structure of the
specification.  Specifically, they allow our symbolic execution to
relate method invocations in \inlinecaml{remove} to the transitions in
the postcondition automaton (\cref{fig:remove-effect}).  Consider the
symbolic execution tree of \texttt{remove} in terms of operations over
symbolic variables, as depicted in \cref{fig:remove-exec}.  Each
program point is associated with a state in the postcondition
automaton that effectively determines the set of future traces (\ie,
sequences of future actions) that are (un)safe to explore.  The
symbolic execution of \texttt{remove} begins with the initial state
$q_2$ of the postcondition automaton because it admits all traces of
actions that \texttt{remove} is safe to perform.  As \texttt{remove}
traverses the input linked list via repetitive \eff{get} invocations
before unsafely invoking \eff{Nxt.put}, symbolic derivatives
intelligently determine that the future safe traces after those
\eff{get} invocations are also represented by $q_2$.  This is because
$q_2 \to q_2$ is the only outgoing transition from $q_2$ that is
compatible with \eff{get} actions.
Recall the frontier state $S_7$ before the unsafe invocation of
\eff{Nxt.put} from before: the past traces of actions already
determine $\var{n}_0\neq\var{a}$ and $\var{n}_2=\var{b}$.  Then our
SE procedure can indeed determine that \inangle{\eff{Nxt.put}\
  \var{n}_0\ \var{n}_2} is unsafe because it is compatible with the
transition $q_2\to\emptyset$, which happens to be the only
compatible outgoing transition from $q_2$.
Note that even if \texttt{remove} continues traversing the linked list
after removing the first found element via a recursive call to
\texttt{loop},
the symbolic derivatives guide SE to avoid unprofitably unrolling the
\texttt{loop} because any future action is unsafe.  In contrast, the
na\"ive trace-based SE described in \cref{sec:motiv-symb-exec} would
wastefully relate each explored execution path of \texttt{remove} with
traces in the postcondition automaton.

To conclude this section, the main contribution of this paper is a new
symbolic execution procedure that computes such symbolic derivatives
as symbolic execution proceeds to maintain a trace of symbolic events
that witness the current execution path and its relationship with the
safety property, in an attempt to accelerate the search for a feasible
execution that violates the property.

\section{Preliminaries}\label{sec:prelim}

The SFA representations of specifications expressed in \LTLf
\cite{degiacomoLinearTemporalLogic2013} facilitate the discussion in
\cref{sec:motiv}.  To properly formulate the relationship between
traces admissible by \LTLf formulae and SFAs (see
\cref{sec:derivative}), however, we need to introduce regular
expressions.  The language of regular expressions (RE) is strictly
more expressive for representing traces than \LTLf and enjoys an
important closure property under the classic derivative computation
\cite{brzozowskiDerivativesRegularExpressions1964}.  In this section,
we present symbolic regular expressions (SREs) whose atoms are
predicates, show how to express common temporal modalities from \LTLf
in terms of RE operations, and relate classic derivatives with states
in finite state automata.
  
\paragraph{Effective Boolean Algebras} Tuple
$(\domain, \preds, \denot{\_}, \bot, \top, \lor, \land, \lnot)$
defines \emph{Effective Boolean Algebra} (EBA
\cite{veanesApplicationsSymbolicFinite2013b}) where \domain is a set
of domain elements and \preds is a set of \emph{predicates}, closed
under the Boolean connectives with $\bot,\top\in\preds$.  The
denotation of $\phi,\psi\in\preds$ are provided by
$\denot{\_}{:}\preds\rightarrow 2^\domain$ where
\begin{mathpar}
  \denot{\bot}=\{\} \and \denot{\top}=\domain \and
  \denot{\phi\lor\psi}=\denot{\phi}\cup\denot{\psi} \and
  \denot{\phi\land\psi}=\denot{\phi}\cap\denot{\psi} \and
  \denot{\lnot\phi}=\domain/\denot{\phi}
\end{mathpar}

\paragraph{Traces} Finite sequences of elements $\alpha,\beta$ from
domain $\Sigma$ are called \emph{traces} \trace.  Let $\epsilon$ be
the empty trace and $\trace_1\cdot\trace_2$ be the associative
concatenation of $\trace_1$ and $\trace_2$.  We write
$\trace_1\trace_2$ for $\trace_1\cdot\trace_2$ when it is clear
from the context that juxtaposition stands for concatenation.
Following the convention, we further denote that
$\Sigma^{(0)}=\{\epsilon\}$, $\Sigma^{(k+1)}=\Sigma\cdot\Sigma^{(k)}$,
for $k\geq0$, and $\Sigma^*=\bigcup_{k\geq0}\Sigma^{(k)}$, where
$L_1 \cdot L_2 = \{\trace_1\trace_2\mid \trace_1\in L_1,
\trace_2\in L_2\}$ for $L_1\subseteq\Sigma^*$ and
$L_2\subseteq\Sigma^*$.  Lastly, we write $L^*$ for the closure of $L$
under concatenation when it is clear from the context that
$L\subseteq\Sigma^*$.

\paragraph{Symbolic Regular Expressions}
We define \emph{Symbolic Regular Expressions} (SRE) modulo Boolean
Algebra
$(\Sigma, \preds, \denot{\_}, \bot, \bullet, \ldisj, \lconj, \lneg)$
such that SREs use \emph{literals} \lit from \preds as predicates over
these characters, \ie, $\denot{\lit}\subseteq\Sigma$, and accept
traces of characters from alphabet $\Sigma$.  The top literal is
denoted by $\bullet$ following the convention of regular expressions.
Note that, to avoid later confusion with boolean predicates, we adopt
a different set of notations for the boolean connectives.  The syntax
of SREs is then defined by the following operations: empty set
($\varnothing$), null ($\varepsilon$), literals ($\lit$),
Kleene Star ($\R^*$), concatenation ($\R_1\cdot\R_2$), negation
($\neg\R$), conjunction ($\R_1\land\R_2$), and disjunction
($\R_1\vee\R_2$).
\[\arraycolsep=4pt \begin{array}{lrl}
  \R & ::= & \varnothing \mid \varepsilon \mid \lit \mid \R^* \mid \R_1 \cdot \R_2 \mid \lnot\R \mid \R_1\land\R_2 \mid \R_1\lor\R_2 \\
  \end{array}\]
Abusing the notation $\denot{\_}$, the denotation of SREs,
$\denot{\R} \subseteq \Sigma^*$, is recursively defined as:
\begin{mathpar}
  \denot{\varnothing} = \{\}
  \and
  \denot{\varepsilon} = \{ \epsilon \}
  \and
  \denot{\R^*} = \denot{\R}^*
  \and
  \denot{\R_1 \cdot \R_2} = \denot{\R_1}\cdot\denot{\R_2}
  \\
  \denot{\lnot\R} = \Sigma^* \setminus \denot{\R}
  \and
  \denot{\R_1 \land \R_2} = \denot{\R_1} \cap \denot{\R_2}
  \and
  \denot{\R_1 \lor \R_2} = \denot{\R_1} \cup \denot{\R_2}
\end{mathpar}
Following the denotation of SREs, we write $\R_1\sqsubseteq\R_2$ for
$\denot{\R_1}\subseteq\denot{\R_2}$ and $\R_1\equiv\R_2$ for
$\denot{\R_1}=\denot{\R_2}$.

\paragraph{Conversion from \LTLf to SRE}
Interestingly, common temporal modalities from \LTLf can be expressed
in SRE.  As the leaf nodes in \LTLf formulae are literals \lit, which
is expressible in SRE, we provide the translation semantics of common
temporal operators, assuming that the operands have already been
converted to SREs, as follows:
\begin{mathpar}
  \Next \R \defeq \bullet \cdot \R
  \quad
  \lit \Until \R \defeq \lit^* \cdot \R
\quad
\Final \R \defeq \bullet^* \cdot \R
\quad
\Global \R \defeq \neg (\bullet^* \cdot \neg\R)
\quad
\lit \WeakUntil \R \defeq \neg(\bullet^*\cdot\R) \vee (\lit^*\cdot\R)
\end{mathpar}
That is, $\Next\R$ (next) holds if $\R$ accepts the trace starting
from the next position; $\lit\Until\R$ (until) holds if there exists
such a position that $\R$ accepts the following trace and $\lit$ holds
until that position; $\Final\R$ (finally) holds if there exists such a
position that $\R$ accepts the following trace; $\Global\R$ holds if
there does not exist such a position that $\R$ rejects the following
trace, and; $\lit\WeakUntil\R$ (weak until) holds if either there does
not exist such a position that $\R$ accepts the following trace, or
$\lit$ holds until such a position.  For simplicity, we limit the
first operand of \Until and \WeakUntil to be a single literal \lit,
which suffices for common cases found in the ADT specifications we
consider, including the modalities used in our evaluation
(\cref{sec:eval}).  We directly use SRE in the rest of the paper.

\paragraph{Derivatives of SRE}
A \textit{derivative} is a notion from language theory.  Given a
language, say defined by an SRE \R, and a string \trace, the
derivative operation returns a new language accepting all strings that
are accepted by \R when appended to \trace, which can be thought of as
a \textit{prefix} to those strings.
\[\denot{\deriv{\trace}{\R}} = \{\trace' \mid \trace \cdot \trace' \in \denot{\R}\}\]
Following the literature on derivatives of regular expressions
\cite{antimirovPartialDerivativesRegular1995,brzozowskiDerivativesRegularExpressions1964,berryRegularExpressionsDeterministic1986c},
we first inductively define a \emph{nullable} predicate $\nu(\R)$ that
determines if \R accepts the empty string.  That is, $\nu(\R)$ iff
$\epsilon\in\denot{\R}$.
\begin{mathpar}
  \nu(\varepsilon) = \nu(\R^*) = \top
  \and
  \nu(\varnothing) = \nu(\lit) = \bot
  \and
  \nu(\neg\R) = \neg\nu(\R)
  \and
  \nu(\R_1 \cdot \R_2) = \nu(\R_1 \land \R_2) = \nu(\R_1) \land \nu(\R_2)
  \and
  \nu(\R_1 \lor \R_2) = \nu(\R_1) \lor \nu(\R_2)
\end{mathpar}
Then the derivatives of SREs follow and can be computed recursively via the
following rules:
\begin{mathpar}
  \deriv{\varepsilon}{\R} = \R
  \and
  \deriv{\alpha\trace}{\R} = \deriv{\trace}{\deriv{\alpha}{\R}}
  \and
  \deriv{\alpha}{\varnothing} = \deriv{\alpha}{\varepsilon} = \varnothing
  \and
  \deriv[r]{\alpha}{\R^*} = \deriv{\alpha}{\R} \cdot \R^*
  \and
  \deriv{\alpha}{\lit} =
  \begin{cases}
    \varepsilon & \text{if } \alpha \in \denot{\lit} \\
    \varnothing & \text{if } \alpha \notin \denot{\lit}
  \end{cases}
  \and
  \deriv[r]{\alpha}{\R_1 \cdot \R_2} =
  \begin{cases}
    \inparen{\deriv{\alpha}{\R_1} \cdot \R_2} \vee \deriv{\alpha}{\R_2} & \text{if }\nu(\R_1) \\
    \deriv{\alpha}{\R_1} \cdot \R_2 & \text{if } \neg\nu(\R_1)
  \end{cases}
  \and\newline
  \deriv[r]{\alpha}{\neg\R} = \neg \deriv{\alpha}{\R}
  \and
  \deriv[r]{\alpha}{\R_1 \wedge \R_2} = \deriv{\alpha}{\R_1} \wedge \deriv{\alpha}{\R_2}
  \and
  \deriv[r]{\alpha}{\R_1 \vee \R_2} = \deriv{\alpha}{\R_1} \vee \deriv{\alpha}{\R_2}
\end{mathpar}

Computing the derivative of a regular expression is a well-known technique for
constructing an automaton that accepts the same language as the given
regular expression.  The construction closely follows a property of
regular expressions --- every SRE \R can be written in the form of a disjunction
as follows:
\[\R[nullable]\vee\bigvee_{\alpha\in\Sigma}\alpha\deriv{\alpha}{\R}
  \quad\text{where}\ \R[nullable]\ \text{is}\
  \varepsilon\ \text{if}\  \nu(\R), \varnothing\ \text{otherwise.}\]
Informally, starting with the initial state, each disjunct
$\alpha\deriv{\alpha}{\R}$ denotes a transition to a new state with
label $\alpha$.  If $\nu(\R)$, then we mark the current state as an
accepting state.  Iteratively, we repeat the same procedure on the new
states with the corresponding derivative $\deriv{\alpha}{\R}$ until no
new state can be added.  Intuitively, each state $q_i$ in the
constructed automaton is denoted by a derivative (also in SRE) of the
original \R -- the derivative accepts the same language as the
constructed automaton with its initial state set to $q_i$.  This can
be manifested by a different disjunctive form
$\bigvee_{\trace\in\Sigma^*}\trace\deriv{\trace}{\R}$.  For each
disjunct $\trace\deriv{\trace}{\R}$, if $\deriv{\trace}{\R}$ is not
empty, then $\trace$ denotes a path from the accepting state to the
state denoted by $\deriv{\trace}{\R}$.  Hence, whether
$\deriv{\trace}{\R}$ denotes an accepting state determines if $\R$
accepts $\trace$:
\[\trace\in\denot{\R}\ \text{iff}\ \nu(\deriv{\trace}{\R})\]

\noindent
However, $\Sigma$ often contains a large if not infinite number of
symbols and thus enumerating over all symbols to build an automata is
inefficient at best, and impossible in the general case.
\emph{Mintermization} solves this problem by constructing a finite set
of equivalence classes over the infinite domain $\Sigma$ such that all
literals $\lit$ can be mapped to elements in this finite set
(\cite{veanesRexSymbolicRegular2010,dantoniMinimizationSymbolicAutomata2014}).
Then, following a similar procedure of constructing automata from
regular expressions, one may construct an equivalent SFA where
transtitions between states are labeled by equivalence classes.  We
will present the characterization of \emph{symbolic derivatives} in
\cref{sec:derivative} as a device to exploit SREs' latent SFA
structures without upfront mintermization and later its computation in
\cref{sec:algo}.

\section{Trace-Based Symbolic Execution}\label{sec:naive}

\newcommand{\finalA}{{\lozenge}}
\newcommand{\globalA}{{\square}}
\newcommand{\nextA}{{\bigcirc}}
\newcommand{\lastA}{{\mathbf{LAST}}}
\newcommand{\allA}{{\mathbf{All}}}
\newcommand{\singletonA}{{\mathbf{Single}}}
\newcommand{\untilA}{\,\mathcal{U}\,}
\newcommand{\seqA}{\mathbf{;}}
\newcommand{\negA}{\neg}
\newcommand{\landA}{\land}
\newcommand{\lorA}{\lor}
\newcommand{\anyA}{\bullet}

\begin{figure}[th]
  \small
  \[\arraycolsep=4pt \begin{array}{llrl}
    \multicolumn{4}{c}{
    \text{Variable}\quad x,y,\ldots
    \qquad
    \text{Symbolic Variable}\quad \symb[\tau]{x}, \symb[\tau]{y},\ldots
    \qquad
    \text{Data Constructor}\quad d
    }                                                                                                                                                                              \\
    \multicolumn{4}{c}{
    \text{Primitive Operator}\quad op
    \qquad
    \text{Effectful API of Representation Type}\quad \eff{f}, \eff{g} \in\Delta

}                                                                                                                                                                              \\
    \text{Simple Type}                 & \tau       & ::=  & \var{unit} \mid \var{bool} \mid \var{int} \mid \ldots \mid \times\ \overline{\tau} \mid +\ \overline{d\ \tau}                        \\
    \text{Constant}                    & c          & ::=  & () \mid \mathbb{B} \mid \mathbb{Z} \mid \ldots \mid (\overline{c}) \mid d\ c                                                         \\
    \text{Symbolic First-Order Value}  & q          & ::=  & c \mid x \mid \symb[\tau]{x} \mid (\overline{q}) \mid d\ q \mid \op\ \overline{q}                                    \\
    \text{Boolean Formula}             & \phi, \Phi & ::=  & q \mid \bot \mid \top \mid \neg\phi \mid \phi \wedge \phi \mid \phi \vee \phi                                        \\
  \text{Symbolic Event}                & \lit       & ::=  & \inangle{x_\textsf{ret} \gets \eff{f}\ \overline{x_\textsf{arg}} \mid \phi} \mid \lneg\lit \mid \lit\lconj\lit \mid \lit\ldisj\lit    \\
    \text{Symbolic Value}              & v          & ::=  & x \mid q \mid (\overline{v}) \mid \funbind{x}e \mid \fixbind{f}\funbind{x}e                                          \\
    \text{Symbolic Expression}         & e          & ::=  & v \mid \?[\tau] \mid \abort \mid \assume{\phi} \mid \admit{\R} \mid \append{\R}                                    \\
                                        &            & \mid & \letin{x}{v~v}\ e \mid \letin{x}{e}\ e  \mid e \otimes e                                                             \\
  \end{array}\]
\[\begin{array}{lll}
  e_1; e_2                    & \doteq \letin{x}{e_1}\ e_2                               & \text{for fresh $x$}   \\
  \assert{\phi}               & \doteq (\assume{\neg\phi}; \abort) \otimes \assume{\phi} &                        \\
  \affirm{\R}                & \doteq (\admit{\neg\R}; \abort) \otimes \admit{\R}     & 
\end{array}\]
\caption{Syntax of the core language.}\label{fig:core-syntax}
\end{figure}

We first introduce a na\"{\i}ve variant of our symbolic execution
framework for falsifying functional ADT implementations that interact
with an underlying effectful representation type.  Symbolic execution
is defined on a core functional language with explicit constructs for
generating symbolic values and expressing specifications of two kinds:
formulae $\Phi$ from decidable theories amenable to SMT solving, which
are standard for symbolic execution techniques, and trace-based
specifications $\R$ expressed as symbolic regular expressions, which
is the novelty of our framework.  \figref{fig:core-syntax} presents
the syntax of our core language, where the expression $e$ is expressed
in monadic normal form
(MNF)~\cite{hatcliffGenericAccountContinuationpassing1994}, a variant
of A-normal form
(ANF)~\cite{flanaganEssenceCompilingContinuations1993} that permits
nested let-bindings.  Recursive functions take the form of
$\fixbind{f}\funbind{x}e$ using an explicit fixpoint construction, and
control flow is modeled by nondeterministic choice $\otimes$ together
with an \kwd{assume} construct.  The symbolic constructs in this
language, including \kwd{assume}, will be discussed throughout the
rest of the section.
\cref{fig:sem-naive} formalizes the naive symbolic execution of the
core language as a small-step, substitution-based operational
semantics over symbolic states $(\Phi, \R, e)$.  In this section,
we first introduce symbolic execution of pure
functional programs and then extend it with the capability to reason over traces
that interact with an ADT's underlying representation type.

To enable symbolic reasoning, the language supports
symbolic variables \symb[\tau]{x}, which stand for constants $c$ of
type $\tau$.  In contrast to program variables $x$, symbolic variables are
internal to symbolic execution: they are never written by developers but
are generated by the \?[\tau] construct during symbolic
execution (Rule \textsc{GenSym}).  The $\tau$ subscript can be omitted
whenever it is clear from the context; $\tau$ denotes simple types (primitive
types, \eg, \var{unit} and \var{int}, product
types $\times\ \overline{\tau}$, and user-defined data types
$+\ \overline{d\ \tau}$) but not function types.  Variables, symbolic
variables, and constants, when composed by tuple constructors
$(\dots)$, data constructors $d$, and primitive operators \op, build
up to symbolic (first-order) values.  Then, Boolean symbolic values, when
composed by logical connectives, build up to Boolean formulae $\Phi$.
Since primitive operators are drawn from decidable first-order
theories, \eg, arithmetic operators, or uninterpreted functions with
user-provided axioms, the satisfiability of Boolean formulae can be
straightforwardly discharged to SMT queries.  

\begin{definition}[Denotation of Boolean Formulae]\label{defn:denote-formula}
  Let $\sigma$ denote an interpretation of symbolic variables as
  constants and $\sigma(\Phi)$ denote the Boolean formula $\Phi$ with
  its symbolic variables substituted for constants according to
  $\sigma$.  Then, the denotation of a \emph{closed Boolean formula}
  $\Phi$, where all variables are symbolic, is the set of
  interpretations $\sigma$ such that $\sigma(\Phi)$ holds, \ie,
  $\denot{\Phi}=\{\sigma \mid \sigma(\Phi) \}$.
\end{definition}

\begin{figure}[t]
  \[ \text{Symbolic State} \qquad S\ ::=\ (\Phi, \R, e) \]
  \input{naive-semantics-rules}
  \caption{Naive trace-augmented semantics.}
  \label{fig:sem-naive}
\end{figure}

Now, we may introduce the symbolic execution of pure functional
programs, in which case each symbolic state $(\Phi, e)$ consists of a
closed Boolean formula $\Phi$ representing the current \emph{path
  condition}, and a \emph{closed} expression $e$ --- all variables are
bound by \kwd{let}, \kwd{fun}, or \kwd{fix}.  The path condition
collects all the conditions that need to be satisfied for the symbolic
state to be reachable, \ie, have a corresponding concrete state.
The initial path condition is true, denoted by $\true$.  Let $e$ be
the expression to be reduced.  Then the initial symbolic state is
$(\true, e)$.  The reduction rules between symbolic states are
described in \cref{fig:sem-naive} if we omit rules related to \R.
Rule \textsc{Assume} describes the augmentation of the current path
condition with the argument of \kwd{assume}, which is also a Boolean
formula $\phi$.  In contrast to path conditions, we use $\phi$ for
Boolean formulae that may involve program variables bound in
expressions, which will be substituted for closed symbolic values via
Rule \textsc{LetVal}.  Each \kwd{assume} along an execution path
further restricts the state space represented by the path condition.
Suppose a path condition $\Phi$ is satisfiable, \ie, there exists
$\sigma\in\denot{\Phi}$.  Then, a sequence of reductions from
$(\true, e)$ to $(\Phi, \abort)$, where \abort represents a failure in
execution, witnesses a \emph{feasible} execution path of $e$ that
leads to a failure.  In practice, \abort is rarely written by
developers and can be expressed using \kwd{assert}, which is
defined as  syntactic sugar (\cref{fig:core-syntax}).  
Whether a path condition $\Phi$ passes an assertion $\assert{\phi}$ is
effectively determined by the satisfiability of $\Phi\wedge\neg\phi$.

To reason about an ADT's interaction with underlying representation
types, we equip symbolic execution with the capability to model such
interactions extensionally, by recording the history of calls to the
representation type's methods, along with their argument and return
values.  In particular, interactions are captured in symbolic regular
expressions (SRE) \R whose literals denote sets of such API calls to
the representation types.  Recall in \cref{sec:prelim}, such literals
are elements from EBA
$(\Sigma, \preds, \denot{\_}, \bot, \bullet, \ldisj, \lconj, \lneg)$.
Here, $\Sigma$ stands for the domain of \emph{events}, denoted by
$c_\textsf{ret}\gets\eff{f}\ \overline{c_\textsf{arg}}$, and $\preds$ includes all
the \emph{symbolic events} $\lit$, each denoting set of events,
according to the syntax shown in \cref{fig:core-syntax}.  An
\emph{atomic symbolic event}
$\inangle{x_\textsf{ret}\gets\eff{f}\ \overline{x_\textsf{arg}}\mid\phi}$ denotes
the calls to \eff{f} such that the arguments $\overline{c_\textsf{arg}}$ and
the return value $c_\textsf{ret}$ satisfy the qualifier $\phi$:
\[
  \denot{\inangle{x_\textsf{ret}\gets\eff{f}\ \overline{x_\textsf{arg}}\mid\phi}}
  \doteq \{ c_\textsf{ret}\gets\eff{f}\ \overline{c_\textsf{arg}} \mid
  [\overline{x_\textsf{arg}\mapsto c_\textsf{arg}}, x_\textsf{ret}\mapsto c_\textsf{ret}]\phi\}
\]
The boolean connectives have standard denotation as shown in
\cref{sec:prelim}.  
Notice that the scope of
$\overline{x_\textsf{arg}}$ and $x_\textsf{ret}$ is limited to the qualifier $\phi$
of the symbolic event.  We omit such variables \emph{local} to the
symbolic event when they are either obvious from or irrelevant to the
context.  For example, we always use $key$ and $val$ to denote keys
and values of the calls to \eff{put} and \eff{get} from key-value
stores, with the result of \eff{put} omitted.  And similar to
\cref{sec:motiv}, we write \inangle{\eff{put}\ \hat{k}\ \hat{v}} for
\inangle{\eff{put}\ key\ val \mid key=\hat{k} \wedge val=\hat{v}} and
\inangle{\eff{put}\ \cancel{\hat{k}}\ \hat{v}} for \inangle{\eff{put}\
  key\ val \mid key\neq \hat{k} \wedge val=\hat{v}}.  An atomic
symbolic event is \emph{closed} if all variables in its qualifier are
either symbolic or local to the event; a symbolic event $\lit$ is
\emph{closed} if all its atomic symbolic events are, and; an SRE is
\emph{closed} if all its symbolic events are.  The denotation does not
apply to all closed SREs but only those SREs without symbolic
variables.  For SREs that reference symbolic variables, we can only
interpret them after interpreting these symbolic variables in a way
consistent with the path condition if any.

By augmenting symbolic states with SRE \R[curr] to represent the
events that have happened, we define a reduction semantics over
$(\Phi,\R[curr],e)$ as shown in \cref{fig:sem-naive}.  We refer to
such an SRE \R[curr] as the \emph{current context} of the execution
from the associated symbolic state.  In addition, we refer to SREs as
\emph{contexts} or \emph{effects} of a method (ADT's or representation
type's) depending on whether they describe admissible traces prior to
calling the method or traces the method is supposed to produce.
Similar to path conditions, the SREs that represent the current
context of symbolic states are always closed.
\cref{defn:reachability} gives the reachability of a symbolic state
$S$ based on the satisfiability of its path condition $\Phi$ and its
current context \R[curr].
\begin{definition}[Reachability]\label{defn:reachability}
  $\sat{\Phi, \R[curr]}$ iff there exists $\sigma \in \denot{\Phi}$
  such that $\denot{\sigma(\R[curr])} \neq \emptyset$.
\end{definition}
\noindent
Henceforth, we omit the carat ($\hat{\phantom{x}}$) on symbolic variables
\symb{x} and assume all variables are symbolic except for those
variables bound in expressions.

For a symbolic state $(\Phi, \R[curr], e)$, its path condition $\Phi$
effectively captures the history of pure computation up to
this state
while its current context \R[curr] captures the history of effectful
computation.  Because the events in \R[curr] are qualified with
reference to the symbolic variables in $\Phi$, both structures
synergyistically enable the recording of a sufficient condition that
allows a computation to reach the symbolic state.

\newcommand{\nodenext}[2]{\EM{\R_{#1 \curvearrowright #2}}}
\newcommand{\nodenextunder}[2]{\EM{\Trace_{#1 \curvearrowright #2}}}
\newcommand{\nodeelem}[2]{\EM{\R_{#1 : #2}}}
\newcommand{\nodeelemunder}[2]{\EM{\Trace_{#1 : #2}}}
\newcommand{\nodebidir}[2]{\EM{\R_{#1 \curvearrowleftright #2}}}
\begin{example}\label{ex:spec}
    The \texttt{remove} method from \cref{fig:motive-ex} can be
  rewritten in our core language, with the conditional expression
  represented by a combination of \kwd{assume} and choice operation
  $\otimes$, as follows:
  \[e_\texttt{remove}\defeq\funbind{hd}\funbind{elem}
    (\assume{hd=\texttt{null}}; hd) \otimes
    (\assume{hd\neq\texttt{null}};\ \dots)
  \]
  Recall that the specification in \cref{fig:motive-ex} requires
  \texttt{remove} to be called in a symbolic state where node $a$
  is linked to node $b$ as its successor.  Its precondition can be
  written as an SRE thus:
  \[\nodenext{a}{b}\defeq
    \bullet^*\cdot \inangle{\eff{Nxt.put}\ a\ b}\cdot
    (\lneg\inangle{\eff{Nxt.put}\ a\ \_})^*\] %
  This context admits traces in which a call to
  \eff{Nxt.put} is made on key $a$ and value $b$, followed by
  subsequent events that do not include calls to \eff{Nxt.put} with
  key $a$.  Hence, the context encapsulates the intended requirement
  on symbolic states prior to calling \texttt{remove}.

  The specification in \cref{fig:motive-ex} also requires that
  \texttt{remove}, when called under the specified context, can link a
  new node other than $a$ to $b$ \emph{only} when $b$ has been
  unlinked from $a$.  The postcondition can be written as an SRE
  parameterized by $a$ and $b$ thus:
  \[\nodebidir{a}{b}\defeq
    ((\lneg\inangle{\eff{Nxt.put}\
      \cancel{a}\ b})^* \cdot\inangle{\eff{Nxt.put}\ a\
      \cancel{b}}\cdot\bullet^*)
    \vee
    (\lneg\inangle{\eff{Nxt.put}\ a\ \cancel{b}})^*
\]
The effect of \textt{remove} admits traces where no node other than
$a$ is linked to $b$ (via \lneg\inangle{\eff{Nxt.put}\ \cancel{a}\ b})
before $a$ is unlinked from $b$ (via \inangle{\eff{Nxt.put}\ a\
  \cancel{b}}), or $a$ is never unlinked from $b$ (via
\lneg\inangle{\eff{Nxt.put}\ a\ \cancel{b}}).  If an execution of
\texttt{remove} produces traces not admissible to \nodebidir{a}{b},
then we conclude that some node, as witnessed by $b$, may unexpectedly
have two predecessors at some point during the execution.  \qed
\end{example}

\begin{toappendix}
  \begin{figure}[h]
      \small
  \[\arraycolsep=4pt \begin{array}{llrl}
    \multicolumn{4}{c}{
    \text{Variable}\quad x,y,\ldots
    \qquad
    \text{Data Constructor}\quad d}
    \\
    \multicolumn{4}{c}{
    \text{Primitive Operator}\quad op
    \qquad
    \text{Effectful API of Representation Type}\quad \eff{f}\in\Delta}
    \\
    \text{Constant}                    & c          & ::=  & () \mid \mathbb{B} \mid \mathbb{Z} \mid \ldots \mid (\overline{c}) \mid d\ c                                                         \\
    \text{Value}              & v          & ::=  & c \mid x \mid (\overline{v}) \mid \funbind{x}e \mid \fixbind{f}\funbind{x}e                                          \\
    \text{Expression}         & e          & ::=  & v \mid \abort \mid \assume{\phi} \mid  e \otimes e \mid \letin{x}{\eff{f}~\overline{v}}~e \\
                                        &            & \mid & \letin{x}{op~\overline{v}} \mid \letin{x}{v~v}\ e \mid \letin{x}{e}\ e 
  \end{array}\]
\caption{Syntax of the source language.}
  \end{figure}
  \begin{figure}[h]
  \[\textsf{spec} ::= \texttt{function}\ (\overline{x})\
    \texttt{ghost}\ \overline{z}\ \texttt{require}\ \phi\
    \texttt{context}\ \R[ctx]\ \texttt{return}\ y\ \texttt{ensures}\ \psi\
    \texttt{effect}\ \R[eff]
  \]
  \judgbox{ \textsf{getHarness}(\textsf{spec}, v_f)=e }{ %
    Construct a harness for symbolic execution of $v_f$
    against \textsf{spec}. %
  }
  \begin{mathpar}
    \textsf{getHarness}(\texttt{function}\ (\overline{x})\
    \texttt{ghost}\ \overline{z}\ \texttt{require}\ \phi\
    \texttt{context}\ \R[ctx]\ \texttt{return}\ y\ \texttt{ensures}\ \psi\
    \texttt{effect}\ \R[eff], v_f)=\newline
    \letin{\overline{x}}{\overline{\?}}\ %
    \letin{\overline{z}}{\overline{\?}}\ %
    \assume{\phi}; \append{\R[ctx]};
    \letin{y}{\textsf{V}\denot{v_f}\ \overline{x}}\ %
    \assert{\psi}; \affirm{\R[ctx]\cdot\R[eff]}
  \end{mathpar}
  \judgbox{\textsf{V}\denot{v}=v}{%
    Translate values in source language to symbolic values.}
  \begin{mathpar}
    \textsf{V}\denot{c} = c \and %
    \textsf{V}\denot{x} = x \and %
    \textsf{V}\denot{(\overline{v})} =
    (\overline{\textsf{V}\denot{v}}) \\%
    \textsf{V}\denot{\funbind{x}e} = \funbind{x}\textsf{E}\denot{e}
    \and %
    \textsf{V}\denot{\fixbind{f}\funbind{x}e} =
    \fixbind{f}\funbind{x}\textsf{E}\denot{e} %
  \end{mathpar}
  \judgbox{\textsf{E}\denot{e}=e}{Translate expressions in source
    language to symbolic expressions.}
  \begin{mathpar}
    \textsf{E}\denot{v} = v \and
    \textsf{E}\denot{\abort} = \abort \and
    \textsf{E}\denot{\assume{\phi}} = \assume{\phi} \and
    \textsf{E}\denot{e_1\otimes e_2} = %
    \textsf{E}\denot{e_1} \otimes \textsf{E}\denot{e_2} \and
    \textsf{E}\denot{\letin{x}{v_f\ v}\ e} =
    \letin{x}{\textsf{V}\denot{v_f}\ \textsf{V}\denot{v}}\
    \textsf{E}\denot{e} \and %
    \textsf{E}\denot{\letin{x}{e_1}\ e_2} =
    \letin{x}{\textsf{E}\denot{e_1}}\ \textsf{E}\denot{e_2} \\
    \textsf{E}\denot{\letin{x}{op~\overline{v}}~e} = %
    \letin{x}{ %
        op~\overline{\textsf{V}\denot{v}}
    }~\textsf{E}\denot{e} \\
    \textsf{E}\denot{\letin{x}{\eff{f}~\overline{v}}~e} =
    \letin{x}{ %
      (\funbind{\overline{x}} %
      \letin{\overline{z}}{\overline{\?}}\ %
      \assert{\phi}; \admit{\R[ctx]}; \\%
      \letin{y}{\?}\ %
      \assume{\psi}; \append{\R[eff]}; y)\ \overline{v} %
    }\ \textsf{E}\denot{e} \newline\text{where the spec of \eff{f} is}\ 
    \texttt{function}\ (\overline{x})\
    \texttt{ghost}\ \overline{z}\ \texttt{require}\ \phi\
    \texttt{context}\ \R[ctx]\ \texttt{return}\ y\ \texttt{ensures}\ \psi\
    \texttt{effect}\ \R[eff]
  \end{mathpar}
  \caption{Translation semantics of the source language with respect to the specification language.}
  \end{figure}
\end{toappendix}
Similar to how path conditions are augmented by \kwd{assume}
constructs, the current contexts of executions are augmented by two
constructs, $\admit{\R[past]}$ and $\append{\R[eff]}$.  The former,
\kwd{admit}, combines the current context \R[curr] and its argument
\R[past] with conjunction, as described by Rule \textsc{Admit}; thus,
it restricts the traces of past events in \R[curr] to only those
admissible by \R[past].  In contrast, \kwd{append} concatenates the
current context \R[curr] with the argument \R[eff], as described by
Rule \textsc{Append}; thus it records new events produced during
symbolic execution.  The initial context before starting the symbolic
execution is $\varepsilon$, indicating that no event has happened yet.

Now, we illustrate that a pair of \kwd{append} and \kwd{affirm}
constructs the translation of the specification attached to an ADT
method, capturing the safety property.  Recall from
\cref{fig:motive-ex} that the specification includes three key
components, ghost variables, required context (\texttt{context}), and
expected effect (\texttt{effect}).  Intuitively, the specification
states that when being executed in a required context (with possible
reference to both ghost variables and method parameters), the method
with the specification attached should produce events in compliance
with the expected effect.  The \kwd{append} helps set up this required
context while the \kwd{affirm} is responsible for affirming that the
context upon exiting the method complies with its argument \R[post]
by conjoining the context \R[curr] with $\neg\R[post]$.  The
satisfiability of $\R[curr]\wedge\neg\R[post]$ then witnesses a
violation of \R[post] in the execution manifested by \R[curr].  To
falsify the implementation of the ADT method against its
specification, we construct a harness $e_{harness}$\footnote{See
  supplementary material for details on the translation of source
  language expressions to core language ones.} that wraps a
call to the ADT method with such a pair of \kwd{append} and
\kwd{affirm}.
\begin{example}\label{ex:harness}
  Continuing from \cref{ex:spec}, the specification of\newline
  \begin{minipage}[t]{.68\textwidth}\hyphenpenalty=10000
    \texttt{remove} is converted into a harness $e_{harness}$.  First,
    symbolic variables $a$ and $b$ are generated (by ``\?'' with the
    same name as the program variables) to denote two arbitrary nodes.
    Second, symbolic variables $hd$ and $u$ are generated (again by
    ``\?'')  to denote the input to \texttt{remove}.  Third, the
    required context \nodenext{a}{b} of \texttt{remove} is appended to
    the initial context $\varepsilon$.  Lastly, after calling
    \texttt{remove}, the postcondition of
  \end{minipage}\hfill
  \begin{minipage}[t]{0.3\textwidth}\vspace*{-1.1cm}
    \begin{align*}
      &e_{harness}\defeq \\
      &\letin{a,b}{\?[node],\?[node]} \\
      &\letin{hd,u}{\?[node],\?[elem]} \\
      &\append{\nodenext{a}{b}}; \\
      &\texttt{remove}\ hd\ u; \\
      &\affirm{\nodenext{a}{b}\cdot\nodebidir{a}{b}}
    \end{align*}  
  \end{minipage}
  the harness is affirmed to check for any violation during the
  execution of the harness.  Notably, the postcondition prepends the
  required context \nodenext{a}{b} to the expected effect
  \nodebidir{a}{b}. \qed
\end{example}
\noindent
In contrast, pair(s) of \kwd{admit} and \kwd{append} construct the
translation of the specifications attached to APIs of representation
types, providing an extensional and underapproximate model for their
behavior.  And the \kwd{admit} relates the context require by the API
with the current context by conjoining them while the following
\kwd{append} records the expected traces of events produced by the
API.
\begin{example}\label{ex:get}
  Taking the same form as the required context \nodenext{a}{b} of
  \texttt{remove} from \cref{ex:spec}, the required context of
  \eff{Nxt.get} is \nodenext{s}{t}, admitting traces where node $s$
  (the argument) is linked to node $t$ (the return value).  And the
  expected effect of the \eff{Nxt.get} is a symbolic event
  \inangle{t\gets\eff{Nxt.get}\ s}.  And thus, calls to \eff{Nxt.get}
  can be replaced by a function $e_\eff{Nxt.get}$ defined as follows:
  \[e_\eff{Nxt.get}\defeq\funbind{s}
    \letin{t}{\?[node]}\ \admit{\nodenext{s}{t}};\
    \append{\inangle{t\gets\eff{Nxt.get}\ s}};\ t\]
  Similarly, let
  $\R_{n : u}\defeq \bullet^*\cdot
  \inangle{\eff{Val.put}\ n\ u}\cdot
  (\lneg\inangle{\eff{Val.put}\ n\ \_})^*$
  denote the required context of \eff{Val.get}, where node $n$ (the
  argument) stores an element $u$ (the return value).
  Correspondingly, calls to \eff{Val.get} are replaced by a function
  $e_\eff{Val.get}$ defined as follows:
  \[e_\eff{Val.get}\defeq\funbind{n}
    \letin{u}{\?[elem]}\ \admit{\nodeelem{n}{u}};\
    \append{\inangle{u\gets\eff{Val.get}\ n}};\ u\]
  Here, the calls to \eff{get} always succeed because our goal is to
  falsify the implementation of \texttt{remove} with respect to the specified safety property. %
  \qed
\end{example}
\noindent
By replacing API calls in the direct translation of the ADT method,
\eg, $e_\texttt{remove}$ from \cref{ex:spec}, with symbolic expressions
that augment the context of execution using \kwd{admit} and
\kwd{append}, we now have an implementation $e_{remove}$ of the ADT
method \texttt{remove} that is ready to be plugged in $e_{harness}$ for symbolic
execution.
\begin{example}\label{ex:naive}
  By substituting the call to \texttt{remove} for $e_\texttt{remove}$
  (\cref{ex:spec}) with $\eff{Val.get}$ and $\eff{Nxt.get}$
  respectively substituted for $e_\eff{Val.get}$ and $e_\eff{Nxt.get}$
  (\cref{ex:get}), the harness $e_{harness}$ (\cref{ex:harness}) is
  closed and ready for symbolic execution.  Initially, the symbolic
  state is $(\true,\varepsilon,e_{harness})$.  The required context
  \nodenext{a}{b} of \texttt{remove} is first appended to $\varepsilon$.
  Following the second branch, $n_0\neq\texttt{null}$ augments the
  path condition.  As \eff{Val.get} is called on $n_0$, $n_0$
  substitutes $n$ in the body of $e_\eff{Val.get}$ and a fresh
  symbolic variable $u_0$ is generated to represent the element stored
  in $n_0$.  The symbolic state becomes
  \[(n_0\neq\texttt{null}, \nodenext{a}{b},
    \letin{u'}{
      \admit{\nodeelem{n_0}{u_0}};\
      \append{\inangle{u_0{\gets}\eff{Val.get}\ n_0}};\
      u_0}\ \dots)
  \]
  As $u_0$ is returned to the top level and substitutes $u'$, the
  current context becomes
  $\nodenext{a}{b}\wedge\nodeelem{n_0}{u_0}\cdot\inangle{u_0{\gets}\eff{Val.get}\
    n_0}$ ($\wedge$ binds SREs tighter than $\cdot$).  Following the
  nested second branch, the path condition becomes
  $n_0\neq\texttt{null}\wedge u_0\neq u$.  As we enter the
  \texttt{loop} and follow the execution path illustrated in
  \cref{sec:motiv-symb-exec}, we
  \begin{enumerate*}
  \item get the successor of $n_0$, $n_1$,
  \item get the element stored in $n_1$, $u_1$,
  \item assume $u_1=u$,
  \item get the successor of $n_1$, $n_2$, and
  \item remove $n_1$ by linking $n_0$ to $n_2$.
  \end{enumerate*}
  The symbolic state becomes
  $(\Phi[bad],\R[bad],
  \affirm{\nodenext{a}{b}\cdot\nodebidir{a}{b}})$, where
 \begin{align*}
   \Phi[bad]\defeq\ &
   n_0\neq\texttt{null}\wedge u_0\neq u\wedge
   n_1\neq\texttt{null}\wedge u_1=u \quad\text{and} \\
   \R[bad]\defeq\ &
   ((((\nodenext{a}{b}{\wedge}
   \nodeelem{n_0}{u_0}\cdot\inangle{u_0{\gets}\eff{Val.get}\ n_0})
   {\wedge}\nodenext{n_0}{n_1}\cdot\inangle{n_1{\gets}\eff{Nxt.get}\ n_0}) \\
   &{\wedge}\nodeelem{n_1}{u_1}\cdot\inangle{u_1{\gets}\eff{Val.get}\ n_1})
   {\wedge}\nodenext{n_1}{n_2}\cdot\inangle{n_2{\gets}\eff{Nxt.get}\ n_1})
   \cdot\inangle{\eff{Nxt.put}\ n_0\ n_2}
 \end{align*}
 $\R[bad]$ denotes the traces that can be produced following the
 execution path.  To show that the affirmation may fail, it is
 sufficient to find an interpretation for the symbolic variables such
 that the path condition $\Phi[bad]$ holds and there exists a trace
 included in $\R[bad]$ but excluded from the postcondition of the
 harness,
 \ie,
 $\sat{\Phi[bad],
   \R[bad]\wedge\neg(\nodenext{a}{b}\cdot\nodebidir{a}{b})}$.\qed
\end{example}

\noindent
As illustrated in \cref{ex:naive}, the size of the SRE that represents
the current context of the execution quickly blows up during symbolic
execution.  This in turn makes the symbolic affirmation check at the
end of each execution path potentially very expensive as we quantify in
\cref{sec:eval}.

To conclude this section, we lift the \kwd{affirm} check at the end of
the harness progress and regard it as a falsification query on the
harness program without the \kwd{affirm} statement.  Then
\cref{defn:naive-falsification} describes a falsification
problem in terms of this trace-based symbolic execution
framework.
\begin{definition}[Na\"{\i}ve Falsification]\label{defn:naive-falsification}
  Given a safety property \R[post].
  If $(\top, \varepsilon, e) \hookrightarrow^* (\Phi, \R[curr], v)$
  and $\sat{\Phi, \R[curr]\wedge\neg\R[post]}$
  then this execution of $e$ is falsified with respect to \R[post].
\end{definition}

\section{Symbolic Execution with Symbolic Derivatives}\label{sec:derivative}

The inefficiency of the naive semantics stems from its failure to
recognize \emph{regularity} --- the capacity of specifications to be
represented as automata structures --- during symbolic execution.  We
can exploit this regularity by underapproximating the required context
or the expected effect of method calls.  This approximation
facilitates a derivative computation, effectively emulating state
transitions in the SFAs associated with SREs.  Specifically, the
underapproximation takes the form of \emph{symbolic traces} $\Trace$,
where only a subset of operations from the SRE (with the same
denotation) are allowed: empty trace ($\varepsilon$), symbolic event
($\lit$), and concatenation ($\Trace[1]\cdot\Trace[2]$).  Then, a
derivative-based notion of symbolic state $S_\mathcal{D}$ that
underapproximates a symbolic state $S$, besides the expression $e$
under execution, is given by (i) $\Phi$ and $\Trace$ to encapsulate
the execution path that leads to $S_\mathcal{D}$; along with (ii)
$\R[cont]$ that predicts the traces allowed to be produced in the
continuation of the execution in compliance with the safety property,
dubbed \emph{continuation effect}.
\[\text{Symbolic Trace}\ \Trace\ ::=\ \varepsilon \mid \lit \mid \Trace \cdot \Trace
  \qquad
  \text{Derivative-Based Symbolic State}\ S_\mathcal{D}\ ::=\ (\Phi,\Trace,\R[cont],e)
\]
In this section, we present
\begin{enumerate*}
\item \emph{symbolic derivatives} that allow us to effectively explore
  and thus exploit the automata structures of specifications, without
  appealing to their calculation (see \cref{sec:algo}),
  and
\item a \emph{derivative-based semantics} that leverages this notion
  to facilitate symbolic execution over $S_\mathcal{D}$ as
  well as the falsification of a given safety property that is both
  sound and complete with respect to the na\"{\i}ve semantics given
  in the previous section.
\end{enumerate*}

\subsection{Symbolic Derivatives}\label{sec:symbolic-derivative}

SREs that represent the context of the current execution or the
arguments to \kwd{admit} and \kwd{append} may refer to symbolic
variables that are also constrained by path conditions, as
discussed in \cref{sec:naive}.  In what follows, we first revisit 
notions on SREs from \cref{sec:prelim} with such symbolic
variables left uninterpreted, \ie, treating symbolic variables as
abstract symbols whose interpretation is unknown.   We then define
\emph{symbolic derivatives} of such SREs, which may also refer to
symbolic variables involved in those SREs.

First, the inclusion and equivalence relationship between two SREs
$\R[1]$ and $\R[2]$ is given by \cref{defn:inclusion} and
\cref{defn:equivalence} such that the relationship holds under any
interpretation of symbolic variables involved in \R[1] and \R[2].
\begin{definition}[Inclusion]\label{defn:inclusion}
  $\R[1] \sqsubseteq \R[2]$ iff
  $\denot{\sigma(\R[1])}\subseteq\denot{\sigma(\R[2])}$ for all $\sigma$.
\end{definition}
\begin{definition}[Equivalence]\label{defn:equivalence}
  $\R[1] \equiv \R[2]$ iff
  $\denot{\sigma(\R[1])}=\denot{\sigma(\R[2])}$ for all $\sigma$.
\end{definition}
\noindent
Second, the nullable operation $\nu$ defined over SREs in
\cref{sec:prelim}, when applied to any symbolic event $\lit$, returns
false irrespective of \lit's qualifiers and any symbolic variables
involved.  Hence, the nullable operation $\nu(\R)$ determines if $\R$
accepts the empty trace $\epsilon$ regardless of the interpretation of
its symbolic variables (\cref{lem:symbolic-nullable}).

\begin{lemmarep}\label{lem:symbolic-nullable}
  $\nu(\R)$ iff $\nu(\sigma(\R))$ for all $\sigma$.
  % \footnote{All
  % proofs are deferred to the full version of this paper
  % \cite{yuanDerivativeGuidedSymbolicExecution2024}.}
  \footnote{See appendix for all proofs.}
\end{lemmarep}
\begin{proof}
  By structural induction on $\R$:
  \begin{enumerate}[label=Case , wide=0pt, font=\bfseries]
  \item $\varnothing$: %
    $\neg\nu(\varnothing)$ and
    $\neg\nu(\sigma(\varnothing))$ for all $\sigma$.
  \item $\varepsilon$: %
    $\nu(\varepsilon)$ and $\nu(\sigma(\varepsilon))$ for all $\sigma$.
  \item $\lit$: %
    $\nu(\lit)$ and $\nu(\sigma(\lit))$ for all $\sigma$.
  \item $\R^*$: %
    $\nu(\R^*)$ and $\nu(\sigma(\R^*))$ for all $\sigma$.
  \item $\R\cdot\R'$: %
    $\nu(\R\cdot\R')=\nu(\R)\wedge\nu(\R')$, which, by IH, is
    equivalent to
    $\nu(\sigma(\R))\wedge\nu(\sigma(\R'))=\nu(\sigma(\R\cdot\R'))$
    for all $\sigma$.
  \item $\neg\R$: %
    $\nu(\neg\R)=\neg\nu(\R)$, which, by IH, is equivalent to
    $\neg\nu(\sigma(\R))=\nu(\sigma(\neg\R))$ for all $\sigma$.
\item $\R\wedge\R'$: %
$\nu(\R\wedge\R')=\nu(\R)\wedge\nu(\R')$, which, by IH, is
equivalent to
$\nu(\sigma(\R))\wedge\nu(\sigma(\R'))=\nu(\sigma(\R\wedge\R'))$
for all $\sigma$.
\item $\R\vee\R'$: %
$\nu(\R\vee\R')=\nu(\R)\vee\nu(\R')$, which, by IH, is
equivalent to
$\nu(\sigma(\R))\vee\nu(\sigma(\R'))=\nu(\sigma(\R\vee\R'))$ for
all $\sigma$. \qedhere
\end{enumerate}
\end{proof}
\noindent
Third, we need to revisit the notion of prefixes of SREs.
Recall in \cref{sec:prelim}, for an SRE \R that does not involve 
symbolic variables, any concrete trace \trace can be a prefix of
\R since the derivative $\deriv{\trace}{\R}$, which contains all
concrete traces that are accepted by \R when appended to \trace, is
always well-defined.  To account for symbolic variables involved in \R, we
consider symbolic traces \Trace, which may also involve these symbolic
variables, as a consolidated form of prefixes of \R.  
\cref{defn:prefix} gives the criteria that a valid prefix \Trace of
\R has to meet.
\begin{definition}\label{defn:prefix}
$\Trace$ is a prefix of $\R$
iff there exists $\R'$ s.t. $\deriv{\trace}{\sigma(\R)}{=}\sigma(\R')$
for all $\trace{\in}\denot{\sigma(\Trace)}$ for all $\sigma$.
\end{definition}
\noindent
Intuitively, \Trace qualifies as a prefix of \R if it represents a
collection of partial runs of an SFA associated with \R.  These
partial runs must begin with the SFA's start state and end at any
arbitrary state.  Notably, the ending state does not need to be
accepting and may even be a \emph{dead state}, from which no accepting
state is accessible.  For example, recall the postcondition automaton
from \cref{fig:remove-automata}: \inangle{\eff{Nxt.put}\ \var{a}\
  \cancel{\var{b}}} and \inangle{\eff{Nxt.put}\ \cancel{\var{a}}\
  \var{b}} are both valid prefixes but their disjunction is not
because some runs end at $q_3$ while the others end at the dead state
denoted by $\varnothing$.  We dub such singleton prefixes as
\emph{next events}.

Following the definition of prefixes, we introduce symbolic
derivatives in \cref{defn:symbolic-derivative}.
\begin{definition}[Symbolic Derivative]\label{defn:symbolic-derivative}
$\Deriv{\Trace}{\R} = \R'$ iff
$\deriv{\trace}{\sigma(\R)} = \sigma(\R')$
for all $\trace \in \denot{\sigma(\Trace)}$ for all $\sigma$.
\end{definition}
\noindent
In contrast to conventional derivatives discussed in
\cref{sec:prelim}, symbolic derivatives of \R are well defined
\emph{only} over its prefixes (\cref{defn:prefix}) but not arbitrary
symbolic traces.  And the property of prefixes ensures that symbolic
derivatives can still be succinctly expressed as SREs with references
to symbolic variables if any.  Notably, \R, its prefix \Trace, and its
symbolic derivative \R' over \Trace shall interpret any referenced
symbolic variables in a consistent way;
\cref{defn:symbolic-derivative} serves as a guard against inconsistent
interpretations.

Since each valid prefix \Trace of \R establishes an equivalence class
where all \trace denoted by \Trace produce the same derivative, a
symbolic derivative \Deriv{\Trace}{\R} is not only a \emph{quotient}
but also a \emph{residual} of \R with respect to \Trace.  As noted by
\cite{prattActionLogicPure1991}, the quotient of \R contains traces
that are accepted by \R when appended to \emph{some} \trace denoted by
prefix \Trace, while the residual of \R contains traces that are
accepted by \R when appended to \emph{any} \trace denoted by prefix
\Trace.  This residuality is manisfest by \cref{cor:residuality}, \ie,
the concatenation of prefix \Trace and \Deriv{\Trace}{\R} is included
in \R itself.
\begin{corollaryrep}[Residuality]\label{cor:residuality}
Let $\R'=\Deriv{\Trace}{\R}$.
Then $\Trace\cdot\R'\sqsubseteq\R$.
\end{corollaryrep}
\begin{proof}
Assume an arbitrary interpretation $\sigma$ that closes \Trace, \R, and \R'.
By \cref{defn:symbolic-derivative},
$\deriv{\trace}{\R}=\R'$ for all $\trace \in \denot{\Trace}$.
That is
$\denot{\R'}=\{ \trace' \mid \trace \cdot \trace' \in \denot{\R} \}$
for all $\trace\in\denot{\Trace}$.
Then
$\denot{\Trace\cdot\R'}= \{ \trace\cdot\trace' \mid
\trace\in\denot{\Trace} \wedge \trace \cdot \trace'\in\denot{\R} \}
\subseteq\denot{\R}$.
By \cref{defn:inclusion}, $\Trace\cdot\R'\sqsubseteq\R$.
\end{proof}
\begin{example}\label{ex:symbolic-derivative}
Consider the expected effect $\nodebidir{a}{b}$ of \texttt{remove}
from \cref{ex:spec}, admitting traces where either no node other
than $a$ may be linked to $b$ before $b$ is unlinked from $a$, or
$a$ is linked to $b$ during the course of execution.
Its next events include \inangle{\eff{Nxt.put}\ a\ \cancel{b}},
\inangle{\eff{Nxt.put}\ \cancel{a}\ b}, and \inangle{\eff{Nxt.put}\
\var{a}\ \var{b}} \ldisj \inangle{\eff{Nxt.put}\ \cancel{\var{a}}\
\cancel{\var{b}}} \ldisj \lneg\inangle{\eff{Nxt.put}}.
Their symbolic derivatives are defined as follows,
with their respective residuality manifested:
\begin{enumerate*}
\item
\Deriv{\inangle{\eff{Nxt.put}\ a\ \cancel{b}}}
{\nodebidir{a}{b}}=$\bullet^*$
because any event is allowed once $b$ is unlinked from $a$;
\item
\Deriv{\inangle{\eff{Nxt.put}\ \cancel{a}\ b}}
{\nodebidir{a}{b}}=$\varnothing$
because it is unsafe to link node other than $a$ to $b$ with $a$ linked to $b$, and;
\item
\Deriv{\inangle{\eff{Nxt.put}\ \var{a}\ \var{b}}
\ldisj \inangle{\eff{Nxt.put}\ \cancel{\var{a}}\ \cancel{\var{b}}}
\ldisj \lneg\inangle{\eff{Nxt.put}}}
{\nodebidir{a}{b}}=\nodebidir{a}{b}
because linking $a$ to $b$ again, linking nodes other than $a$ and
$b$, or \inlinecaml{get} calls have no effect on subsequent traces
admissible by \nodebidir{a}{b}.
\end{enumerate*}
\qed
\end{example}

Recall that the nullablility of derivative \deriv{\trace}{\R}
determines if a concrete trace \trace is accepted by~\R.
Given a prefix \Trace of an SRE \R, the nullability of
symbolic derivative \Deriv{\Trace}{\R} determines, as established by
\cref{cor:symbolic-recognition}, whether the symbolic trace \Trace is
included in \R, \ie, all runs of \Trace in \R end at an accepting
state irrespective of the interpretation of symbolic variables.
\begin{corollaryrep}\label{cor:symbolic-recognition}
Let $\R'=\Deriv{\Trace}{\R}$.
Then
\begin{enumerate*}
\item $\Trace\sqsubseteq\R$ iff $\nu(\R')$ and
\item $\Trace\sqsubseteq\neg\R$ iff $\neg\nu(\R')$.
\end{enumerate*}
\end{corollaryrep}
\begin{proof}
Consider arbitrary $\R'=\Deriv{\Trace}{\R}$.
We have the following chain of equivalence.
$\Trace\sqsubseteq\R$, iff (\cref{defn:inclusion});
$\denot{\sigma(\Trace)}\subseteq\denot{\sigma(\R)}$ for all $\sigma$, iff;
$\nu(\deriv{\trace}{\sigma(\R)})$ for all $\trace\in\denot{\sigma(\Trace)}$
for all $\sigma$, iff (\cref{defn:symbolic-derivative});
$\nu(\sigma(\R'))$ for all $\sigma$, iff (\cref{lem:symbolic-nullable});
$\nu(\R')$.
Therefore, $\Trace\sqsubseteq\R$ iff $\nu(\R')$.
In addition, $\neg\R'=\Deriv{\Trace}{\neg\R}$ (\crefpart{cor:symbolic-derivative}{negation})
and $\nu(\neg\R')=\neg\nu(\R')$.
Then $\Trace\sqsubseteq\neg\R$ iff $\neg\nu(\R')$.
\end{proof}
\noindent
Therefore, by enumerating prefixes of \R, we may sample symbolic
traces included in \R.
\begin{example}\label{ex:sample-pre}
  Consider the required context \nodenext{a}{b} of \texttt{remove}
  from \cref{ex:spec}, admitting traces where $a$ is linked to $b$ and
  stays pointing to $b$.  The prefixes of \nodenext{a}{b} include
  $\inangle{\eff{Nxt.put}\ \var{a}\ \var{b}}\cdot
  (\lneg\inangle{\eff{Nxt.put}\ \var{a}\ \cancel{\var{b}}})^n$ and
  $\inangle{\eff{Nxt.put}\ \var{a}\ \var{b}}\cdot
  \inangle{\eff{Nxt.put}\ \var{a}\ \cancel{\var{b}}}\cdot
  \inangle{\eff{Nxt.put}\ \var{a}\ \var{b}}\cdot
  (\lneg\inangle{\eff{Nxt.put}\ \var{a}\ \cancel{\var{b}}})^n$ for any
  number $n$ of repetitions, which all lead to the same symbolic derivative:
  \[
    (\lneg\inangle{\eff{Nxt.put}\ a\ \cancel{b}})^*\vee
    (\inangle{\eff{Nxt.put}\ a\ \cancel{b}}\cdot\R_{a\curvearrowright b})
  \]
  admitting traces where either no subsequent event invalidates the
  link between $a$ and $b$, or $a$ is linked to $b$ again after being
  unlinked.  The symbolic derivative is nullable because its first
  disjunct is a Kleene Star.  Therefore, all these prefixes are
  included in \nodenext{a}{b}. \qed
\end{example}

\begin{toappendix}
\begin{corollaryrep}\label{cor:acceptance}
$\trace\in\denot{\sigma(\R)}$ iff
there exists $\Trace$ such that
$\trace\in\denot{\sigma(\Trace)}$ and
$\nu(\Deriv{\Trace}{\R})$.
\end{corollaryrep}
\begin{proof}
We prove two directions separately:
\begin{enumerate}
\item Let $\trace\in\denot{\sigma(\R)}$ and
$\sigma(\R')=\deriv{\trace}{\sigma(\R)}$.
Then $\nu(\sigma(\R'))$, thus $\nu(\R')$.
Let $\denot{\sigma(\Trace)}=\{\trace\}$.
Then by \cref{defn:symbolic-derivative},
$\Deriv{\Trace}{\R}=\R'$.
\item We have $\trace\in\denot{\sigma(\Trace)}$
and $\nu(\Deriv{\Trace}{\R})$.
By \cref{defn:symbolic-derivative},
$\deriv{\trace}{\sigma(\R)}=\sigma(\Deriv{\Trace}{\R})$
hence $\nu(\deriv{\trace}{\sigma(\R)})$.
That is $\sigma\in\denot{\sigma(\R)}$.
\qedhere
\end{enumerate}
\end{proof}
\begin{corollaryrep}\label{cor:symbolic-derivative}
$\R'=\Deriv{\Trace}{\R}$ if any of following conditions hold:
\begin{enumerate}[(1)]
\item\label{cancellation}
$\R\equiv\Trace\cdot\R'$;
\item\label{weakening}
there exists $\Trace'$ such that $\Trace\sqsubseteq\Trace'$
and $\R'=\Deriv{\Trace'}{\R}$;
\item\label{negation}
there exists $\R'\equiv\neg\R$ and $\R''\equiv\neg\R'$
such that $\R''=\Deriv{\Trace}{\R'}$, and;
\item\label{concatenation}
there exists $\Trace=\Trace_1\cdot\Trace_2$,
and $\R_1=\Deriv{\Trace_1}{\R}$
such that $\R'=\Deriv{\Trace_2}{\R_1}$.
\end{enumerate}
\end{corollaryrep}
\begin{proof}
Consider arbitrary $\sigma$ and $\trace\in\denot{\sigma(\Trace)}$.
By \cref{defn:symbolic-derivative}, to show
$\R'=\Deriv{\Trace}{\R}$, it is sufficient to prove
$\deriv{\trace}{\sigma(\R)}=\sigma(\R')$, %
which we derive from each condition:
\begin{enumerate}[(1)]
\item By assumption, $\R\equiv\Trace\cdot\R'$. %
By \cref{defn:equivalence},
$\denot{\sigma(\R)}= \denot{\sigma(\Trace\cdot\R')}=
\denot{\sigma(\Trace)}\cdot\denot{\sigma(\R')}$.  That is,
$\deriv{\trace}{\sigma(\R)}=\sigma(\R')$.
\item By assumption, $\Trace\sqsubseteq\Trace'$ and
$\R'=\Deriv{\Trace'}{\R}$. %
Consider arbitrary $\sigma$ and $\trace\in\denot{\sigma(\Trace)}$.
By \cref{defn:inclusion}, $\trace\in\denot{\sigma(\Trace')}$. %
By \cref{defn:symbolic-derivative},
$\deriv{\trace}{\sigma(\R)}=\sigma(\R')$. %
\item By assumption, $\R'\equiv\neg\R$, $\R''\equiv\neg\R'$, and
$\R''=\Deriv{\Trace}{\R'}$. %
By \cref{defn:equivalence},
$\denot{\sigma(\R')}=\denot{\neg\sigma(\R)}$ and
$\denot{\sigma(\R'')}=\denot{\neg\sigma(\R')}$. %
By \cref{defn:symbolic-derivative},
$\deriv{\trace}{\sigma(\R')}=\sigma(\R'')$. %
Then $\deriv{\trace}{\neg\sigma(\R)}=\neg\sigma(\R')$. %
That is $\deriv{\trace}{\sigma(\R)}=\sigma(\R')$.
\item By assumption, $\Trace=\Trace_1\cdot\Trace_2$,
$\R_1=\Deriv{\Trace_1}{\R}$, and $\R'=\Deriv{\Trace_2}{\R_1}$.
Let $\trace=\trace_1\cdot\trace_2$ such that
$\trace_1\in\denot{\sigma(\Trace_1)}$ and
$\trace_2\in\denot{\sigma(\Trace_2)}$. %
By \cref{defn:symbolic-derivative},
$\deriv{\trace_1}{\sigma(\R)} = \sigma(\R_1)$
$\deriv{\trace_2}{\sigma(\R_1)} = \sigma(\R')$. %
That is $\deriv{\trace}{\sigma(\R)}=\sigma(\R')$.
\end{enumerate}
\end{proof}
\end{toappendix}

\subsection{Derivative-based Semantics}

\begin{toappendix}
\begin{figure}[h]
\include{deriv-semantics-rules}
\caption{Derivative-based semantics.}
\end{figure}
\end{toappendix}

\begin{figure}[tb]
\begin{mathpar}
\Infer{DAdmit}{
\Trace[past]\sqsubseteq\R[past] \\
\Trace' \equiv \Trace \wedge \Trace[past] \\
}{
\left( \Phi, \Trace, \R[cont], \admit{\R[past]} \right) \hookrightarrowd
\left( \Phi, \Trace', \R[cont], () \right)
}

\Infer{DAppend}{
\Trace[eff]\sqsubseteq\R[eff] \\
\Trace[new] \equiv \Trace[eff] \wedge \Trace[prefix] \\
}{
\left(\Phi, \Trace, \R[cont], \append{\R[eff]} \right) \hookrightarrowd
\left(\Phi, \Trace\cdot\Trace[new], \Deriv{\Trace[prefix]}{\R[cont]}, ()\right)
}
\end{mathpar}
\caption{Selected rules of derivatve-based semantics.}
\label{fig:deriv-semantics}
\end{figure}

Now, we facilitate symbolic execution with symbolic derivatives.  As
\admit{\R[past]} and \append{\R[eff]} are the only two constructs
that augment contexts in symbolic states, we only present their
reduction rules in \cref{fig:deriv-semantics}, exhibiting the complete
set of rules in the supplementary material.  In contrast to the
na\"{\i}ve semantics, a derivative-based semantics begins symbolic
execution with the postcondition, denoted by \R[post].  Recall from
\cref{sec:naive}, \R[post] is the concatenation of the required
context and the expected effect attached to the ADT method to be
falsified.  Effectively, \R[post] predicts that the context will be
set up before calling the method, and the execution of the method
complies with its specified effect.  During symbolic execution,
we maintain the continuation effect \R[cont]
such that it precisely predicts the safe traces to be produced as 
execution continues.
\begin{example}\label{ex:init}
Consider the harness program $e_{harness}$ from \cref{ex:harness}.
Regarding the trailing \kwd{affirm} as a postcondition to be
affirmed upon finishing each execution path, we assume some symbolic
variables $a$ and $b$ and discharge \kwd{affirm} from $e_{harness}$
as part of the symbolic state.  In a derivative-based semantics, the
initial symbolic state is
\[
(\true, \varepsilon, \nodenext{a}{b}\cdot\nodebidir{a}{b},
\letin{hd,u}{\?[node],\?[elem]}\
\append{\nodenext{a}{b}};\
\texttt{remove}\ hd\ u) %
\]
where $hd$ and $u$ will then immediately be replaced by symbolic variables with the
same name for demonstration's purposes.
\qed
\end{example}

Rule \textsc{DAdmit} describes the semantics of \admit{\R[past]}.
Given a symbolic trace \Trace, a sequence of symbolic events, as an
underapproximation of what has happened so far following the current
execution, \admit{\R[past]} imposes constraint on \Trace, also in an
underapproximated fashion.  The underapproximation of \R[past] can be
found by sampling symbolic traces $\Trace[past]\sqsubseteq\R[past]$
via symbolic derivatives.  Then the execution is forked on each
\Trace[past] and its conjunction $\Trace'$ with \Trace.  
Intuitively, the conjunction $\Trace'$ is the pairwise conjunction
of events in \Trace and \Trace[past] 
(see \cref{sec:algo} for details).  A
straightforward pruning strategy then is to discard \Trace[past] with
\begin{enumerate*}
\item a different number of events than \Trace or
\item an event associated with a different effectful function than 
the corresponding event in \Trace.
\end{enumerate*}
In both cases, the conjunction $\Trace'$ trivially denotes an empty set.
We describe such \Trace[past] as \emph{incompatible} with \Trace.
Note that we deliberately exclude from the compatibility check the
consistency check between qualifiers of paired symbolic events to
avoid generating an excessive number of SMT queries.

\begin{example}\label{ex:admit-library-pre}
Consider the naive symbolic state prior to calling \eff{Val.get} on
$n_0$ from \cref{ex:naive},
$(n_0\neq\texttt{null},\nodenext{a}{b},
\letin{u'}{
\admit{\nodeelem{n_0}{u_0}};\
\append{\inangle{u_0{\gets}\eff{Val.get}\ n_0}};\
u_0}\ \dots)$.
A derivative-based symbolic state that underapproximates this naive
state is
$(n_0\neq\texttt{null},\nodenextunder{a}{b}, \nodebidir{a}{b}, \dots)$,
where
$\nodenextunder{a}{b}\defeq
\inangle{\eff{Nxt.put}\ \var{a}\ \var{b}}\cdot
(\lneg\inangle{\eff{Nxt.put}\ \var{a}\ \cancel{\var{b}}})^3
\sqsubseteq\nodenext{a}{b}$
as shown in \cref{ex:sample-pre}, and \nodebidir{a}{b} is the
continuation effect prior to the call as we will discuss shortly in
\cref{ex:append-client-pre}.
The \admit{\nodeelem{n_0}{u_0}} operation enforces the required context of the call to
\eff{Val.get}.  A symbolic trace that underapproximates
\nodeelem{n_0}{u_0} and is compatible with the current context
\nodenextunder{a}{b} is
$\nodeelemunder{n_0}{u_0}\defeq
\lneg\inangle{\eff{Val.put}\ n_0\ u_0}\cdot
\inangle{\eff{Val.put}\ n_0\ u_0}\cdot
(\lneg\inangle{\eff{Val.put}\ n_0\ \cancel{u_0}})^2$.
Augmented by \nodeelemunder{n_0}{u_0}, the context in the symbolic
state then becomes
$\inangle{\eff{Nxt.put}\ \var{a}\ \var{b}}\cdot
\inangle{\eff{Val.put}\ n_0\ u_0}\cdot
(\lneg\inangle{\eff{Val.put}\ n_0\ \cancel{u_0}}\lconj
\lneg\inangle{\eff{Nxt.put}\ \var{a}\ \cancel{\var{b}}})^2$.
\qed  
\end{example}

Rule \textsc{DAppend} describes the semantics of \append{\R[eff]},
where new events are to be appended to the current symbolic trace
\Trace.  First, we underapproximate events to be produced by
\append{\R[eff]}, again by sampling symbolic traces
$\Trace[eff]\sqsubseteq\R[eff]$.  Second, we enumerate prefixes
\Trace[prefix] of \R[cont] that are compatible with \Trace[eff], along
with the symbolic derivative \Deriv{\Trace[prefix]}{\R[cont]}.  The
execution can be forked for each pair of \Trace[eff] and
\Trace[prefix].  Recall that \R[cont] imposes constraints on the
events produced during symbolic execution, including those produced by
\kwd{append}.  As long as the behavior of \kwd{append}, in this case,
the underapproximation \Trace[eff], complies with the constraints
imposed by \Trace[prefix], \Deriv{\Trace[prefix]}{\R[cont]} represents
the constraint on events to be produced after \kwd{append} and thus
can safely replace \R[cont] in the next symbolic state.  To enforce
this compliance, we take the conjunction of \Trace[eff] and
\Trace[prefix] and append the result \Trace[new] to the current
symbolic path \Trace.  Effectively, we relate an underapproximated
behavior of \kwd{append} with the postcondition of the method to be
falsified, from which \R[cont] is derived, and track this relation in
the symbolic state.

\begin{example}\label{ex:append-client-pre}
Consider the initial symbolic state from \cref{ex:init}.  The
\kwd{append} construct requires the traces of past events to be
admissible to its argument \nodenext{a}{b} before calling
\texttt{remove}.  Then \texttt{remove} can be called in a context
represented by any symbolic trace
$\nodenextunder{a}{b}\sqsubseteq\nodenext{a}{b}$.  Furthermore, each
such \nodenextunder{a}{b} is also a prefix of the postcondition
$\nodenext{a}{b}\cdot\nodebidir{a}{b}$.  The symbolic derivative
$\Deriv[r]{\nodenextunder{a}{b}}{\nodenext{a}{b}\cdot\nodebidir{a}{b}}=
\nodebidir{a}{b}$
becomes the continuation effect after evaluating the \kwd{append} operation.
The symbolic state prior to calling \texttt{remove} is
$(\true,\varepsilon,\nodebidir{a}{b},e_\texttt{remove}\ n_0\ u)$.
\qed
\end{example}

\begin{example}\label{ex:append-library-post}
Continuing from \cref{ex:admit-library-pre}, the symbolic state
after reducing the \kwd{admit} is
$(n_0\neq\texttt{null},
\nodenextunder{a}{b}\wedge\nodeelemunder{n_0}{u_0},
\nodebidir{a}{b}, \append{\inangle{u_0{\gets}\eff{Val.get}\ n_0}})$.
The \kwd{append} construct records the call to \eff{Val.get} and
appends a singleton event \inangle{u_0{\gets}\eff{Val.get}\ n_0} to
the context.  Correspondingly, the continuation effect
\nodebidir{a}{b} is updated by its symbolic derivative over the
\eff{Val.get} event, \Deriv{\inangle{u_0{\gets}\eff{Val.get}\
n_0}}{\nodebidir{a}{b}}, which is \nodebidir{a}{b} itself as shown in
\cref{ex:symbolic-derivative}.  \qed
\end{example}

Now we leverage derivative-based semantics to falsify program $e$ with
respect to the postcondition \R[post] on the symbolic trace produced
by $e$ and show that the falsification is sound.  Consider an
execution that is recorded by a
reduction from the initial symbolic state to some final symbolic
state,
$(\true,\varepsilon,\R[post],e)\hookrightarrowd^*(\Phi,\Trace,\R[cont],v)$.
This execution is falsified by \R[post] if the final state is
reachable, \ie, $\sat{\Phi,\Trace}$, and its continuation effect
is not nullable, \ie, $\neg\nu(\R[cont])$.  The soundness of
the falsification relies on two key properties of a derivative-based
semantics: \textcircled{1} the continuation effect is properly updated
to denote future traces that are safe to
produce as the execution continues, as established by
\cref{lem:sensible-deriv-singlestep}.
\begin{lemmarep}\label{lem:sensible-deriv-singlestep}
If $\R[cont]{=}\Deriv{\Trace}{\R[post]}$
and $(\Phi,\Trace,\R[cont],e){\hookrightarrowd}(\Phi',\Trace',\R[cont]',e')$
then $\R[cont]'{=}\Deriv{\Trace'}{\R[post]}$.
\end{lemmarep}
\begin{proof}
By assumption, $\R[cont]=\Deriv{\Trace}{\R[post]}$.
By rule induction on
$(\Phi, \Trace, \R[cont], e)\hookrightarrowd(\Phi', \Trace',\allowbreak \R[cont]', e')$,
the cases where $\Trace\neq\Trace'$ and/or $\R[cont]\neq\R[cont]'$ are:
\begin{enumerate}[label=Rule , wide=0pt, font=\bfseries]
\item\textsc{DAdmit}.
By assumption, $\R[cont]=\R[cont]'$
and $\Trace'\equiv\Trace\wedge\Trace[past]$.
By \cref{defn:inclusion}, $\Trace'\sqsubseteq\Trace$.
By \crefpart{cor:symbolic-derivative}{weakening},
$\R[cont]'=\Deriv{\Trace'}{\R[post]}$.
\item\textsc{DAppend}.
By assumption,
$\Trace'=\Trace\cdot\Trace[new]$,
$\Trace[new]\equiv\Trace[eff]\wedge\Trace[prefix]$, and
$\R[cont]'=\Deriv{\Trace[prefix]}{\R[cont]}$.
By \cref{defn:inclusion},
$\Trace[new]\sqsubseteq\Trace[prefix]$.
By \crefpart{cor:symbolic-derivative}{concatenation},
$\R[cont]'
=\Deriv{\Trace[new]}{\R[cont]}
=\Deriv{\Trace[new]}{\Deriv{\Trace}{\R[post]}}
=\Deriv{\Trace'}{\R[post]}$.
\item\textsc{DLetExp}.
By assumption, $e=\letin{x}{e_1}\ e_2$, $e'=\letin{x}{e_1'}\ e_2$,
and $(\Phi,\Trace,\R[cont],e_1)\hookrightarrowd(\Phi',\Trace',\R[cont]',e_1')$.
By IH, $\R[cont]'=\Deriv{\Trace'}{\R[post]}$.
\qedhere
\end{enumerate}
\end{proof}
\begin{toappendix}
\begin{corollary}\label{cor:sensible-deriv-multistep}
If $\R[cont]{=}\Deriv{\Trace}{\R[post]}$ and
$(\Phi,\Trace,\R[cont],e) {\hookrightarrowd^*}
  (\Phi',\Trace',\R[cont]',e')$ then
  $\R[cont]'{=}\Deriv{\Trace'}{\R[post]}$.
\end{corollary}
\begin{proof}
  By assumption, $\R[cont]=\Deriv{\Trace}{\R[post]}$. %
  By induction on transitive closure of %
  $(\Phi,\Trace,\R[cont],e) {\hookrightarrowd^*}\allowbreak
  (\Phi',\Trace',\R[cont]',e')$.
  \begin{enumerate}
  \item Base Case: Follows directly from
    \cref{lem:sensible-deriv-singlestep}.
  \item Inductive Case: %
    By assumption,
    $(\Phi,\Trace,\R[cont],e) \hookrightarrowd^*
    (\Phi^\circ,\Trace^\circ,\R[cont]^\circ,e^\circ)$ and
    $(\Phi^\circ,\Trace^\circ,\R[cont]^\circ,e^\circ)\allowbreak
    \hookrightarrowd^* (\Phi',\Trace',\R[cont]',e')$. %
    By IH, $\R[cont]^\circ=\Deriv{\Trace^\circ}{\R[post]}$
    and $\R[cont]'=\Deriv{\Trace'}{\R[post]}$.
  \end{enumerate}
\end{proof}
\end{toappendix}

\begin{toappendix}
\begin{lemma}\label{lem:decreasing-singlestep}
  If $(\Phi,\Trace,\R,e) \hookrightarrowd (\Phi',\Trace',\R',e')$
  then $\denot{\Phi'}\subseteq\denot{\Phi}$.
\end{lemma}
\begin{proof}
  By rule induction on
  $(\Phi,\Trace,\R,e) \hookrightarrowd (\Phi',\Trace',\R',e')$, the
  only case where $\Phi\neq\Phi'$ is Rule \textsc{DAssume}.  By
  assumption, $\Phi'=\Phi\wedge\phi$.  By \cref{defn:denote-formula},
  $\denot{\Phi'}\subseteq\denot{\Phi}$.
\end{proof}
\begin{corollary}\label{cor:decreasing-multistep}
  If $(\Phi,\Trace,\R,e) \hookrightarrowd^* (\Phi',\Trace',\R',e')$
  then $\denot{\Phi'}\subseteq\denot{\Phi}$.
\end{corollary}
\begin{proof}
  By induction on transitive closure of
  $(\Phi,\Trace,\R,e) \hookrightarrowd (\Phi',\Trace',\R',e')$. %
  The base case follows directly from
  \cref{lem:decreasing-singlestep}.  In inductive case,
  $(\Phi,\Trace,\R,e) \hookrightarrowd
  (\Phi^\circ,\Trace^\circ,\R^\circ,e^\circ)$ and
  $(\Phi^\circ,\Trace^\circ,\R^\circ,e^\circ) \hookrightarrowd
  (\Phi',\Trace',\R',e')$.  By IH,
  $\denot{\Phi'}\subseteq\denot{\Phi^\circ}\subseteq\denot{\Phi}$.
\end{proof}
\end{toappendix}
\noindent
That is, \R[cont] in the final symbolic state, given that \Trace
records past events, correctly predicts future events to be produced
in compliance with the postcondition \R[post], \ie,
$\R[cont]=\Deriv{\Trace}{\R[post]}$.  Then $\neg\nu(\R[cont])$
suggests that without new events being produced, \Trace fails to
comply with \R[post], \ie, $\Trace\sqsubseteq\neg\R[post]$ by
\cref{cor:symbolic-recognition}.  Because the execution stops at value
$v$ and no more events are to be produced, the execution is indeed
falsified by \R[post].  \textcircled{2} Execution, 
including a falsified one,
underapproximate those of the non-derivative based na\"{\i}ve semantics, as established by
\cref{lem:monotonicity-singlestep}.
\begin{lemmarep}\label{lem:monotonicity-singlestep}
  If $\Trace \sqsubseteq \R[curr]$ and
  $(\Phi, \Trace, \R[cont], e) \hookrightarrowd (\Phi', \Trace',
  \R[cont]', e')$ then there exists $\R[curr]'$ such that
  $\Trace' \sqsubseteq \R[curr]'$ and
  $(\Phi, \R[curr], e) \hookrightarrow (\Phi', \R[curr]', e')$.
\end{lemmarep}
\begin{proof}
  By assumption, $\Trace\sqsubseteq\R[curr]$.
  By rule induction on 
  $(\Phi, \Trace, \R[cont], e) \hookrightarrowd (\Phi', \Trace', \R[cont]', e')$,
  the cases where $\Trace\neq\Trace'$ and/or $\R[cont]\neq\R[cont]'$ are:
  \begin{enumerate}[label=Rule , wide=0pt, font=\bfseries]
  \item\textsc{DAdmit}.
    By assumption,
    $e=\admit{\R[past]}$, $e'=()$,
    $\Phi=\Phi'$, $\R[cont]=\R[cont]'$,
    $\Trace[past]\sqsubseteq\R[past]$, and
    $\Trace'\equiv\Trace\wedge\Trace[past]$.
    Let $\R[curr]'=\R[curr]\wedge\Trace[past]$.
    Then by \cref{defn:inclusion}, $\Trace'\sqsubseteq\R[curr]'$.
    By Rule \textsc{Admit},
    $(\Phi,\R[curr],\admit{\R[past]})\hookrightarrow(\Phi,\R[curr]',())$.
    $\R[curr]'$ is a satisfying witness to the conclusion.
  \item\textsc{DAppend}.
    By assumption,
    $e=\append{\R[eff]}$, $e'=()$,
    $\Phi=\Phi'$, $\R[cont]'=\Deriv{\Trace[prefix]}{\R[cont]}$,
    $\Trace[eff]\sqsubseteq\R[eff]$,
    $\Trace[new]\equiv\Trace[eff]\wedge\Trace[prefix]$, and
    $\Trace'=\Trace\cdot\Trace[new]$.
    By \cref{defn:inclusion},
    $\Trace[new]\sqsubseteq\Trace[eff]\sqsubseteq\R[eff]$.
    Let $\R[curr]'=\R[curr]\cdot\R[eff]$.
    Then again by \cref{defn:inclusion},
    $\Trace'\sqsubseteq\R[curr]'$.
    By Rule \textsc{Append},
    $(\Phi,\R[curr],\append{\R[eff]})\hookrightarrow(\Phi,\R[curr]', e')$.
    $\R[curr]'$ is a satisfying witness to the conclusion.
  \item\textsc{DLetExp}.
    By assumption, $e=\letin{x}{e_1}\ e_2$, $e'=\letin{x}{e_1'}\ e_2$,
    and $(\Phi,\Trace,\R[cont],e_1)\hookrightarrowd(\Phi',\Trace',\R[cont]',e_1')$.
    By IH, there exists $\R[curr]'$ such that
    $\Trace'\sqsubseteq\R[curr]'$ and
    $(\Phi,\R[curr],e_1)\hookrightarrow(\Phi',\R[curr]',e_1')$.
    By Rule \textsc{LetExp},
    $(\Phi,\R[curr],e)\hookrightarrow(\Phi',\R[curr]',e')$.
    $\R[curr]'$ is a satsifying witness to the conclusion.
    \qedhere
  \end{enumerate}
\end{proof}
\begin{toappendix}
\begin{corollary}\label{cor:monotonicity-multistep}
  If $\Trace \sqsubseteq \R[curr]$ and
  $(\Phi, \Trace, \R[cont], e) \hookrightarrowd^* (\Phi', \Trace', \R[cont]', e')$
  then there exists $\R[curr]'$ such that
  $\Trace' \sqsubseteq \R[curr]'$ and
  $(\Phi, \R[curr], e) \hookrightarrow^* (\Phi', \R[curr]', e')$.
\end{corollary}
\begin{proof}
  By assumption, $\Trace\sqsubseteq\R[curr]$.
  By induction on transitive closure of
  $(\Phi,\Trace,\R[cont],e) \hookrightarrowd^* (\Phi',\Trace',\R[cont]',e')$.
  \begin{enumerate}
  \item Base Case: Follows directly from \cref{lem:monotonicity-singlestep}.
  \item Inductive Case: %
    By assumption,
    $(\Phi,\Trace,\R[cont],e) {\hookrightarrowd^*}
    (\Phi^\circ,\Trace^\circ,\R[cont]^\circ,e^\circ)$ and
    $(\Phi^\circ,\Trace^\circ,\R[cont]^\circ,e^\circ)\allowbreak
    \hookrightarrowd^* (\Phi',\Trace',\R[cont]',e')$. %
    By IH, we have $\R[curr]^\circ$ such that
    $\Trace^\circ\sqsubseteq\R[curr]^\circ$ and
    $(\Phi,\R[curr],e)\hookrightarrow^*(\Phi^\circ,\R[curr]^\circ,e^\circ)$.
    Again by IH, we have $\R[curr]'$ such that
    $\Trace'\sqsubseteq\R[curr]'$ and
    $(\Phi^\circ,\R[curr]^\circ,e^\circ) \hookrightarrow^*
    (\Phi',\R[curr]',e')$. %
    By transitivity of $\hookrightarrow^*$,
    $(\Phi,\R[curr],e) \hookrightarrow^* (\Phi',\R[curr]',e')$. \qedhere
  \end{enumerate}
\end{proof}
\end{toappendix}
\noindent
That is, there exists an execution paths
$(\true,\varepsilon,e)\hookrightarrow^*(\Phi,\R[curr],v)$ in the na\"{\i}ve
semantics such that $\Trace\sqsubseteq\R[curr]$.  As all traces
denoted by \Trace fail to comply with \R[post], there exist some
trace in \R[curr] that fails to comply with \R[post].  We conclude
with \cref{thm:soundness}, establishing that given a falsified
execution in derivative-based semantics, there exists a corresponding
execution in naive semantics that overapproximates this execution and
thus can also be falsified.
\begin{theoremrep}[Soundness of $\hookrightarrowd^*$]\label{thm:soundness}
  Assume $\sat{\Phi, \Trace}$. %
  If $(\true, \varepsilon, \R[post], e) \hookrightarrowd^* (\Phi, \Trace, \R[cont], v)$
  and $\neg\nu(\R[cont])$
  then $e$ is falsified against \R[post].
\end{theoremrep}
\begin{proofsketch}\leavevmode
  \begin{enumerate}
  \item First, the single-step reduction in
    \cref{lem:sensible-deriv-singlestep} can be extended to
    multi-step.  Since $\R[post]=\Deriv{\varepsilon}{\R[post]}$, we have
    $\R[cont]=\Deriv{\Trace}{\R[post]}$.
  \item Then, by \cref{cor:symbolic-recognition} on
    $\neg\nu(\R[cont])$, we have $\Trace\sqsubseteq\neg\R[post]$.
  \item Additionally, the single-step reduction in
    \cref{lem:monotonicity-singlestep} can also be extended to
    multi-step.  Since $\varepsilon\sqsubseteq\varepsilon$, there exists
    \R[curr] such that $\Trace\sqsubseteq\R[curr]$ and
    $(\true,\varepsilon,e) \hookrightarrow^* (\Phi,\R[curr],v)$.
  \item Now that $\Trace\sqsubseteq\R[curr]\wedge\neg\R[post]$ and
    $\sat{\Phi,\Trace}$, we have
    $\sat{\Phi,\R[curr]\wedge\neg\R[post]}$.
  \item Lastly, by \cref{defn:naive-falsification}, $e$ is falsified
    against \R[post].\qedhere
  \end{enumerate}
\end{proofsketch}
\begin{proof}
  By \crefpart{cor:symbolic-derivative}{cancellation},
  $\R[post] = \Deriv{\varepsilon}{\R[post]}$.
  By \cref{cor:sensible-deriv-multistep},
  $\R[cont]=\Deriv{\Trace}{\R[post]}$.
  By \cref{cor:symbolic-recognition},
  $\Trace\sqsubseteq\neg\R[post]$.
  By \cref{cor:monotonicity-multistep},
  there exists $\R[curr]$ such that $\Trace \sqsubseteq \R[curr]$ and
  $(\true, \varepsilon, e) \hookrightarrow^* (\Phi, \R[curr], v)$.
  By \cref{defn:inclusion} and \cref{defn:reachability},
  $\sat{\Phi, \R[curr] \wedge \neg\R[post]}$.
  By \cref{defn:naive-falsification},
  $e$ is falsified against \R[post].
\end{proof}
\begin{example}\label{ex:falsify}
  Consider the execution path from \cref{ex:naive}.  Through calls to
  \eff{Nxt.get} and \eff{Val.get}, the execution iterates over two
  nodes, $n_0$ and $n_1$, of the given linked list before finding a
  node storing element $u$, \ie, $n_1$.  Then $n_1$ is removed by
  linking $n_0$ to its successor \ie, $n_2$.  Following
  \cref{ex:admit-library-pre,ex:append-library-post}, the execution
  before the removal can be manifested in a derivative-based symbolic
  state:
  $(\Phi[bad], \Trace[prestate], \nodebidir{a}{b},
  \eff{Nxt.put}\,n_0\,n_2;\,n_0)$, where the path condition is
  $\Phi[bad]$ from \cref{ex:naive} and
  \begin{gather*}
    \Trace[prestate]\defeq
    \inangle{\eff{Nxt.put}\,key\,val\mid key{=}a{=}n_1{\wedge} val{=}b{=}n_2}
    {\cdot}
    \inangle{\eff{Val.put}\,n_0\,u_0}{\cdot}
    \inangle{\eff{Nxt.put}\,n_0\,n_1}{\cdot}
    \inangle{\eff{Val.put}\,n_1\,u_1}\\
    {\cdot}
    \inangle{u_0{\gets}\eff{Val.get}\,n_0}{\cdot}
    \inangle{n_1{\gets}\eff{Nxt.get}\,n_0}{\cdot}
    \inangle{u_1{\gets}\eff{Val.get}\,n_1}{\cdot}
    \inangle{u_2{\gets}\eff{Nxt.get}\,n_2}
  \end{gather*}
  To relate the event \inangle{\eff{Nxt.put}\,n_0\,n_2} with
  \nodebidir{a}{b}, we consider \nodebidir{a}{b}'s next event
  \inangle{\eff{Nxt.put}\,\cancel{a}\,b}, leading to a symbolic
  derivative of $\varnothing$ as shown in
  \cref{ex:symbolic-derivative}.  Hence, the conjunction between
  \inangle{\eff{Nxt.put}\,n_0\,n_2} and
  \inangle{\eff{Nxt.put}\,\cancel{a}\,b} witnesses this relation and
  is appended to the context $\Trace[prestate]$.  The symbolic state
  becomes: $(\Phi[bad],\Trace[bad],\varnothing,n_0)$, where
  \[
    \Trace[bad]\defeq\Trace[prestate]\cdot
    \inangle{\eff{Nxt.put}\,key\,val\mid %
      key=n_0\neq a \wedge val=n_2=b}
  \]
  $a=n_1$ and $b=n_2$ witnesses the reachability of the final symbolic
  state.  In combination with $\neg\nu(\varnothing)$, the execution is
  falsified.  In fact, this execution underapproximates the execution
  shown in \cref{ex:naive}, \ie, $\Trace[bad]\sqsubseteq\R[bad]$,
  which could have been falsified but proves too costly using naive
  semantics. \qed
\end{example}

Furthermore, this refined semantics guarantees  completeness with
respect to falsification.  Consider an execution in the na\"{\i}ve semantics that
is manifested by a reduction from the initial symbolic state to some
final symbolic state,
$(\true,\varepsilon,e)\hookrightarrow^*(\Phi,\R[curr],v)$.  According to
\cref{defn:naive-falsification}, the execution is falsified with
respect to the postcondition \R[post] as long as
$\sat{\Phi,\R[curr]\wedge\neg\R[post]}$ holds.  Looking backward
from the final state, it is sufficient to falsify the execution if
there exists some underapproximation of the execution, encapsulated by
a symbolic trace $\Trace[curr]\sqsubseteq\R[curr]$, and some prefix
$\Trace[prefix]\sqsubseteq\neg\R[post]$ such that
$\Trace[curr]\wedge\Trace[prefix]$ represents a viable execution.
Effectively, all compatible pairs of symbolic paths
$\Trace[curr]\sqsubseteq\R[curr]$ and prefixes \Trace[prefix] of
\R[post] are exhaustively explored by executions in a derivative-based
semantics.  \cref{lem:completeness} establishes this exhaustiveness on
each reduction step.
\begin{lemmarep}\label{lem:completeness}
  Given a safety property $\R[post]$.
  If $(\Phi,\R[curr],e)\hookrightarrow(\Phi',\R[curr]',e')$
  then for all $\Trace[curr]'\sqsubseteq\R[curr]'$,
  prefix $\Trace[prefix]'$ of $\R[post]$, and
  $\Trace'\equiv\Trace[curr]'\wedge\Trace[prefix]'$,
  there exists $\Trace[curr]\sqsubseteq\R[curr]$,
  prefix $\Trace[prefix]$ of $\R[post]$, and
  $\Trace\equiv\Trace[curr]\wedge\Trace[prefix]$
  such that
  $(\Phi,\Trace,\Deriv{\Trace[prefix]}{\R[post]},e)
  \hookrightarrowd
  (\Phi',\Trace',\Deriv{\Trace[prefix]'}{\R[post]},e')$.
\end{lemmarep}
\begin{proof}
  By rule induction on
  $(\Phi,\R[curr],e)\hookrightarrow(\Phi',\R[curr]',e')$,
  the cases where $\R[curr]\neq\R[curr]'$ are:
  \begin{enumerate}[label=Rule , wide=0pt, font=\bfseries]
  \item\textsc{Admit}.
    By assumption, $e=\admit{\R[past]}$, $e'=()$,
    and $\R[curr]'=\R[curr]\wedge\R[past]$.
    Consider arbitrary $\Trace[curr]'\sqsubseteq\R[curr]'$,
    prefix $\Trace[prefix]'$ of $\R[past]$, and
    $\Trace'\equiv\Trace[curr]'\wedge\Trace[prefix]$.
    Then there exists $\Trace[curr]\sqsubseteq\R[curr]$
    and $\Trace[past]\sqsubseteq\R[past]$ such that
    $\Trace[curr]'\equiv\Trace[curr]\wedge\Trace[past]$.
    Let $\Trace[prefix]=\Trace[prefix]'$
    and $\Trace\equiv\Trace[curr]\wedge\Trace[prefix]$.
    Then $\Trace'\equiv\Trace\wedge\Trace[past]$.
    Let $\R[cont]
    =\Deriv{\Trace[prefix]}{\R[post]}
    =\Deriv{\Trace[prefix]'}{\R[post]}$.
    By Rule \textsc{DAdmit},
    $(\Phi,\Trace,\R[cont],\admit{\R[past]})\allowbreak
    \hookrightarrowd(\Phi',\Trace',\R[cont],())$.
  \item\textsc{Append}.
    By assumption, $e=\append{\R[eff]}$, $e'=()$,
    and $\R[curr]'=\R[curr]\cdot\R[eff]$.
    Consider arbitrary $\Trace[curr]'\sqsubseteq\R[curr]'$
    prefix $\Trace[prefix]'$ of $\R[past]$, and
    $\Trace'\equiv\Trace[curr]'\wedge\Trace[prefix]'$.
    Then there exists $\Trace[curr]\sqsubseteq\R[curr]$
    and $\Trace[eff]\sqsubseteq\R[eff]$ such that
    $\Trace[curr]'\equiv\Trace[curr]\cdot\Trace[eff]$.
    Let $\Trace[prefix]$ be a prefix of $\Trace[prefix]'$
    such that $|\Trace[curr]|=|\Trace[prefix]|$ and
    $\Trace[prefix]'\equiv\Trace[prefix]\cdot\Trace[diff]$.
    Then $\Trace[prefix]$ is also a prefix of $\R[post]$
    and $\Deriv{\Trace[prefix]'}{\R[post]}=
    \Deriv[r]{\Trace[diff]}{\Deriv{\Trace[prefix]}{\R[post]}}$.
    Let $\Trace\equiv\Trace[curr]\wedge\Trace[prefix]$
    and $\Trace[new]\equiv\Trace[eff]\wedge\Trace[diff]$.
    such that $\Trace'=\Trace\cdot\Trace[new]$.
    By Rule \textsc{DAppend},
    $(\Phi,\Trace,\Deriv{\Trace[prefix]}{\R[post]},\append{\R[eff]})
    \hookrightarrowd(\Phi',\Trace',\Deriv{\Trace[prefix]'}{\R[post]},())$.
  \item\textsc{LetExp}.
    By assumption, $e=\letin{x}{e_1}\ e_2$, $e'=\letin{x}{e_1'}\ e_2$,
    and $(\Phi,\R[curr],e_1)\hookrightarrow(\Phi',\R[curr]',e_1')$.
    Consider arbitrary $\Trace[curr]'\sqsubseteq\R[curr]'$,
    prefix $\Trace[prefix]'$ of $\R[post]$, and
    $\Trace'\equiv\Trace[curr]'\wedge\Trace[prefix]'$.
    By IH, there exists $\Trace[curr]\sqsubseteq\R[curr]$,
    prefix $\Trace[prefix]$ of $\R[post]$, and
    $\Trace\equiv\Trace[curr]\wedge\Trace[prefix]$
    such that
    $(\Phi,\Trace,\Deriv{\Trace[prefix]}{\R[post]},e_1)
    \hookrightarrowd
    (\Phi',\Trace',\Deriv{\Trace[prefix]'}{\R[post]},e_1')$,
    which is, by Rule \textsc{DLetExp},
    $(\Phi,\Trace,\Deriv{\Trace[prefix]}{\R[post]},\allowbreak\letin{x}{e_1}\ e_2)
    \hookrightarrowd
    (\Phi',\Trace',\Deriv{\Trace[prefix]'}{\R[post]},\letin{x}{e_1'}\ e_2)$.
    \qedhere
  \end{enumerate}
\end{proof}
\begin{toappendix}
\begin{corollary}\label{cor:completeness}
  Given a safety property $\R[post]$.
  If $(\Phi,\R[curr],e)\hookrightarrow^*(\Phi',\R[curr]',e')$
  then for all $\Trace[curr]'\sqsubseteq\R[curr]'$,
  prefix $\Trace[prefix]'$ of $\R[post]$, and
  $\Trace'\equiv\Trace[curr]'\wedge\Trace[prefix]'$,
  there exists $\Trace[curr]\sqsubseteq\R[curr]$,
  prefix $\Trace[prefix]$ of $\R[post]$, and
  $\Trace\equiv\Trace[curr]\wedge\Trace[prefix]$
  such that
  $(\Phi,\Trace,\Deriv{\Trace[prefix]}{\R[post]},e)
  \hookrightarrowd^*
  (\Phi',\Trace',\Deriv{\Trace[prefix]'}{\R[post]},e')$.
\end{corollary}
\begin{proof}
  By induction on transitive closure of
  $(\Phi,\R[curr],e)\hookrightarrow^*(\Phi',\R[curr]',e')$,
  \begin{enumerate}
  \item Base Case: Follows directly from \cref{lem:completeness}.
  \item Inductive Case:
    By assumption,
    $(\Phi,\R[curr],e)\hookrightarrow^*(\Phi^\circ,\R[curr]^\circ,e^\circ)$ and
    $(\Phi^\circ,\R[curr]^\circ,e^\circ)\hookrightarrow^*(\Phi',\R[curr]',e')$.
    Consider arbitrary $\Trace[curr]'\sqsubseteq\R[curr]'$,
    prefix $\Trace[prefix]'$ of $\R[post]$, and
    $\Trace'\equiv\Trace[curr]'\wedge\Trace[prefix]'$.
    By IH, there exists $\Trace[curr]^\circ\sqsubseteq\R[curr]^\circ$,
    prefix $\Trace[prefix]^\circ$ of $\R[post]$, and
    $\Trace^\circ\equiv\Trace[curr]^\circ\wedge\Trace[prefix]^\circ$
    such that
    $(\Phi^\circ,\Trace^\circ,\Deriv{\Trace[prefix]^\circ}{\R[post]},e^\circ)
    \hookrightarrowd^*
    (\Phi',\Trace',\Deriv{\Trace[prefix]'}{\R[post]},e')$.
    Again by IH, there exists $\Trace[curr]\sqsubseteq\R[curr]$,
    prefix $\Trace[prefix]$ of $\R[post]$, and
    $\Trace\equiv\Trace[curr]\wedge\Trace[prefix]$
    such that
    $(\Phi,\Trace,\Deriv{\Trace[prefix]}{\R[post]},e)\allowbreak
    \hookrightarrowd^*
    (\Phi^\circ,\Trace^\circ,\Deriv{\Trace[prefix]^\circ}{\R[post]},e^\circ)$.
    By transitivity of $\hookrightarrowd$,
    $(\Phi,\Trace,\Deriv{\Trace[prefix]}{\R[post]},e)
    \hookrightarrowd^*\break
    (\Phi',\Trace',\Deriv{\Trace[prefix]'}{\R[post]},e')$.\qedhere
  \end{enumerate}
\end{proof}
\end{toappendix}
\noindent
As a result, \cref{thm:completeness} establishes that given a
falsified execution manifested using the na\"{\i}ve semantics, there exists an
underapproximating execution in a derivative-based semantics that can
also be falsified.
\begin{theoremrep}[Completeness of $\hookrightarrowd^*$]
  \label{thm:completeness}
  If $(\true, \varepsilon, e)\allowbreak
  {\hookrightarrow^*}(\Phi, \R[curr], v)$
  and $\sat{\Phi, \R[curr] {\wedge} \neg\R[post]}$\fussy
  then there exists $\Trace{\sqsubseteq}\R[curr]$
  and $\neg\nu(\R[cont])$ such that
  $(\true, \varepsilon, \R[post], e)
  {\hookrightarrowd^*}(\Phi, \Trace, \R[cont], v)$
  and $\sat{\Phi, \Trace}$.
\end{theoremrep}
\begin{proofsketch}\leavevmode
    \begin{enumerate}
  \item First, by \cref{defn:naive-falsification}, there exists \Phi
    and \R[curr] such that $\sat{\Phi,\R[curr]\land\neg\R[post]}$ and
    $(\true,\varepsilon,e) \hookrightarrow^* (\Phi,\R[curr],v)$.
  \item Then, let $\Trace[curr]\sqsubseteq\R[curr]$ and
    $\Trace[prefix]\sqsubseteq\neg\R[post]$ such that
    $\sat{\Phi,\Trace[curr]\wedge\Trace[prefix]}$.
  \item Furthermore, the single-step reduction in
    \cref{lem:completeness} can be extended to multi-step.  As a
    result,
    $(\true,\varepsilon,\R[post],e) \hookrightarrowd^*
    (\Phi,\Trace[curr]\land\Trace[prefix],\Deriv{\Trace[prefix]}{\R[post]},v)$.
  \item Lastly, we have $\neg\nu(\Deriv{\Trace[prefix]}{\R[post]})$
    from $\Trace[prefix]\sqsubseteq\lnot\R[post]$.\qedhere
  \end{enumerate}
\end{proofsketch}
\begin{proof}
  By \cref{defn:naive-falsification},
  there exists \Phi and \R[curr] such that
  $\sat{\Phi,\R[curr]\land\neg\R[post]}$ and
  $(\true, \varepsilon, e)\allowbreak\hookrightarrow^*(\Phi, \R[curr], v)$.
  By \cref{defn:reachability},
  there exists $\sigma$ and $\trace$ such that $\sigma\in\denot{\Phi}$
  and $\trace\in\denot{\sigma(\R[curr]\wedge\neg\R[post])}$.
  That is $\trace\in\denot{\sigma(\R[curr])}$
  and $\trace\notin\denot{\sigma(\R[post])}$.
  By \cref{cor:acceptance} and \cref{cor:symbolic-recognition},
  there exists $\Trace[curr]\sqsubseteq\R[curr]$
  such that $\trace\in\denot{\sigma(\Trace[curr])}$.
  Let $\Trace[prefix]=\trace$.
  Again by \cref{cor:acceptance},
  $\neg\nu(\Deriv{\Trace[prefix]}{\R[post]})$.
  Let $\Trace\equiv\Trace[curr]\wedge\Trace[prefix]$.
  Then $\Trace\sqsubseteq\Trace[curr]\sqsubseteq\R[curr]$
  and $\trace\in\denot{\sigma(\Trace)}$.
  By \cref{defn:reachability}, $\sat{\Phi,\Trace}$.
  By \cref{cor:completeness},
  there exists $\Trace[curr]'\sqsubseteq\varepsilon$
  and $\Trace'\equiv\Trace[curr]'\wedge\Trace[prefix]'$
  such that $(\top,\Trace',\Deriv{\Trace[prefix]'}{\R[post]},e)
  \hookrightarrowd^*(\Phi,\Trace,\R[cont],v)$.
  That is $\Trace[curr]'=\Trace[prefix]'=\Trace'=\varepsilon$,
  $\Deriv{\Trace[prefix]'}{\R[post]}=\R[post]$,
  and $(\top,\varepsilon,\R[post],e)
  \hookrightarrowd^*(\Phi,\Trace,\R[cont],v)$.
  Hence, $\Trace$ and $\R[cont]$ are satisfying witness to the conclusion.
\end{proof}

The completeness argument requires the symbolic execution to
exhaustively relate the events produced during execution and the safe
events required by the postcondition.  In the hope of finding a
falsified execution at the earliest, symbolic derivative enables
strategic exploration of this relationship during the symbolic
execution.  Consider an unfinished execution
$(\true,\varepsilon,\R[post],e_0)\hookrightarrowd^*(\Phi,\Trace,\R[cont],e)$.
Recall that the continuation effect \R[cont] predicts future
traces that are safe to produce if we finish the execution from the
current symbolic state $(\Phi,\Trace,\R[cont],e)$.  Hence, the
concatenation of the current symbolic trace \Trace and the continuation
effect \R[cont] gives an optimistic overapproximation of the
safe behavior of the execution when finished.  Then
$\neg\sat{\Phi,\Trace\cdot\R[cont]}$ essentially states that all
behavior is unsafe following this execution.  Therefore, without
finishing the execution, we may determine it is falsified.  In theory,
we also require that the execution can be finished in a satisfiable
state, as stated in \cref{thm:soundness-varnothing}.
\begin{theoremrep}[Soundness of $\varnothing$]
  \label{thm:soundness-varnothing}
  Assume $(\Phi, \Trace, \R[cont], e) \hookrightarrowd^* (\Phi', \Trace', \R[cont]', v)$ and $\sat{\Phi', \Trace'}$.
  If $(\true, \varepsilon, \R[post], e_0) \hookrightarrowd^* (\Phi, \Trace, \R[cont], e)$
  and $\neg\sat{\Phi,\Trace\cdot\R[cont]}$
  then $e$ is falsified against \R[post].
\end{theoremrep}
\begin{proofsketch}\leavevmode
  \begin{enumerate}
  \item First, by transitivity of $\hookrightarrowd^*$,
    $(\true,\varepsilon,\R[post],e_0)
    \hookrightarrowd^*(\Phi',\Trace',\R[cont]',v)$.
  \item Then, by \cref{thm:soundness}, it is sufficient to prove $\neg\nu(\R[cont]')$.
  \item By multi-step variant of \cref{lem:sensible-deriv-singlestep}
    on $\R[cont] = \Deriv[r]{\Trace}{\Trace \cdot \R[cont]}$,
    $\R[cont]' = \Deriv[r]{\Trace'}{\Trace \cdot \R[cont]}$.
  \item $\neg\sat{\Phi,\Trace\cdot\R[cont]}$ suggests that
    $\Trace\cdot\R[cont]$ is equivalent to $\varnothing$ under the
    path condition $\Phi$ or its refined path condition $\Phi'$.  So
    does its derivative $\R[cont]'$.  We have
    $\neg\nu(\R[cont]')$. \qedhere
  \end{enumerate}
\end{proofsketch}
\begin{proof}
  By transitivity of $\hookrightarrowd^*$,
  $(\true,\varepsilon,\R[post],e_0)
  \hookrightarrowd^*(\Phi',\Trace',\R[cont]',v)$.
  By \cref{thm:soundness},
  it is sufficient to prove $\neg\nu(\R[cont]')$.
  By \crefpart{cor:symbolic-derivative}{cancellation},
  $\R[cont] = \Deriv[r]{\Trace}{\Trace \cdot \R[cont]}$.
  By \cref{cor:sensible-deriv-multistep},
  $\R[cont]' = \Deriv[r]{\Trace'}{\Trace \cdot \R[cont]}$.
  By \cref{cor:residuality},
  $\Trace' \cdot \R[cont]' \sqsubseteq \Trace \cdot \R[cont]$.
  By \cref{defn:reachability},
  there exists $\sigma' \in \denot{\Phi'}$
  and $\trace'\in\denot{\sigma'(\Trace')}$.
  By \cref{cor:decreasing-multistep},
  $\sigma'\in\denot{\Phi'}\subseteq\denot{\Phi}$.
  By \cref{defn:inclusion},
  $\denot{\sigma'(\Trace')}\cdot\denot{\sigma'(\R[cont]')}
  =\denot{\sigma'(\Trace'\cdot\R[cont]')}
  \subseteq\denot{\sigma'(\Trace\cdot\R[cont])}
  =\emptyset$.
  Therefore, $\denot{\sigma'(\R[cont]')}=\emptyset$
  and $\neg\nu(\sigma'(\R[cont]'))$.
  By \cref{lem:symbolic-nullable}, $\neg\nu(\R[cont]')$.
\end{proof}
\noindent
However, in practice, as long as the current symbolic state is
satisfiable, \ie, $\sat{\Phi,\Trace}$, it is safe to assume that the
execution can be finished in a satisfiable symbolic state, which in
turn witnesses the falsification.  Another practical concern is that
checking $\neg\sat{\Phi,\Trace\cdot\R[cont]}$ can be expensive as
discussed in \cref{sec:naive}.  Instead, we check whether \R[cont]
is \emph{syntactically} equal to $\varnothing$, which implies
$\neg\sat{\Phi,\Trace\cdot\R[cont]}$.  If not, we continue the
execution without compromising soundness.  With a standard set of
rewriting rules, \eg, $\varnothing\vee\R\equiv\varnothing$, the
syntactic approach is effective in falsifying unfinished execution for
programs considered in \cref{sec:eval}.  In fact, \cref{ex:falsify} is
such a case -- had \texttt{remove} not stopped at the first
node found to store the given element, we can still conclude that the
execution is falsified without needing to finish iterating over the
remaining linked list.

Intuitively, we exploit the existence of a dead state in the automaton
associated with the postcondition \R[post].  As events produced
during symbolic execution are related to transitions in the automaton,
it is sufficient to falsify a execution if the events produced can be
related to transitions in the automaton that leads to a dead state.
This is similar to the recognition of a string in a deterministic
automaton, where it is sufficient to determine the string cannot be
accepted if a character causes the automata to enter a dead state.
However, each event produced may still nondeterministically be related
to transtions from the current state in the postcondition automaton.

We would like to further exploit the structure of the postcondition
automaton by \emph{actively} looking for a dead state.  Again we
consider the symbolic state of an unfinished execution, its
continuation symbolic derivative \R[cont] is a symbolic derivative of
the postcondition \R[post] and thus \R[cont] represents a state in the
automaton associated with \R[post].  The minimal distance of the
state denoted by \R[cont] to a dead state gives us a lower bound on
the number of events that the current execution needs to produce in
order to be falsified.  It is also the minimal length of \R[cont]'s
prefixes such that the derivative over them denotes a dead state, \ie,
\textsf{DistToDead}(\R[cont]), where
\[
  \textsf{DistToDead}(\R)\defeq
  \min_{\Deriv{\Trace}{\R}=\varnothing}|\Trace|
\]
When new events \Trace[eff] are produced during the execution as in
Rule \textsc{DAppend}, among all \R[cont]'s prefixes \Trace[prefix]
that are compatible with \Trace[eff], we prioritize relating
\Trace[eff] with prefixes that brings us closer to a dead state,
according to \textsf{DistToDead}(\Deriv{\Trace[prefix]}{\R[cont]}).
In practice, we set a cut-off constant to limit the depth of such
exploitation.
\begin{example}
  Consider a different execution from what is shown in
  \cref{ex:falsify}.  \inangle{\eff{Nxt.put}\,a\,\cancel{b}} is also a
  next event of \nodebidir{a}{b} but leads to a symbolic derivative of
  $\bullet^*$.  Correspondingly, the event
  \inangle{\eff{Nxt.put}\,key\,val\mid key=n_0=a\wedge val=n_2\neq b}
  is appended to the symbolic trace.  While the symbolic state happens
  to becomes unreachable ($n_2\neq b$ contradicts $n_2=b$ from
  \Trace[prestate]) and thus can be pruned, it does not have to be the
  case and nondeterministic time may be spent on this infeasible
  execution before it is pruned. \qed
\end{example}

\section{Algorithm}\label{sec:algo}

In this section, we substitute the declarative components of
derivative-based semantics with their algorithmic equivalents, thus
demonstrating the derivative-based semantics is a sound and relatively
complete procedure for falsification.

First, we show that the reachability check (\cref{defn:reachability})
of derivative-based symbolic states, \ie, $\sat{\Phi,\Trace}$, can be
straightforwardly discharged to SMT queries like conventional symbolic
execution techniques.  Intuitively, since a symbolic path \Trace is a
sequence of symbolic events, we would like to collect constraints from
each symbolic event.  The constraint of an atomic symbolic event can
be built as:
\[\TrConstr(\inangle{x_\textsf{ret}\gets\eff{f}\ \overline{x_\textsf{arg}}\mid\phi})=
  [\overline{x_\textsf{arg}\mapsto\symb{x_\textsf{arg}}},x_\textsf{ret}\mapsto\symb{x_\textsf{ret}}]\phi
  \quad\text{for fresh}\ \overline{\symb{x_\textsf{arg}}}\ \text{and}\ \symb{x_\textsf{ret}}
\]
To facilitate constraint collection, we give a
stratified representation of symbolic events $\lit$ as a disjunction
of atomic symbolic events associated with \emph{disjoint} effectful functions:
\[\textstyle
  \lit \defeq \dots\parallel
  \inangle{x_\textsf{ret}\gets\eff{f}\ \overline{x_\textsf{arg}}\mid\phi}
  \parallel\dots \qquad\text{such that}\quad
  \denot{\lit}=\bigcup_{\inangle{x_\textsf{ret}\gets\eff{f}\ \overline{x_\textsf{arg}}\mid\phi}\in\lit}
  \denot{\inangle{x_\textsf{ret}\gets\eff{f}\ \overline{x_\textsf{arg}}\mid\phi}}
\]
Since the effectful functions $\eff{f}$ associated with the disjuncts
in $\lit$ are different, the constraint of a symbolic event \lit is simply the disjunction of constraints from \lit's atomic symbolic events, and the constraint of a symbolic path \Trace is the conjunction of constraints from \Trace's symbolic events:
\begin{mathpar}\textstyle
  \TrConstr(\varepsilon){=}\true \and
  \TrConstr(\lit){=}
  \bigvee_{\inangle{\eff{f}\mid\phi}\in\lit}
  \TrConstr(\inangle{\eff{f}\mid\phi}) \and
  \TrConstr(\Trace_1{\cdot}\Trace_2){=}
  \TrConstr(\Trace_1){\wedge}\TrConstr(\Trace_2)
\end{mathpar}  
It immediately follows that, as established by
\cref{cor:reachability}, the reachability of a symbolic state can be
determined by the satisfiabiliy of the conjunction between its path
condition $\Phi$ and the constraints from its current symbolic path
$\Trace$.
\begin{corollary}\label{cor:reachability}
  $\sat{\Phi,\Trace}$ iff $\sigma\in\denot{\Phi\wedge\TrConstr(\Trace)}$.
\end{corollary}

In response to the stratified representation of symbolic events, we
discharge their boolean connectives using \cref{defn:literal-algebra},
which was part of the syntax in \cref{sec:naive}.
\begin{definition}[Events Algebra]\label{defn:literal-algebra}
  The boolean operations on symbolic events can be defined as:
\begin{align*}
  \lneg (\parallel_i \inangle{\eff{f}_i \mid \phi_i})
  =& (\parallel_i \inangle{\eff{f}_i \mid \neg\phi_i}) \parallel
     (\parallel_{\eff{g} \in \Delta / \overline{\eff{f}_i}} \inangle{\eff{g}})
  \\
  (\parallel_i \inangle{\eff{f}_i \mid \phi_i}) \sqcap
  (\parallel_j \inangle{\eff{g}_j \mid \psi_j})
  =& \parallel_{\eff{f}_i = \eff{g}_j} \inangle{\eff{f}_i \mid \phi_i \wedge \psi_j}
  \\
  (\parallel_i \inangle{\eff{f}_i \mid \phi_i}) \sqcup
  (\parallel_j \inangle{\eff{g}_j \mid \psi_j})
  =& (\parallel_{\eff{f}_i = \eff{g}_j} \inangle{\eff{f}_i \mid \phi_i \vee \psi_j}) \parallel
     (\parallel_{\eff{f}_i \notin \overline{\eff{g}_j}} \inangle{\eff{f}_i \mid \phi_i}) \parallel
     (\parallel_{\eff{g}_j \notin \overline{\eff{f}_i}} \inangle{\eff{g}_j \mid \psi_j})
\end{align*}
\end{definition}
\noindent
Again, all atomic symbolic events in $\lit$ are associated with
different effectful functions and thus are disjoint.  The negation of
$\lit$ includes atomic symbolic events from $\lit$ with their
qualifiers negated and atomic symbolic events out of $\lit$ with
$\true$ qualifier.  The conjunction of $\lit_1$ and $\lit_2$ includes
atomic symbolic events included by both $\lit_1$ and $\lit_2$ with the
qualifier being their conjunctions.  The disjunction of $\lit_1$ and
$\lit_2$ includes atomic symbolic events included by both $\lit_1$ and
$\lit_2$ with the qualifier being their disjunctions, as well as
atomic symbolic events included only in $\lit_1$ or $\lit_2$.
\cref{defn:literal-algebra} preserves the disjointness requirement in
the result and is consistent with the denotation \denot{\lit} above.

Before providing algorithms for computing prefixes and symbolic
derivatives, we first demonstrate a procedure for finding next events
of a given SRE \R by rediscovering the notion of ``next literals''
presented in \cite{keilSymbolicSolvingExtended2014}.  For
$\varnothing$ and $\varepsilon$, their next event can only be bottom.
For $\lit$, its next event is simply $\lit$ itself.  For $\R^*$, its
next events are the same as those of $\R$.  For $\neg\R$, its next
events include those of $\R$ and the complement of their disjunction.
For $\R_1\land\R_2$, its next events includes the conjunction of
events included in both $\R_1$ and $\R_2$.  For $\R_1\lor\R_2$, its next
events includes not only the conjunction of events included in both
$\R_1$ and $\R_2$, but also the conjunction of each event from $\R_1$ and
the negation of the disjunction of $\R_2$'s next events, and vice
versa, defined as a join operation $\Join$ between two sets of events.
As a result, the disjunction of $\R_1\vee\R_2$'s next events is
equivalent to the disjunction of $\R_1$'s and $\R_2$'s.  For
$\R_1\cdot\R_2$, its next events are determined by the join of those
of $\R_1$ and those of $\R_2$ if $\R_1$ is nullable. Otherwise, its next
events includes only those of $\R_2$.
\begin{definition}[Admissible Next Events]\label{defn:next}
  The set $\Lits$ of events admissible to $\R$ can be computed as:
  \begin{mathpar}
    \FAnext(\varnothing) = \FAnext(\varepsilon) = \{ \bot \}
    \and
    \FAnext(\lit) = \{ \lit \}
    \and
    \FAnext(\R^*) = \FAnext(\R)
    \and
    \FAnext(\R_1 \cdot \R_2) =
    \begin{cases}
      \FAnext(\R_1) \Join \FAnext(\R_2) & \nu(\R_1) \\
      \FAnext(\R_1) & \text{otherwise}
    \end{cases}
    \and
    \FAnext(\neg\R) = \FAnext(\R) \cup \left\{\FAnext(\R)^\complement\right\}
    \and
    \FAnext(\R_1 \wedge \R_2) = \FAnext(\R_1) \sqcap \FAnext(\R_2)
    \and
    \FAnext(\R_1 \vee \R_2) = \FAnext(\R_1) \Join \FAnext(\R_2)
  \end{mathpar}
  where the dual of an event set $\Lits$ is
  $\textstyle\Lits^\complement\defeq\lneg\bigsqcup_{\lit \in
    \Lits}\lit$  and the join of two event sets $\Lits_1$ and
  $\Lits_2$ is
  $ \Lits_1 \Join \Lits_2 \defeq \{ \lit_1 \sqcap \lit_2, %
  \lit_1 \sqcap \Lits_2^\complement, %
  \Lits_1^\complement \sqcap \lit_2 %
  \mid \lit_1 \in \Lits_1, \lit_2 \in \Lits_2 \}$.
\end{definition}
\begin{toappendix}
\begin{lemma}[Partial Equivalence\cite{keilSymbolicSolvingExtended2014}]
  \label{lem:partial-equivalence}
  \begin{enumerate*}
  \item Let $\Lits = \FAnext(\R)$.  Then $\lit\in\Lits$ is a
    singleton prefix of \R.
  \item Let $\lit = \FAnext(\R)^\complement$.  Then
    $\Deriv{\lit}{\R}=\varnothing$.
  \end{enumerate*}
\end{lemma}
\end{toappendix}
\noindent
\cref{defn:next} provides such a \FAnext operation such that each
symbolic event $\lit\in\FAnext(\R)$ is a singleton prefix of \R
(\cref{defn:prefix}).  Due to the negation rule, the disjunction of
$\FAnext(\R)$ overapproximates the set of events admissible to \R.
Then the negation of this disjunction, \ie, $\FAnext(\R)^\complement$
is also a next event of \R, the derivative over which is
$\varnothing$.  $\FAnext(\R)\cup\{\FAnext(\R)^\complement\}$ gives
us a set of symbolic events that covers the entire space of possible
events and are all amenable to symbolic derivative computation of \R.
Now, the symbolic derivative of \R over its next events can be
computed inductively in a similar fashion to \cref{sec:prelim} by the
following lemma:
\begin{lemmarep}\label{lem:symbolic-derivative}
  Given $\R$ and its prefix $\Trace$, the symbolic derivative
  $\Deriv{\Trace}{\R}$ can be computed via:
\begin{mathpar}
  \Deriv{\varepsilon}{\R} = \R
  \and
  \Deriv{\Trace_1 \cdot \Trace_2}{\R} = \Deriv{\Trace_2}{\Deriv{\Trace_1}{\R}}
  \and
  \Deriv{\lit}{\varnothing} = \Deriv{\lit}{\varepsilon} = \varnothing
  \and
  \Deriv[r]{\lit}{\R^*} = \Deriv{\lit}{\R} \cdot \R^*
  \and
  \Deriv{\lit'}{\lit} =
  \begin{cases}
    \varepsilon & \lit' \sqsubseteq \lit \\
    \varnothing & \lit' \sqsubseteq \lneg\lit
  \end{cases}
  \and
  \Deriv[r]{\lit}{\R_1 \cdot \R_2} =
  \begin{cases}
    \inparen{\Deriv{\lit}{\R_1} \cdot \R_2} \vee \Deriv{\lit}{\R_2} & \nu(\R_1) \\
    \Deriv{\lit}{\R_1} \cdot \R_2 & \text{otherwise}
  \end{cases}
  \and\newline
  \Deriv[r]{\lit}{\neg\R} = \neg \Deriv{\lit}{\R}
  \and
  \Deriv[r]{\lit}{\R_1 \land \R_2} = \Deriv{\lit}{\R_1} \wedge \Deriv{\lit}{\R_2}
  \and
  \Deriv[r]{\lit}{\R_1 \vee \R_2} = \Deriv{\lit}{\R_1} \vee \Deriv{\lit}{\R_2}
\end{mathpar}
\end{lemmarep}
\begin{proof}
  By induction on $\Deriv{\Trace}{\R} = \R'$.
  \begin{enumerate}[label=Case , wide=0pt, font=\bfseries]
  \item $\Deriv{\varepsilon}{\R}$.
    We have $\denot{\varepsilon} = \{\epsilon\}$
    and $\deriv{\varepsilon}{\sigma(\R)} = \sigma(\R)$ for all $\sigma$.
    Then by \cref{defn:symbolic-derivative},
    $\Deriv{\varepsilon}{\R} \equiv \R$.
  \item $\Deriv{\Trace_1 \cdot \Trace_2}{\R}$.
    Follows directly from \crefpart{cor:symbolic-derivative}{concatenation}.
  \item $\Deriv{\lit}{\varnothing}$.
    We have $\deriv{\alpha}{\varnothing} = \varnothing$
    for all $\alpha \in \denot{\sigma(\lit)}$ for all $\sigma$.
    By \cref{defn:symbolic-derivative},
    $\Deriv{\lit}{\varnothing} = \varnothing$.
  \item $\Deriv{\lit}{\varepsilon}$.
    We have $\deriv{\alpha}{\varepsilon} = \varnothing$
    for all $\alpha \in \denot{\sigma(\lit)}$ for all $\sigma$.
    By \cref{defn:symbolic-derivative},
    $\Deriv{\lit}{\varepsilon} = \varepsilon$.
  \item $\Deriv{\lit}{\R^*}$.
    By \cref{defn:symbolic-derivative},
    $\deriv{\alpha}{\sigma(\R)} = \sigma(\Deriv{\lit}{\R})$
    for all $\alpha \in \denot{\sigma(\lit)}$ for all $\sigma$.
    Then for all $\sigma$ for all $\alpha \in \denot{\sigma(\lit)}$,
    $\sigma(\Deriv{\lit}{\R} \cdot \R^*) = \deriv{\alpha}{\sigma(\R)} \cdot \sigma(\R)^* = \deriv{\alpha}{\sigma(\R)^*}$.
    Again by \cref{defn:symbolic-derivative},
    $\Deriv{\lit}{\R} = \Deriv{\lit}{\R} \cdot \R^*$.
  \item $\Deriv{\lit'}{\lit}$ when $\lit' \sqsubseteq \lit$.
    By \cref{defn:inclusion}, we have $\alpha \in \denot{\sigma(\lit)}$
    and $\deriv{\alpha}{\sigma(\lit)} = \varepsilon$
    for all $\alpha \in \denot{\sigma(\lit')}$ for all $\sigma$.
    By \cref{defn:symbolic-derivative},
    $\Deriv{\lit'}{\lit} = \varepsilon$.
  \item $\Deriv{\lit'}{\lit}$ when $\lit' \sqsubseteq \neg\lit$.
    By \cref{defn:inclusion}, We have $\alpha \notin \denot{\sigma(\lit)}$
    and $\deriv{\alpha}{\sigma(\lit)} = \varnothing$
    for all $\alpha \in \denot{\sigma(\lit')}$ for all $\sigma$.
    By \cref{defn:symbolic-derivative},
    $\Deriv{\lit'}{\lit} = \varnothing$.
  \item $\Deriv[r]{\lit}{\R_1 \cdot \R_2}$.
    By \cref{defn:symbolic-derivative},
    $\deriv{\alpha}{\sigma(\R_1)} = \sigma(\Deriv{\lit}{\R_1})$ and
    $\deriv{\alpha}{\sigma(\R_2)} = \sigma(\Deriv{\lit}{\R_2})$
    for all $\alpha \in \denot{\sigma(\lit)}$ for all $\sigma$.
    Then when $\nu(\R_1)$,
    for all $\sigma$ we have $\epsilon \in \denot{\sigma(\R_1)}$
    and for all $\alpha \in \denot{\sigma(\lit)}$,
    $\sigma((\Deriv{\lit}{\R_1} \cdot \R_2) \vee \Deriv{\lit}{\R_2})
    = (\deriv{\alpha}{\sigma(\R_1)} \cdot \sigma(\R_2)) \vee \deriv{\alpha}{\sigma(\R_2)}
    = \deriv{\alpha}{\sigma(\R_1 \cdot \R_2)}$,
    which is $\Deriv[r]{\lit}{\R_1\cdot\R_2}=(\Deriv{\lit}{\R_1} \cdot \R_2) \vee \Deriv{\lit}{\R_2}$ by \cref{defn:symbolic-derivative}.
    Otherwise for all $\sigma$ we have $\epsilon \notin \denot{\sigma(\R_1)}$
    and for all $\alpha \in \denot{\sigma(\lit)}$,
    $\sigma(\Deriv{\lit}{\R_1} \cdot \R_2)
    = \deriv{\alpha}{\sigma(\R_1)} \cdot \sigma(\R_2)
    = \deriv{\alpha}{\sigma(\R_1 \cdot \R_2)}$,
    which is $\Deriv[r]{\lit}{\R_1\cdot\R_2}=\Deriv{\lit}{\R_1} \cdot \R_2$ by \cref{defn:symbolic-derivative}.
  \item $\Deriv[r]{\lit}{\neg\R}$.
    Follows directly from \crefpart{cor:symbolic-derivative}{negation}.
  \item $\Deriv[r]{\lit}{\R_1\land\R_2}$.
    By \cref{defn:symbolic-derivative},
    $\deriv{\alpha}{\sigma(\R_1)} = \sigma(\Deriv{\lit}{\R_1})$ and
    $\deriv{\alpha}{\sigma(\R_2)} = \sigma(\Deriv{\lit}{\R_2})$
    for all $\alpha \in \denot{\sigma(\lit)}$ for all $\sigma$.
    Then for all $\sigma$ for all $\alpha \in \denot{\sigma(\lit)}$,
    $\sigma(\Deriv{\lit}{\R_1} \land \Deriv{\lit}{\R_2})
    = \deriv{\alpha}{\sigma(\R_1)} \land \deriv{\alpha}{\sigma(\R_2)}
    = \deriv{\alpha}{\sigma(\R_1 \land \R_2)}$.
    Again by \cref{defn:symbolic-derivative},
    $\Deriv[r]{\lit}{\R_1\land\R_2}=\Deriv{\lit}{\R_1} \land \Deriv{\lit}{\R_2}$.
  \item $\Deriv[r]{\lit}{\R_1\lor\R_2}$.
    By \cref{defn:symbolic-derivative},
    $\deriv{\alpha}{\sigma(\R_1)} = \sigma(\Deriv{\lit}{\R_1})$ and
    $\deriv{\alpha}{\sigma(\R_2)} = \sigma(\Deriv{\lit}{\R_2})$
    for all $\alpha \in \denot{\sigma(\lit)}$ for all $\sigma$.
    Then for all $\sigma$ for all $\alpha \in \denot{\sigma(\lit)}$,
    $\sigma(\Deriv{\lit}{\R_1} \lor \Deriv{\lit}{\R_2})
    = \deriv{\alpha}{\sigma(\R_1)} \lor \deriv{\alpha}{\sigma(\R_2)}
    = \deriv{\alpha}{\sigma(\R_1 \lor \R_2)}$.
    Again by \cref{defn:symbolic-derivative},
    $\Deriv[r]{\lit}{\R_1\lor\R_2}=\Deriv{\lit}{\R_1} \lor \Deriv{\lit}{\R_2}$.
    \qedhere
  \end{enumerate}
\end{proof}
\noindent
The main difference is, when computing the symbolic derivative of a
symbolic event $\lit$ over another $\lit'$, we need to perform an
inclusion check between them.  If all events denoted by $\lit'$ are
included by $\lit$, then the symbolic derivative is $\varepsilon$.  If
all events denoted by $\lit'$ are not included by $\lit$, then the
symbolic derivative is $\varnothing$.  Since $\lit'$ is guaranteed to
be a singleton prefix of $\lit$, it is impossible that some events
denoted by $\lit'$ are included by $\lit$ while some are excluded.
Thus, checking whether $\lit'\sqsubseteq\lit$ is sufficient.  The
inclusion check essentially involves checking the validity of
$\textsf{constr}(\lneg\lit'\sqcup\lit)$, which is well-suited for SMT
solving.

Using the \FAnext operation and the computation of symbolic
derivatives over symbolic events, we may enumerate prefixes of
arbitrary length from an SRE \R along with their corresponding
symbolic derivatives by following the rules in \cref{fig:prefix-enum}.
Rule \textsc{Pfx-$\varepsilon$} states that $\varepsilon$ is a prefix
of $\R$ and the corresponding symbolic derivative is $\R$ itself.
Rule \textsc{Pfx-$\lit$} states that all next events of \R is a prefix
of $\R$.  Rule \textsc{Pfx-$\cdot$} states that given any prefix
$\Trace_1$ of $\R$ along with the corresponding symbolic derivative
$\R_1$, and any prefix $\Trace_2$ of $\R_1$ along with the
corresponding symbolic derivative $\R_2$, the concatenation of
$\Trace_1$ and $\Trace_2$ is still a prefix of $\R$ with $\R_2$
being the corresponding symbolic derivative.
\begin{figure}[h]
\begin{mathpar}\small
  \Infer{Pfx-$\varepsilon$}{}{
    (\varepsilon,\R)\triangleright\R
  }

  \Infer{Pfx-$\lit$}{\textstyle
    \lit\in\FAnext(\R)\cup\left\{ \FAnext(\R)^\complement \right\}
  }{
    (\lit,\Deriv{\lit}{\R})\triangleright\R
  }

  \Infer{Pfx-$\cdot$}{
    (\Trace_1,\R_1)\triangleright\R \\
    (\Trace_2,\R_2)\triangleright\R_1 \\
  }{
    (\Trace_1\cdot\Trace_2,\R_2)\triangleright\R
  }
\end{mathpar}
\caption{Enumerate prefixes of a given \R and compute their symbolic derivatives.}
\label{fig:prefix-enum}
\end{figure}
\noindent
Intuitively, these rules allow us to construct a deterministic SFA
that accepts the same set of traces as \R and all paths in the SFA
are enumerated, including those that lead to dead states.  As
established by \cref{lem:soundness-prefix-enum}, each enumerated
prefix is indeed a prefix of $\R$.
\begin{lemmarep}[Soundness of Prefix Enumeration]\label{lem:soundness-prefix-enum}
  If $(\Trace,\R')\triangleright\R$ then $\R'=\Deriv{\Trace}{\R}$.
\end{lemmarep}
\begin{proof}
  By rule induction on $(\Trace,\R')\triangleright\R$.
  \begin{enumerate}[label=Case , wide=0pt, font=\bfseries]
  \item \textsc{Pfx-$\varepsilon$}.
    By assumption, $\Trace=\varepsilon$ and $\R'=\R$.
    By \cref{lem:symbolic-derivative}, \R'=\Deriv{\varepsilon}{\R}.
  \item \textsc{Pfx-$\lit$}.
    By assumption, $\Trace=\lit$ and $\R'=\Deriv{\lit}{\R}$.
  \item \textsc{Pfx-$\cdot$}.
    By assumption, $\Trace=\Trace_1\cdot\Trace_2$,
    $(\Trace_1,\R_1)\triangleright\R$,
    $(\Trace_2,\R_2)\triangleright\R_1$, and $\R'=\R_2$.
    By IH, $\Deriv{\Trace_1}{\R_1}=\R$
    and $\Deriv{\Trace_2}{\R_1}=\R_2$.
    By \cref{lem:symbolic-derivative},
    $\R_2=\Deriv{\Trace_2}{\Deriv{\Trace_1}{\R}}
    =\Deriv{\Trace_1\cdot\Trace_2}{\R_1}$.
  \end{enumerate}
\end{proof}
\noindent
A completeness result then states that all paths in the SFA can be
enumerated.  As an SRE \R may have different SFA representations, a
prefix \Trace of \R may not correspond to a path in the SFA
constructed by our prefix enumeration.  However, it is guaranteed
that, as established by \cref{lem:completeness-prefix-enum}, a set of
prefixes, \ie, a set of paths in the SFA, can be found by enumeration
such that their disjunction includes all traces denoted by such a
prefix $\Trace$.
\begin{toappendix}
  \begin{corollary}\label{cor:prefix-enum-fixed-length}
    Given $\R$.  We have
    $\bigvee_{(\Trace,\R')\triangleright\R, |\Trace|=n}
    \Trace\equiv\bullet^n$ for all $n\in\mathbb{N}$.
  \end{corollary}
  \begin{proof}
    By induction on $n$.
    \begin{enumerate}[label=Case ,wide=0pt,font=\bfseries]
    \item $n=0$.
      $\bigvee_{(\Trace,\R')\triangleright\R, |\Trace|=0}\Trace=\varepsilon$.
    \item $n=1$.
      $\bigvee_{(\Trace,\R')\triangleright\R, |\Trace|=1}\Trace
      =\bigvee_{\lit\in\FAnext(\R)\cup\left\{ \FAnext(\R)^\complement \right\}}\lit$
      , which is equivalent to $\bullet$ by \cref{lem:partial-equivalence}.
    \item $n>1$. Let $\R_2=\R'$
      \begin{align*}
        &\bigvee_{(\Trace,\R_2)\triangleright\R, |\Trace|=n}\Trace &&\\
        =&\bigvee_{
           n=i+j, |\Trace_1|=i, |\Trace_2|=j,
           (\Trace_1\cdot\Trace_2,\R_2)\triangleright\R
           } \Trace_1\cdot\Trace_2 &&\\
        =&\bigvee_{n=i+j, |\Trace_1|=i, (\Trace_1,\R_1)\triangleright\R}
           \bigvee_{|\Trace_2|=j, (\Trace_2,\R_2)\triangleright\R_1}
           \Trace_1\cdot\Trace_2 && \text{by Rule \textsc{Pfx-$\cdot$}} \\
        \equiv&\bigvee_{n=i+j, |\Trace_1|=i, (\Trace_1,\R_1)\triangleright\R}\Trace_1\cdot
                \bigvee_{|\Trace_2|=j, (\Trace_2,\R_2)\triangleright\R_1}\Trace_2
                                                                   && \text{by \cref{defn:equivalence}} \\
\\
        \equiv&\bigvee_{n=i+j, |\Trace_1|=i, (\Trace_1,\R_1)\triangleright\R}
           \Trace_1\cdot\bullet^j && \text{by IH} \\
        \equiv&\bigvee_{n=i+j}
           \left(\bigvee_{|\Trace_1|=i, (\Trace_1,\R_1)\triangleright\R}\Trace_1\right)
           \cdot\bullet^j && \text{by \cref{defn:equivalence}} \\
        \equiv&\bigvee_{n=i+j}\bullet^i\cdot\bullet^j && \text{by IH} \\
        \equiv&\bullet^n && \text{by \cref{defn:equivalence}}
      \end{align*}
    \end{enumerate}
  \end{proof}
\end{toappendix}
\begin{lemmarep}[Completeness of Prefix Enumeration]\label{lem:completeness-prefix-enum}
  If $\R'=\Deriv{\Trace}{\R}$
  then there exists $\overline{\Trace_i}^i$ such that
  $(\Trace_i,\R')\triangleright\R$ for all $i$ and
  $\Trace\sqsubseteq\bigvee_i\Trace_i$.
\end{lemmarep}
\begin{proof}
  By \cref{cor:prefix-enum-fixed-length}, we have
  $\bigvee_{(\Trace',\R')\triangleright\R,|\Trace'|=|\Trace|}\Trace'
  \equiv\bullet^{|\Trace|}\sqsupseteq\Trace$.
  In addition, for each $(\Trace',\R')\triangleright\R$,
  $\R'=\Deriv{\Trace'}{\R}$ by \cref{lem:soundness-prefix-enum}.
\end{proof}

Sampling symbolic traces \Trace from a given SRE \R is a special case
of enumerating prefixes whose symbolic derivative is nullable, as
shown in \cref{cor:symbolic-recognition}.  Intuitively, the sampled
symbolic traces correspond to the paths that lead to an accepting
state in the SFA.  For the purpose of sampling symbolic traces, we may
ignore paths that lead to a dead state without compromising the
completeness of sampling.  That is, when applying Rule
\textsc{PFX-\lit} for trace sampling, we ignore
$\FAnext(\R)^\complement$, whose corresponding symbolic derivative is
always $\varnothing$.

Lastly, we show how to relate symbolic traces of the same length by
computing their conjunction.  The following rules effectively perform
pairwise conjunction between symbolic events
(\cref{defn:literal-algebra}) from two symbolic traces $\Trace_1$ and
$\Trace_2$:
\begin{mathpar}
  \varepsilon\wedge\varepsilon\equiv\varepsilon
  \and
  \lit_1\wedge\lit_2\equiv\lit_1\sqcap\lit_2
  \and
  (\Trace_{11}\cdot\Trace_{12})\wedge(\Trace_{21}\cdot\Trace_{22})
  \equiv(\Trace_{11}\wedge\Trace_{21})\cdot(\Trace_{12}\wedge\Trace_{22})
\end{mathpar}
where $|\Trace_{11}|=|\Trace_{21}|$ and
$|\Trace_{12}|=|\Trace_{22}|$.  A symbolic trace equivalent to the
conjunction of $\Trace_1$ and $\Trace_2$ is returned following the
rules.

To conclude this section, the prefix enumeration algorithm gives us a
sound and relatively complete equivalent for the premises of Rules
\textsc{DAdmit} and \textsc{DAppend}.  By enumerating prefixes in
increasing length, minimal traces of events are produced and appended
along symbolic execution.

\section{Implementation and Evaluation}\label{sec:eval}
We have implemented a symbolic execution engine in \ocamllang based on
a derivative-based semantics, called \HATch that targets the
falsification of \ocamllang-like ADT implementations that interact
with their underlying representation types via API calls.  \HATch
takes as input the implementation of an ADT's method, its behavioral
specification, and the behavioral specifications of the underlying
representation types, and performs symbolic execution against an
execution harness as described in Example~\ref{ex:harness}.  Symbolic
execution is performed in increasing depth of explored execution
traces.  \HATch performs two additional optimizations that are not
discussed in \cref{sec:algo}.  First, it not only tracks the atomic
symbolic events that are included in a symbolic event \lit but also
tracks those that are excluded.  This frees us from enumerating all
other available APIs when computing the negation.  Second, since the
prefixes to be enumerated are combined with a given symbolic trace,
the enumeration of prefixes is interleaved with a compatibility check
against the trace. This interleaving helps avoid enumerating prefixes
that are known to be incompatible.

\newcommand{\propwidth}{4cm}

\begin{table}
  \centering\tiny
  \caption{Falsification of a variety of safety property violations in ADT implementations.}\label{tab:eval}
  \setlength{\tabcolsep}{2pt}.
  \begin{tabular}{m{.9cm}cp{\propwidth}|p{4.5cm}|c|rr}\toprule
    \multirow{2}{*}{\textbf{ADT}}                 & \multirow{2}{*}{\textbf{Repr. Type}} & \multirow{2}{*}{\textbf{Safety Property}} & \multirow{2}{*}{\textbf{Violation to the Safety Property}}                & \textbf{Time (s)}   & \multicolumn{2}{c}{\textbf{Speedup over}}   \\
                                                  &                                      &  &                                                                      & \textbf{To Falsify} & \textbf{Na\"{\i}ve} & \textbf{Verifier} \\\midrule\midrule
    \multirow{4}{*}{\textbf{Stack}}               & \multirow{2}{*}{LinkedList}          & \multirow{2}{\propwidth}{Elements are stored at unique locations.} & Overwrite an existing node when pushing.                             & 0.51  & \texttimes{}4.9  & \texttimes{}3.2       \\
                                                  &                                      &  & Make the linked list circular during concatenation.                  & 0.25  & O/M              & \texttimes{}13.2      \\\cmidrule(lr){2-3}
                                                  & \multirow{2}{*}{KVStore}             & \multirow{2}{\propwidth}{Elements are linked linearly.} & Push the new element in the middle of the stack.                     & 1.11  & T/O              & \texttimes{}4.6       \\
                                                  &                                      &  & Concatenate elements to the middle of a stack.                       & 0.94  & O/M              & \texttimes{}6.5       \\\midrule
    \multirow{2}{*}{\textbf{Queue}}               & LinkedList                           & Elements are stored at unique locations. & Overwrite an existing node when enqueueing.                          & 0.73  & \texttimes{}2.5  & \texttimes{}2.7       \\\cmidrule(lr){2-3}
                                                  & Graph                                & Degrees of vertices are at most one. & Overwrite an existing vertex when enqueueing.                       & 1.75  & T/O              & \texttimes{}7.4       \\\midrule
    \multirow{2}{*}{\textbf{Set}}                 & KVStore                              & Each key is associated with a distinct value. & Put a duplicated element.                                            & 0.87  & T/O              & \texttimes{}1.4       \\\cmidrule(lr){2-3}
                                                  & Tree                                 & The underlying tree is a binary search tree. & Insert a smaller element to the right subtree.                       & 1.10  & \texttimes{}40.7 & \texttimes{}11.1      \\\midrule
    \multirow{2}{*}{\textbf{Heap}}                & LinkedList                           & Elements are stored at unique locations, sorted. & Insert after a node with a larger value.                             & 0.11  & \texttimes{}12.9 & \texttimes{}13.2      \\\cmidrule(lr){2-3}
                                                  & Tree                                 & Parents are smaller than their children. & Insert a smaller element to the right subtree.                       & 1.00  & \texttimes{}2.4  & \texttimes{}2.5       \\\midrule
    \multirow{4}{*}{\textbf{Min Set}}             & \multirow{2}{*}{Set}                 & \multirow{2}{\propwidth}{The cached element has been inserted and is no larger than other elements.} & Record the minimum without inserting it.                             & 1.14  & \texttimes{}1.3  & \texttimes{}1.3       \\
                                                  &                                      &  & Insert a new minimum without recording it.                           & 1.32  & \texttimes{}9.0  & \texttimes{}9.9       \\\cmidrule(lr){2-3}
                                                  & \multirow{2}{*}{KVStore}             & \multirow{2}{\propwidth}{The cached element has been put and is no larger than other elements.} & Record the minimum without putting it.                               & 0.66  & T/O              & \texttimes{}4.3       \\
                                                  &                                      &  & Overwrite an existing element when putting.                          & 1.95  & \texttimes{}10.7 & \texttimes{}14.9      \\\midrule
    \multirow[b]{3}{*}{\textbf{Lazy Set}}            & Tree                                 & The underlying tree is a binary search tree. & Insert a smaller element to the right subtree.                       & 1.09  & \texttimes{}4.8  & \texttimes{}11.5      \\\cmidrule(lr){2-3}
                                                  & Set                           & The same element is never inserted twice. & Insert a duplicated element.                                         & 0.49  & \texttimes{}1.2  & \texttimes{}1.3       \\\cmidrule(lr){2-3}
                                                  & KVStore                              & Each key is associated with a distinct value. & Put a duplicated element.                                            & 0.88  & \texttimes{}49.8 & \texttimes{}1.5       \\\midrule
    \multirow{4}{*}{\textbf{DFA}}                 & \multirow{2}{*}{KVStore}             & \multirow{2}{\propwidth}{Each state is associated with a non-empty list of next states via unique labels.} & Put an overlapping transtion with the same label.                    & 0.66  & \texttimes{}29.9 & \texttimes{}29.9      \\
                                                  &                                      &  & A transition is reversed instead of deleted.        & 1.04  & \texttimes{}15.0 & \texttimes{}15.2      \\\cmidrule(lr){2-3}
                                                  & \multirow{2}{*}{Graph}               & \multirow{2}{\propwidth}{The outgoing edges of each state are labeled by different characters.} & Connect two connected nodes with the same label.             & 0.98  & \texttimes{}12.9 & \texttimes{}12.8      \\
                                                  &                                      &  & Connect two nodes instead of disconnecting them.                     & 0.97  & \texttimes{}16.5 & \texttimes{}16.5      \\\midrule
    \multirow[c]{6}{1.2cm}{\textbf{Connected Graph}} & \multirow{3}{*}{LinkedList}          & \multirow{3}{\propwidth}{Edges (pairs of vertices) are uniquely stored with connected vertices being valid.} & Insert a vertex pair twice during initialization.                   & 0.27  & T/O              & \texttimes{}16.4      \\
                                                  &                                      &  & Insert a vertex without ensuring its connectivity.       & 1.31  & O/M              & \texttimes{}9.9       \\
                                                  &                                      &  & Insert a duplicated vertex pair.                                    & 1.44  & O/M              & \texttimes{}11.2      \\\cmidrule(lr){2-3}
                                                  & \multirow{3}{*}{Graph}               & \multirow{3}{\propwidth}{All vertices are connected in the graph.} & Create a duplicated edge during initialization.                      & 1.13  & \texttimes{}1.7  & \texttimes{}1.8       \\
                                                  &                                      &  & Create a vertex without ensuring its connectivity.       & 2.05  & \texttimes{}6.0  & \texttimes{}6.1       \\
                                                  &                                      &  & Disconnect a vertex from the rest of the graph.                     & 2.38  & \texttimes{}20.0 & \texttimes{}16.2      \\\midrule
    \multirow[b]{2}{1.2cm}{\textbf{Colored Graph}}   & Graph                                & Vertices are colored before being connected to vertices with different colors. & Create an edge between two vertices colored the same.             & 3.68  & T/O              & T/O                   \\\cmidrule(lr){2-3}
                                                  & KVStore                              & Each vertex is associated with a list of vertices with different colors. & Put an edge between two vertices with the same color.                & 7.82  & T/O              & T/O                   \\\midrule
    \textbf{Linked List}                          & KVStore                              & Each node has at most one predecessor. & Put a new predecessor to a node before deleting its old predecessor. & 7.03  & O/M              & T/O                   \\\bottomrule
  \end{tabular}                                                                                                                                                                                           
\end{table}                                                                                                                                                                                               

In our evaluation, we consider the following research questions:
\begin{enumerate*}[label=\textbf{Q\arabic*.}, leftmargin=*]
\item Can \HATch's behavioral specifications effectively capture interesting
  safety properties?
\item Can \HATch's use of symbolic derivatives improve trace exploration for
  falsification? %
\item Can \HATch enhance assurance through falsification when
  verification is challenging?
\end{enumerate*}

We evaluate \HATch on stateful variants of functional ADTs (see
\cref{tab:eval}) drawn from different
sources~\cite{okasakiPurelyFunctionalData1999,
  miltnerDatadrivenInferenceRepresentation2020a,zhouHATTrickAutomatically2024}.
The ADTs we consider are implemented using different effectful representation types
(i.e., \textbf{Repr. Type} column), including key-value stores, linked lists,
sets, trees, and graphs.  We introduce artificial bugs in their
methods as summarized in the Violation column, and evaluate \HATch's
capability to falsify these buggy implementations.  The next column
reports the time \HATch takes to falsify the violation.

To demonstrate the effectiveness of symbolic derivatives, we implement
a variant of \HATch following the description given in
\cref{sec:naive}, and report \HATch's speedup over this variant.
The satisfiability of a path condition $\Phi$ and a SRE $\FA$
(\cref{defn:naive-falsification}) is checked by first replacing
logical formulae with the elements from a finite equivalence class.
An $\FA$ then becomes an ordinary regular expression amenable to SMT
solving, whose non-emptiness, along with the satisfiability of
the logical formulae, witness its satisfiability.  Our results
show that without using derivatives, symbolic execution is unable to falsify
\begin{enumerate*}
\item 7 violations (out of a possible 20) within 60 seconds, resulting in timeouts (T/O) due to
  excessive calls to the SMT solver, and
\item 5 violations under an 8 GB memory limit, leading to
  out-of-memory errors (O/M) due to the complexity involved in
  constructing equivalence classes.
\end{enumerate*}

To demonstrate the effectiveness of \HATch against a verification
procedure, we compare its performance with recent work on
representation invariant verification
\cite{zhouHATTrickAutomatically2024}, and report its speedup over that
verifier in terms of the time taken to identify a violation.  Overall,
\HATch demonstrates significant improvement in performance, measured
in orders of magnitude, compared to both the non-derivative aware
engine and the verifier.  It is noteworthy that it is able to
efficiently handle two challenging ADTs, colored graphs and linked
lists, falsifying their buggy implementation in a small (< 8) number of
seconds, whereas the other approaches are unable to provide any result
within the given resource bound (60 seconds, 8 GB).

\section{Related Work}\label{sec:related}

\paragraph{Symbolic Execution for Functional Languages}
While symbolic execution has been typically used in the context of imperative
languages for bug finding~\cite{baldoniSurveySymbolicExecution2018},
there have been recent efforts that
apply SE in a functional programming
setting.
\cite{xuStaticContractChecking2009a} and \cite{nguyenSoftContractVerification2014} use SE to
verify contracts in Haskell and pure Racket, respectively,
with \cite{nguyenSoftContractVerification2017} extending contract verification to handle Racket programs with mutable state.
SE has also been used for underapproximate reasoning to identify weak library specifications that lead to
type-checking failures of client programs in LiquidHaskell
\cite{hallahanLazyCounterfactualSymbolic2019}.
Our goals in this paper are substantially different, focused on falsifying safety properties of
functional ADTs that interact with opaque and effectful libraries.

\paragraph{Temporal Verification}
Model checking has been applied for software verification against
temporal specifications, e.g., LTL and CTL.  Early work shows how
transition systems can be extracted from programs to abstract their
behavior in a form amenable for automata-based inclusion checking to
validate temporal
specifications~\cite{clarkeModelCheckingAbstraction1994}.  More
recently, type and effect systems have been proposed to infer a
conservative overapproximation of effects produced during execution of
higher-order functional
programs~\cite{skalkaHistoryEffectsVerification2004}.  The granularity
of effects inferred has been improved by regarding past effects as a
handle for reasoning about hidden states
~\cite{nanjoFixpointLogicDependent2018,songAutomatedTemporalVerification2022,sekiyamaTemporalVerificationAnswerEffect2023}.
The use of SFAs as a basis for specification and verification has also
been explored in~\cite{zhouHATTrickAutomatically2024} that introduces
Hoare Automata Types (HATs) as a new refinement type abstraction for
verifying programs against effectful trace-based temporal
specifications.  In contrast to these efforts, \HATch considers this
style of specification in the context of underapproximate reasoning,
exploiting the structure of SFAs to enable efficient falsification.

\paragraph{Derivatives of Regular Expressions}
The classic notion of derivatives of regular expressions provides a
lazy and algebraic approach for constructing automaton-based
recognizers from given regular expressions, effectively relating
automaton states to their regular expression counterparts.
Brzozowski's derivative
\cite{brzozowskiDerivativesRegularExpressions1964} initially
introduced this concept for constructing deterministic finite
automata, followed by Antimirov's partial derivative
\cite{antimirovPartialDerivativesRegular1995} for nondeterministic
finite automata, later extended to handle complement and intersection
operations \cite{caronPartialDerivativesExtended2011}.  While it is
known that the classic derivative approach either overapproximates or
underapproximates with predicates in regular expressions, computing
``next literals'' has been proposed as a remedy
\cite{keilSymbolicSolvingExtended2014}.  Our formulation of symbolic
derivatives, while largely inspired by this work, accounts for
universally quantified variables in regular expressions, which are
ubiquitous in program analysis tasks.  However, the ``next literal''
approach can generate an exponential number of transitions in worst
cases.  Recent work on \emph{transition regexes}
\cite{stanfordSymbolicBooleanDerivatives2021} introduces a novel form
of symbolic derivatives that embeds potentially exponential choices
within nested conditionals, enabling lazy exploration of transitions
and algebraic simplification.  Incorporating these symbolic
derivatives into our symbolic execution engine thus may benefit the
reasoning of specifications with richer control structures, presenting
a promising avenue for future research.

\paragraph{Dynamic Trace-Based Reasoning}
Traces as a form of (in)correctness specification have been widely
adopted by dynamic analysis techniques.  Various runtime monitoring
systems rely on a language of traces
\cite{avgustinovMakingTraceMonitors2007, chenMopEfficientGeneric2007,
  goldsmithRelationalQueriesProgram2005,
  havelundMonitoringJavaPrograms2001,
  meredithEfficientMonitoringParametric2008}.  Regular properties over
traces are also used to guide path exploration in dynamic symbolic
execution~\cite{zhangRegularPropertyGuided2015}.  Arbitrary trace predicates
are now supported in Racket contracts~\cite{moyTraceContracts2023}.
We leave for future work the exploration of non-regular trace languages amenable for derivative
computation to enable the falsification of even richer safety
properties.

\section{Conclusions}
\label{sec:conc}

This paper presents a new symbolic execution procedure that integrates
trace-based temporal specifications to reason about ADTs that interact
with effectful libraries.  We demonstrate how to leverage these
specifications, specifically their latent SFA representations, to
manifest the hidden state maintained by an ADT's underlying
representation.  More significantly, we introduce the concept of a
symbolic derivative, a new encoding of symbolic states that relate
admissible specification traces with path exploration decisions, and
show how they enable significant efficiency gains by allowing paths
that are irrelevant to the falsification of a given safety property to
be quickly pruned by a symbolic execution engine.  Our ideas provide
new insight into how trace-guided specifications can enable effective
reachability-based program analyses.

\begin{acks}                            %
  We thank the anonymous POPL reviewers for their detailed comments
  and constructive feedback.  We also thank Guannan Wei for
  stimulating discussions and suggestions on the draft of the paper.
  This material is based upon work supported by the Defense Advanced Research Projects Agency (DARPA) and the Naval Information Warfare Center Pacific (NIWC Pacific) under N6600 1-22-C-4027, STR Research with funding from the US government,
  and the National Science Foundation under CCF-SHF 2321680.
\end{acks}

\section*{Data-Availability Statement}
Our implementation of \HATch, the benchmark suite used, and the
instructions for reproducing results are available at
\citet{yuanArtifactDerivativeGuidedSymbolic2024}.
\bibliographystyle{ACM-Reference-Format}
\bibliography{references}

\end{document}